\documentclass[a4paper,11pt]{article}

\usepackage{subfiles}

\usepackage[english]{babel}
\usepackage[utf8]{inputenc}
\usepackage[T1]{fontenc}
\usepackage{lmodern}

\usepackage{amsmath}
\usepackage{amssymb}

\usepackage{amsthm}
\usepackage{booktabs}
\usepackage{subcaption}
\usepackage{xcolor}
\usepackage{graphicx}
\usepackage{bm}

\usepackage[top=25mm, bottom=25mm, left=25mm, right=25mm]{geometry}

\usepackage[round]{natbib}

\usepackage[ruled]{algorithm2e}
\usepackage{sectsty}
\usepackage{multirow}
\usepackage{titling}

\usepackage[ruled]{algorithm2e}

\usepackage{subfiles} 

\usepackage[mathscr]{euscript}
\newcommand{\RO}{$\mathscr{RO}$}
\newcommand{\PRO}{$\mathscr{PRO}$}
\newcommand{\ARF}{$\mathscr{ARF}$}
\newcommand{\ARO}{$\mathscr{ARO}$}
\newcommand{\PARO}{$\mathscr{PARO}$}

\usepackage{enumitem}
\newlist{CC}{enumerate}{1}
\setlist[CC]{label=Case \arabic*)}
\newlist{STP}{enumerate}{1}
\setlist[STP]{label=Step \arabic*)}

\graphicspath{{img/}{../img/}}

\usepackage{hyperref}
\usepackage{cleveref}

\usepackage{natbib}
 \bibpunct[, ]{(}{)}{,}{a}{}{,}%

\usepackage{thmtools}
\theoremstyle{plain}
\declaretheorem{theorem}

\newlist{lemlist}{enumerate}{1}
\setlist[lemlist]{label=(\alph{lemlisti}), ref=\thelemma\alph{lemlisti},noitemsep}
\Crefname{theorem}{Theorem}{Theorems}
\declaretheorem[sibling=theorem,
name=Lemma,
Refname={Lemma,Lemmas},
numberwithin=section]{lemma}
\Crefname{lemlisti}{Lemma}{Lemmas}
\addtotheorempostheadhook[lemma]{\crefalias{lemlisti}{lemma}}

\declaretheorem{assumption}

\theoremstyle{definition}
\declaretheorem[qed=$\blacksquare$]{definition}

\DeclareMathOperator*{\argmax}{arg\,max}

\newcounter{drafts}

\newcounter{module}
\newcounter{model}
\newcounter{example}
\newcommand \linedabstractkw[2]{
  \renewcommand\maketitlehookd{%
    \mbox{}\medskip\par
    \centering
    \hrule\medskip
    \begin{minipage}{0.9\textwidth}
    #1\\

    \textit{Keywords: }#2
    \end{minipage}\medskip\hrule\medskip
    }
}

\title{
Adjustable robust treatment-length optimization in radiation therapy
}
\author{S.C.M. ten Eikelder\thanks{Department of Econometrics and Operations Research, Tilburg University, The Netherlands} \and A. Ajdari\thanks{Department of Radiation Oncology, Massachusetts General Hospital and Harvard Medical School, USA} \and T. Bortfeld\protect\footnotemark[2] \and D. den Hertog\thanks{Department of Operations Management, University of Amsterdam, The Netherlands}}

\newcommand{\dobib}{ 
    \bibliographystyle{apalike}
    \bibliography{../ReferenceList} 
}
\begin{document}
\renewcommand{\dobib}{}

\linedabstractkw{Traditionally, optimization of radiation therapy (RT) treatment plans has been done before the initiation of RT course, using population-wide estimates for patients' response to therapy. However, recent technological advancements have enabled monitoring individual patient response during the RT course, in the form of biomarkers. Although biomarker data remains subject to substantial uncertainties, information extracted from this data may allow the RT plan to be adapted in a biologically informative way. We present a mathematical framework that optimally adapts the treatment-length of an RT plan based on the acquired mid-treatment biomarker information, while accounting for the inexact nature of this information. We formulate the \emph{adaptive treatment-length optimization problem} as a 2-stage problem, wherein the information about the model parameters gathered during the first stage influences the decisions in the second stage. Using Adjustable Robust Optimization (ARO) techniques we derive explicit optimal decision rules for the stage-2 decisions and solve the optimization problem. The problem allows for multiple worst-case optimal solutions. To discriminate between these, we introduce the concept of Pareto Adjustable Robustly Optimal (\PARO{}) solutions. 

\hspace*{4ex} In numerical experiments using lung cancer patient data, the ARO method is benchmarked against several other static and adaptive methods. In the case of exact biomarker information, there is sufficient space to adapt, and numerical results show that taking into account both robustness and adaptability is not necessary. In the case of inexact biomarker information, accounting for adaptability and inexactness of biomarker information is particularly beneficial when robustness (w.r.t. organ-at-risk (OAR) constraint violations) is of high importance. If minor OAR violations are allowed, a nominal folding horizon approach (NOM-FH) is a good performing alternative, which can outperform ARO. Both the difference in performance and the magnitude of OAR violations of NOM-FH are highly influenced by the biomarker information quality.
}{radiation therapy, adjustable robust optimization, decision rules, optimal stopping}

\maketitle

\section{Introduction}\label{sec: introduction}
In radiation therapy (RT), the goal is to deliver a curative amount of radiation dose to the target volume(s), while keeping the dose to all organs-at-risk (OARs) within tolerable limits. As the radiation beam delivers energy to all tissues that are on its path, the OARs will (often) inevitably receive some dose as well. Treatment planning has a spatial component, where the optimal dose distribution is determined, and a temporal component, where the optimal number of treatment sessions, or fractions, is determined. Whereas the former is predominantly a geometric problem, the latter involves radiobiological effects. 

Technological advances in treatment monitoring through imaging and other forms of data acquisition allow for a more accurate assessment of an individual's radiation response \citep{Baumann16}. Biologically-based adaptive treatments aim to monitor the treatment, acquire mid-treatment biomarker information, and adapt the remainder of the treatment course accordingly. Many approaches to adaptive treatment planning have been studied in the literature. To the best of our knowledge, all existing approaches assume that all information acquired mid-treatment is exact, i.e., gives a perfect representation of the current state of treatment response. Unfortunately, the limited availability and accuracy of required biomarkers pose a primary challenge for adaptive treatments \citep{Baumann16}. Any information from biomarker data acquired during treatment remains subject to uncertainties, stemming from both measurement errors and the inexactness in the translation of measured data to usable model parameters. Therefore, it is crucial that any adaptive treatment planning approach takes this into account. \citet{Ajdari19} provide an overview of relevant mathematical (optimization) tools. We present an approach to optimally adapt the treatment length of RT using inexact mid-treatment biomarker information.

Specifically, we take an adjustable robust optimization \citep[ARO,][]{BenTal04, Yanikoglu19} approach that accounts for the inexactness of biomarker information. ARO is an extension of robust optimization (RO) that takes into account the flow of information over time and exploits the fact that some decisions need to be taken only after the data has (partially) revealed itself. By using ARO, we ensure that the treatment plan delivered in the initial treatment stage (prior to obtaining biomarker information) is `adaptation aware'. That is, it is designed with adaptation in mind, which may yield more flexibility at the time of treatment adaptation. In the standard paradigm, ARO assumes that the revealed information is exact; we employ the ARO methodology developed in \citet{DeRuiter17} for the case when revealed information is not exact, but provides only an estimate of the true parameters.

\subsection*{Contributions}
We consider a stylized two-stage ARO model to optimally adapt the treatment-length based on inexact biomarker information acquired mid-treatment. Although the stylized model makes several simplifying assumptions to aid the analysis, we believe it captures several important characteristics of realistic adaptive treatment planning, and it enables a precise analysis of the influence of uncertainty in biomarker information. Our main contributions are:
\begin{itemize}
\item We develop mathematical tools based on ARO that enable us to (i) optimally adapt the dose per fraction and treatment length after acquiring mid-treatment biomarker information, (ii) analyze the influence of biomarker information uncertainty.
\item We present explicit optimal decision rules for a difficult (non-convex, mixed-integer) yet practically relevant ARO problem.
\item We show that there are multiple optimal solutions for the worst-case scenario, and that these differ in performance in non-worst-case scenarios. To handle this, we introduce the concept of Pareto Adjustable Robustly Optimal (\PARO{}) solutions, a generalization of Pareto Robustly Optimal (\PRO{}) solutions \citep{Iancu14} to two-stage robust optimization problems. In case the acquired biomarker information is exact, \PARO{} solutions are obtained.
\item We perform a computational study using real lung cancer patient data to determine the optimal timing of acquiring biomarker information in case biomarker quality improves over time. Later biomarker acquisition also limits adaptation possibilities, and the optimal balance depends on the improvement rate.
\end{itemize}

\subsection*{Literature review}
There is a large body of adaptive treatment planning research in RT, the majority of which focusses not on biologically-based uncertainties but on geometric uncertainties and anatomical changes. \citet{Chan13}, \citet{Mar15}, \citet{Bock17}, \citet{Bock19} and \citet{Lim20} present adaptive treatment planning approaches that starts with delivery of the original treatment plan, often derived using RO methods. At given adaptation moments, the `state' of the patient (e.g., anatomical changes, tumor shrinkage, breathing motion pattern) is observed, and the treatment is re-optimized for the remainder of the treatment plan. In RO terminology, these approaches are known as folding horizon (FH) methods. They disregard adaptation possibilities initially, and re-optimize the updated model once mid-treatment information is acquired.

Several treatment plan adaptation approaches have been proposed for adapting to biological information, differing in considered biomarker information, adapted treatment plan decisions and used methodology. \citet{Ghate11} and \citet{Kim12} propose a theoretical stochastic control framework to optimally adapt the dose distribution over a fixed number of fractions, based on hypothetically-observed tumor states. \citet{Saberian16b} concretize this theoretical framework, using simulated hypoxia (insufficient oxygen supply at cellular level) status as biomarker. \citet{Long15} consider a model with a constraint on the probability of radiation-induced lung toxicity, which depends on an a priori unknown model parameter. The problem is formulated as a two-stage model (before and after parameter observation), and the optimal dose distribution is determined for each stage. They consider a finite scenario set for the parameter, and the lung toxicity constraint either has to hold in expectation or has to hold for all considered scenarios.

\citet{Nohadani17} consider a two-stage model to adapt to hypoxia information, where the uncertainty is time-dependent. As the hypoxia information ages the uncertainty grows, until it resets at the observation moment. In each stage the dose distribution is optimized; for the second stage a finite adaptability approach is taken. It is shown that total information degradation is minimized if the observation moment is set mid-treatment. \citet{Dabadghao20} consider a similar time-dependent uncertainty set, for adapting to hypoxia information in a multi-stage setting. In each stage the mean tumor dose per fraction is optimized using a folding horizon approach. It is shown that total information degradation is minimized if the observation moments are set equidistant. They introduce a cost of observation (additional dose due to mid-treatment positron emission tomography (PET) scans), and determine the optimal number of observations. Both papers assume that the hypoxia state varies over time, and is exactly observed at the observation moment(s). In contrast, we assume that uncertain parameters are constant in time, and consider inexact biomarker information. Moreover, they solely adapt the dose, whereas we also adapt the total number of treatment fractions itself.

Adapting the treatment length (i.e., the fractionation schedule) based on observed radiation response has been studied before in the literature. \citet{Saka14} consider a two-stage model where after a fixed number of treatment fractions hypoxia information is acquired. Based on this information both the remaining number of treatment fractions and the dose distribution are re-optimized, in order to maximize average hypoxia-corrected tumor dose. They focus on maintaining hypoxia-corrected fraction size requirements. A similar approach is taken by \citet{Ajdari18}, where after each treatment fraction the tumor cell density in each voxel is observed, and adaptations are made after each treatment fractions instead of only once. The objective is to minimize the total number of tumor cells remaining (TNTCR) at the end of the treatment course. Both approaches can be considered FH methods. In contrast to our approach, it is assumed that any information acquired mid-treatment is exact.

\citet{Iancu20} propose a conceptual robust monitoring and stopping model. They consider a system with a state $x(t)$, and after each observation moment the uncertainty in the system state $x(t)$ grows as $t$ increases. At a new observation moment, the uncertainty reduces to zero, i.e., an exact observation is made. They consider multiple observation moments, and the goal is to time these optimally. At each observation moment, the (state-dependent) direct stopping reward is compared to the worst-case continuation reward, and the according action is taken. Their model does not allow for controls that influence state variables, i.e., applying their model to RT optimization problems would not allow to adjust the dose distribution or the mean dose per fraction.

\subsection*{Notation and organization}
All variables and constants are $1$-dimensional (belong to $\mathbb{R}$ or $\mathbb{N}$) unless indicated otherwise. In functions, a semicolon (;) is used to separate variables and constant arguments from uncertain parameters. Optimal solutions to optimization problems are indicated with an asterisk ($^{\ast}$). Properties of optimal solutions to optimization problems have calligraphic font (e.g., \ARO{}) to distinguish them from methods with the same or similar abbreviations.

The remainder of this paper is organized as follows. \Cref{sec: problem-formulation} introduces the used biological models, background information on biomarkers and states modeling choices. \Cref{sec: exact} introduces the adjustable treatment-length optimization problem under the assumption of exact biomarker information and solves this using ARO techniques. \Cref{sec: inexact} generalizes this to inexact biomarker information. \Cref{sec: results} presents and discusses results of numerical experiments on a lung cancer data set. Finally, \Cref{sec: conclusion} concludes the paper.

\section{Adaptive fractionation} \label{sec: problem-formulation}
\subsection{The fractionation problem} \label{sec: fractionation}
Spatial optimization exploits the fact that, by mounting the beam head on a gantry, the tumor can be targeted from various angles. It aims to find the combination of beam angles and weights that gives the best trade-off between tumor dose conformity and healthy tissue sparing. There is a large body of literature on this topic, see for example \citet{Shepard99,Ehrgott08} for reviews. The result of the spatial optimization problem is a combination of beam angles and weights that gives the best trade-off between tumor dose conformity and healthy tissue sparing. The resulting dose distribution gives the dose to each voxel (3-dimensional subvolume) of the tumor and OARs. \Cref{fig: dose-distribution} gives an example of a slice of a dose distribution. The beam angles and weights are chosen such that the target (contoured in black) receives a high dose and a nearby organ-at-risk (OAR, contoured in red) receives a low dose. 
\begin{figure}[htb!]
\centering
\includegraphics[scale=0.5]{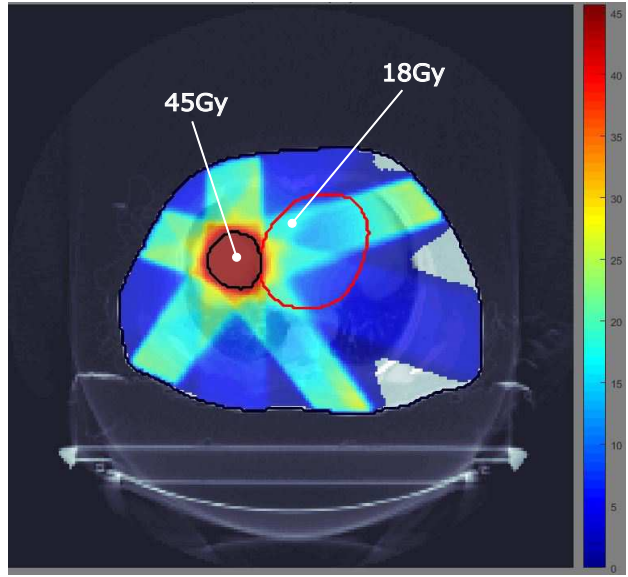} 
\caption{\small Slice of a 3D dose distribution (illustration). Target is contoured in black, OAR is contoured in red. The target volume has a mean dose of 45Gy. The indicated OAR voxel with dose 18Gy has a dose sparing factor of 0.4. \label{fig: dose-distribution}}
\end{figure}

Typically, this dose is not delivered in a single treatment session, but spread out over multiple treatment fractions (fx). The underlying idea is that compared to tumor cells, healthy tissues often have better repair capabilities between fractions \citep{Fowler89,Withers85}. On the other hand, a treatment spread-out over a large number of treatment fractions may not deliver sufficient damage to the target volume, and increases the risk of tumor proliferation. The effect of fractionation differs per healthy tissue type and per tumor site, see, e.g., \citet{Hall12} for further details. Determining the optimal number of treatment fractions is known as the fractionation problem. Treatments with a higher number of fractions and a lower dose per fraction than the conventional regimen are known as \emph{hyperfractionated} treatments. Treatments with a lower number of fractions and a higher dose per fraction than conventional are known as \emph{hypofractionated} treatments.

In each treatment fraction, a scaled version of the dose distribution is administered. Scaling the dose distribution influences the absolute dose delivered to each voxel, but the relative dose remains unchanged. That is, for each voxel we can define a \emph{dose sparing factor}: the dose of this voxel as a fraction of the mean target dose. In \Cref{fig: dose-distribution}, the target volume receives a uniform dose of 45Gy over the entire treatment course. The indicated OAR voxel receives 18Gy and thus has a dose sparing factor of $0.4$. Consequently, if in an individual treatment fraction we administer a mean target dose of 3Gy, the OAR voxel receives 1.2Gy. Thus, with a fixed dose distribution, the sole decision in treatment fraction $t$ is $d_t$, the mean tumor dose in that fraction. The corresponding dose to any voxel $i$ with dose sparing factor $s_i$ is $s_id_t$. Altogether, with a fixed dose distribution fractionation problems can be described using a low number of decision variables. Typically, the target dose is homogeneous, so for simplicity we assume that each target voxel receives dose $d_t$ (i.e., has dose sparing factor 1). Nevertheless, we emphasize that this modeling approach allows for both heterogeneous dose distributions in target and OAR volumes.

The BED model \citep{Fowler89,Fowler10,Hall12} states that the biological effect of an $N$-fraction dose sequence $d = (d_1,\dotsc,d_N)$ (in Gray (Gy)) delivered to a tumor volume is given by
\begin{align} \label{eq: BED-T}
\text{BED}_T = \sum_{t=1}^N d_t + \frac{1}{\alpha/\beta} \sum_{t=1}^N d_t^2,
\end{align}
which is a model governed by a single parameter, the $\alpha/\beta$ ratio, which signifies the fractionation sensitivity of the tumor tissue. The BED to the OAR can be described by
\begin{align} \label{eq: BED-OAR}
\text{BED}_{\text{OAR}} = \sum_{t=1}^N \sigma d_t + \frac{1}{\alpha/\beta} \sum_{t=1}^N \sigma d_t^2,
\end{align}
where $\sigma$ is the generalized dose-sparing factor. For the maximum BED in the OAR volume, $\sigma$ is the maximum of the individual dose sparing factors $s_i$. In order to describe a mean BED constraint or dose-volume BED constraint other choices for $\sigma$ can be used \citep{Saberian16a,Perko18}.

For notational convenience, let $\tau$ (for tumor) and $\rho$ (for risk) denote the inverse $\alpha/\beta$ ratio of the tumor and OAR volume, respectively. \citet{Mizuta12} consider the problem of minimizing OAR BED subject to a lower bound $\text{BED}_{\text{T}}^{\text{pres}}$ on tumor BED. The number of fractions $N$ is restricted to be at most $N^{\text{max}}$. The problem reads
\begin{subequations} \label{eq: mizuta-problem}
\begin{align}
\min_{d,N}  ~&~ \sigma \sum_{t=1}^N d_t + \rho \sigma^2 \sum_{t=1}^N d_t^2 \\
\text{s.t.} ~&~ \sum_{t=1}^N d_t + \tau \sum_{t=1}^N d_t^2 \geq \text{BED}_{\text{T}}^{\text{pres}} \label{eq: mizuta-problem2} \\
            ~&~ d_1,\dotsc,d_N \geq 0, N \in \{1,\dotsc,N^{\text{max}}\}.
\end{align}
\end{subequations}

Let $(d^{\ast},N^{\ast})$ denote the optimal solution to \eqref{eq: mizuta-problem}. A simple analysis in \citet{Mizuta12} reveals the following important result:
\begin{align} \label{eq: mizuta-rule}
N^{\ast} =
\begin{cases}
1  & \text{if } \tau \geq \sigma \rho  \\
N^{\text{max}} \text{ and } d_1^{\ast} = \dotsc = d_{N^{\text{max}}}^{\ast}  & \text{otherwise.}
\end{cases}
\end{align}
In both cases the optimal dose $d^{\ast}$ is such that \eqref{eq: mizuta-problem2} is active. It can be shown that the same result holds if we maximize tumor BED subject to an upper bound on OAR BED \citep{Bortfeld15}. A similar result has been derived for the case with multiple OARs \citep{Saberian16a}. There is a large body of research that optimizes the number of treatment fractions for different model formulations (see \citet{Saberian17} and references therein).

In the current paper, we restrict to one dose-limiting OAR. For many tumor sites, there is a single OAR that restricts the doses that can be delivered, and other OARs are much less restrictive. For example, for lung cancer the mean lung dose is an important indicator of toxicity. On the other hand, for head and neck cancer many OARs must be accounted for. We emphasize that other OARs are not completely disregarded. They are taken into account implicitly, because the original dose distribution was planned with all relevant OARs taken into consideration. Moreover, by restricting the minimum and maximum (mean target) dose per fraction, extreme deviations from the standard fractionation schedule are avoided, which is also designed whilst taking multiple OARs into account.

\subsection{Adaptive fractionation using biomarkers} \label{sec: biomarkers}
Most fractionation optimization methods assume the tumor and OAR fractionation sensitivity parameters $\tau$ and $\rho$ are known exactly. There is much research on the $\alpha/\beta$ ratios for different tumor sites \citep{vanLeeuwen18} and OAR sites \citep{Kehwar05}, but they remain subject to considerable uncertainties. We assume box uncertainty of the form:
\begin{align} \label{eq: Z}
Z := \big\{(\rho,\tau): \rho_L \leq \rho \leq \rho_U, \tau_L \leq \tau \leq \tau_U \big\},
\end{align}
with $0<\rho_L<\rho_U$ and $0 < \tau_L < \tau_U$. It is assumed that there is a nominal scenario, e.g., parameter values $\bar{\tau}$ and $\bar{\rho}$ derived from literature. There are two reasons for assuming a box uncertainty set. First, to the best of our knowledge there is little evidence that there is any correlation between the $\alpha/\beta$ ratios of target volumes and normal tissues. Second, box uncertainty leads to simpler models, which allow for a more detailed analysis of optimal fractionation decisions.

\citet{Ajdari17b} also consider a box uncertainty set, and determine a \emph{robustly optimal} fractionation scheme, i.e., one that is that is feasible for all possible realizations and that is optimal for the worst-case realization. If biomarker information acquired during treatment provides more accurate information on fractionation sensitivity than what was available at the start of the treatment, such a static RO approach may be overly conservative.

\citet{Somaiah15} give an overview of various mechanisms for determining fractionation sensitivity. Using blood samples, one can quantify the involvement of  non-homologous end-joining (NHEJ) and homologous recombination (HR) in cells. For details on how to measure these, we refer to \citet{Bindra13} and \citet{Barker10}, respectively. Change in the expression of epidermal growth factor receptor (EGFR) genes can also give some hints regarding the fractionation sensitivity \citep{Somaiah15}, which can also be measured mid-treatment. Lastly, \citet{Somaiah15} mention that there is a close link between proliferation index and hypoxia, both of which can be measured during RT using different PET tracers. We note that there is evidence that some of these mechanisms could be subject to change during RT, depending on, amongst others, the delivered dose, hypoxia, and immune system interaction. However, as a first study to adapting to inexact biomarker information, we make the assumption that fractionation sensitivity is static throughout treatment, i.e., there is a static `true’ $(\rho,\tau)$. In \Cref{sec: exact} we assume to observe (measure) the true $(\rho,\tau)$ exactly, and in \Cref{sec: inexact} we assume to observe only an estimation/approximation $(\hat{\rho},\hat{\tau})$.

The quality of the observed parameter estimates depends amongst others on the suitability of the biomarker itself, the measurement accuracy and when the biomarker measurement is taken during the treatment course. The relationship between the data quality and the moment of biomarker observation is complex, and it is impossible to exactly quantify this. For some biomarkers the data quality may greatly improve in the first few fractions, with a diminishing improvement in later fractions\footnote{This is especially true in the case of certain blood biomarkers of innate immune status (such as interleukin (IL)-6 or tumor necrosis factor (TNF)-$\alpha$) which are also the markers of inflammation. In these biomarkers, as the biomarker acquisition is shifted towards later in the RT course, the information regarding the immune status gets mixed with the RT-induced inflammation and loses its specificity to immune system condition.}. For others (e.g., functional imaging such as PET and magnetic resonance imaging (MRI)), the data quality is poor at the first couple of fractions and only increases substantially in later fractions\footnote{This is because the effect of RT on tissue is cumulative and is morphologically manifested only after a certain amount of dose (which depends on the underlying tissue threshold) is delivered.}. In practice, some biomarkers, e.g., radiographic information, may also exhibit a decreasing data quality for very late observation moments (due to, for instance, interference from acute inflammation in the lung). Such behavior is rare, and as such not considered here. We impose a minimum dose per fraction, and make the assumption that biomarker data quality increases in the number of treatment fractions. In this way, the change in biomarker quality is influenced by both the dose delivered and the time passed. We will investigate several functional forms for this relationship in the numerical experiments.

\subsection{Modeling choices}
In order to establish a meaningful model for the adjustable robust optimization approach, we restrict the dose sequence $d=(d_1,\dotsc,d_{N^{\text{max}}})$ in several ways. In addition to a maximum number of fractions, we also set a minimum $N^{\text{min}}$. Furthermore, we assume there is a single moment $N_1$ where we can adapt the treatment. Under the assumption that the uncertain parameters remain constant over time, more than one observation moment is not useful if the parameter is observed exactly. With inexact observations, there can be value in multiple observations, but given the patient burden (logistically and the delivery of additional dose) and financial cost of PET scans this is not considered here.

The dose per fraction in the first $N_1$ fractions is assumed to be the same, denote this by $d_1$. Variable $N_2$ denotes the number of fractions after adaptation; also these fractions have equal dose, denoted by $d_2$. In current clinical practice, uniform fractionated treatments are the standard. By restricting to only two different dose levels, extreme deviations from standard protocols are prevented. The above implies
\begin{align}
N_2 \in \big\{N_2^{\text{min}},\dotsc, N_2^{\text{max}}\big\},
\end{align}
with $N_2^{\text{min}} = \max\{1,N^{\text{min}}-N_1\}$ and $N_2^{\text{max}} = N^{\text{max}}-N_1 $. We additionally set the constraint that $d_1, d_2 \geq d^{\text{min}}$, for some predetermined value $d^{\text{min}}$. Aside of preventing an unrealistically low dose per fraction, the minimum dose serves a modeling purpose for stage 1. As noted in \Cref{sec: biomarkers}, the biomarker quality can depend on both dose and time. The current model implicitly makes the assumption that an early response can only be observed via biomarkers once $N_1$ fractions of dose at least $d_{\text{min}}$ have been delivered. Thus, in our models, this can be interpreted as a threshold. In the numerical experiments we investigate several temporal relationships between $N_1$ and biomarker quality. Lastly, we set a maximum dose per fraction $d_1^{\text{max}}$ in stage 1, to avoid delivering dosages in stage 1 that severely restrict adaptation possibilities in stage 2. We will later impose some restrictions on the allowed combinations of $d^{\text{min}}$, $d_1^{\text{max}}$ and $N_2^{\text{max}}$. 

\Cref{fig: overview} provides a schematic overview of the situation.
\begin{figure}
\centering
\includegraphics{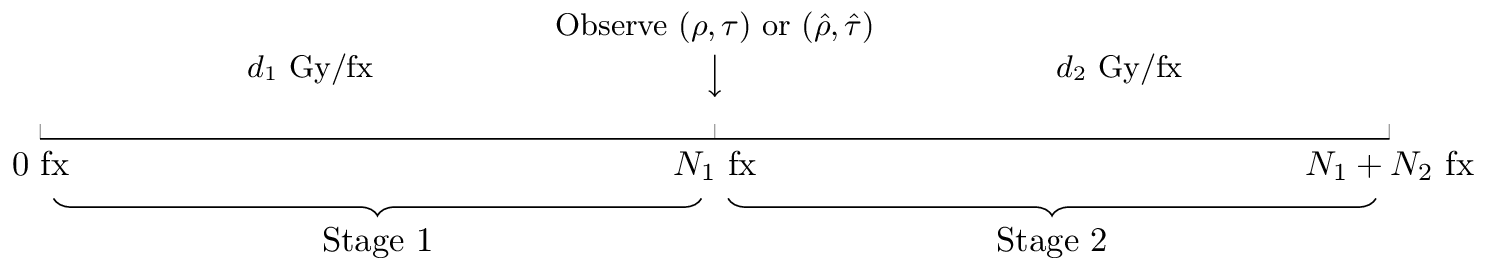}
\caption{\small Schematic overview of the considered model. There are 3 variables: $d_1$, $d_2$ and $N_2$. First, we deliver $N_1$ fractions of dose $d_1$ per fraction. After this, we observe $(\rho,\tau)$ or $(\hat{\rho},\hat{\tau})$. Subsequently, we deliver $N_2$ fractions of dose $d_2$ per fraction. \label{fig: overview}}
\end{figure}
We wish to maximize the tumor BED, subject to the constraint that the BED to the OAR is below the generalized tolerance level $\text{BED}_{\text{tol}}$, given by
\begin{align}
\text{BED}_{\text{tol}}(\rho) = \varphi D \Big(1 + \frac{\varphi D}{T}\rho \Big),
\end{align}
i.e., the OAR is known to tolerate a total dose of $D$ Gy if delivered in $T$ fractions under dose shape factor $\varphi$. The dose shape factor is a parameter characterizing the spatial heterogeneity of a dose distribution, for more details see \citet{Saberian16a} and \citet{Perko18}. Note that $\text{BED}_{\text{tol}}(\rho)$ is a function of uncertain parameter $\rho$, the inverse $\alpha/\beta$ ratio of the OAR.

We emphasize that the model resulting from our modeling choices and assumptions does not directly represent a realistic decision-making problem in radiation therapy treatment planning. Nevertheless, it captures several important aspects of fractionation optimization. Moreover, using ARO on this stylized model, we gain insight into optimal decision rules and the role of uncertainty in adaptive fractionation optimization.

\section{ARO: Biomarkers provide exact information} \label{sec: exact}
We present an adjustable robust optimization approach that optimally adjusts the remainder of the treatment once biomarker information has provided the true value of parameters $\tau$ and $\rho$. This serves as a stepping stone to the inexact data model.

\subsection{Problem formulation} \label{sec: exact-problem-formulation}
After delivering $N_1$ fractions we observe the true $(\rho,\tau)$. At the time that the stage-1 dose per fraction $d_1$ has to be decided, it is only known that $(\rho,\tau)$ belongs to uncertainty set $Z$. Put together, the Exact Data Problem (EDP) reads:
\begin{subequations} \label{eq: exact}
\begin{align}
 \max_{d_1,d_2(\rho,\tau),N_2(\rho,\tau)} ~&~ \min_{(\rho,\tau)\in Z} N_1 d_1 + N_2(\rho,\tau) d_2(\rho,\tau) + \tau (N_1 d_1^2 + N_2(\rho,\tau) d_2(\rho,\tau)^2) \label{eq: exact-1} \\
 \begin{split}
\text{s.t.}~~~~~~ ~&~  \sigma (N_1 d_1 + N_2(\rho,\tau) d_2(\rho,\tau))  + \rho \sigma^2 (N_1 d_1^2 + N_2(\rho,\tau) d_2(\rho,\tau)^2) \\
&  ~~ \leq \text{BED}_{\text{tol}}(\rho),~~ \forall (\rho,\tau) \in Z
 \end{split}  \label{eq: exact-2} \\[0.5em]
            ~&~ N_2(\rho,\tau) \in \{N_2^{\text{min}},\dotsc, N_2^{\text{max}}\},~~\forall (\rho,\tau) \in Z  \label{eq: exact-3}\\
            ~&~ d_2(\rho,\tau) \geq d^{\text{min}}, ~~ \forall (\rho,\tau) \in Z \label{eq: exact-4} \\
            ~&~ d^{\text{min}} \leq d_1 \leq d_1^{\text{max}}. \label{eq: exact-5}
\end{align}
\end{subequations}
The value for the stage-1 dose $d_1$ has to be decided before the value of $(\rho,\tau)$ is revealed; in ARO this is also commonly referred to as a \emph{here-and-now} variable or decision. The values stage-2 dose $d_2$ and stage-2 number of fractions $N_2$ need to be decided only after $(\rho,\tau)$ is revealed as they may depend on the values of these parameters. Hence, they are written as functions $d_2(\rho,\tau)$ and $N_2(\rho,\tau)$ of the uncertain parameters $(\rho,\tau)$. In ARO such variables are also commonly referred to as \emph{wait-and-see} variables or decisions. In this paper, we will adhere to the terms \emph{stage 1} and \emph{stage 2}, however.

Before we solve \eqref{eq: exact}, we need some definitions and assumptions. The remaining BED tolerance level of the OAR, if $N'$ fractions with dose $d'$ have been administered, is given by
\begin{align}\label{eq: B}
B(d',N';\rho) = \text{BED}_{\text{tol}}(\rho) - \sigma d' N' -  \rho \sigma^2 (d')^2N'.
\end{align}
Subsequently, define the function
\begin{align}\label{eq: g}
g(d',N',N''; \rho) := \frac{-1 + \sqrt{1 + \frac{4\rho}{N''}B(d',N';\rho)}}{2\sigma \rho}.
\end{align}
The value of $g$ can be interpreted as the maximum dose that can be delivered in $N''$ fractions if already $N'$ fractions of dose $d'$ are (scheduled to be) delivered. It is obtained by solving the equality version of \eqref{eq: exact-2} for $d_1$ or $d_2$. Functions of this form will be used frequently throughout the remainder of this paper.

The following assumption on the relation between $d^{\text{min}}$, $d_1^{\text{max}}$ and the bounds on $N_2$ makes sure that for a given optimal number of fractions, it is feasible to deliver that number of fractions with minimum dose.
\begin{assumption}\label{ass: exact-dminmax}
It holds that
\begin{align}
d^{\text{min}} \leq d_1^{\text{max}} \leq \min \Big\{g(d^{\text{min}},N_2^{\text{min}},N_1; \rho_L), g(d^{\text{min}},N_2^{\text{max}},N_1; \frac{\tau_L}{\sigma}), g(d^{\text{min}},N_2^{\text{max}},N_1; \rho_U) \Big\}.
\end{align}
\end{assumption}
\noindent The particular form of the upper bound on $d_1^{\text{max}}$ will become clear later. Numerical experiments indicate that results are not very sensitive to the choices of $d^{\text{min}}$ and $d_1^{\text{max}}$. 

We continue by formally defining several properties of solutions. Let $X(\rho,\tau)$ denote the feasible region defined by constraints \eqref{eq: exact-2}-\eqref{eq: exact-5} for fixed $(\rho,\tau)$.
\begin{definition}[Adjustable robustly feasible] \label{def: ARF}
A tuple $(d_1,d_2(\cdot),N_2(\cdot))$ is adjustable robustly feasible (\ARF{}) to \eqref{eq: exact} if $(d_1,d_2(\rho,\tau),N_2(\rho,\tau)) \in X(\rho,\tau)$ for all $(\rho,\tau) \in Z$.
\end{definition}
\begin{definition}[Adjustable robustly optimal] \label{def: ARO}
A tuple $(d_1,d_2(\cdot),N_2(\cdot))$ is adjustable robustly optimal (\ARO{}) to \eqref{eq: exact} if it is \ARF{} and there does not exist an \ARF{} tuple $(\bar{d_1},\bar{d_2}(\cdot),\bar{N_2})$ such that
\begin{align}
\begin{aligned}
\min_{(\rho,\tau)\in Z} N_1 d_1 + N_2(\rho,\tau) d_2(\rho,\tau) + \tau (N_1 d_1^2 + N_2(\rho,\tau) d_2(\rho,\tau)^2) \\ < \min_{(\rho,\tau)\in Z} N_1 \bar{d_1} + \bar{N_2}(\rho,\tau) \bar{d_2}(\rho,\tau) + \tau (N_1 \bar{d_1}^2 + \bar{N_2}(\rho,\tau) \bar{d_2}(\rho,\tau)^2).
\end{aligned}
\end{align}
\end{definition}
\noindent We also define the \ARO{} property for the stage-1 decisions $d_1$ individually.
\begin{definition}[Adjustable robustly optimal $d_1$] \label{def: ARO-x}
A stage-1 decision $d_1$ is \ARO{} to \eqref{eq: exact} if there exist decision rules $d_2(\cdot)$ and $N_2(\cdot)$ such that $(d_1,d_2(\cdot),N_2(\cdot))$ is \ARO{} to \eqref{eq: exact}.
\end{definition}
\noindent Lastly, we define optimality of a decision rule.
\begin{definition}[Optimal decision rule] \label{def: optimal-DR}
For a given $d_1$, a decision rule pair $(d_2(\cdot),N_2(\cdot))$ is optimal to \eqref{eq: exact} if $(d_1,d_2(\cdot),N_2(\cdot))$ is \ARF{} and for any $(\rho,\tau) \in Z$ it holds that
\begin{align}
\begin{aligned}
N_1 d_1 + N_2(\rho,\tau) d_2(\rho,\tau) + \tau (N_1 d_1^2 + N_2(\rho,\tau) d_2(\rho,\tau)^2) \\ \geq N_1 d_1 + \bar{N_2}(\rho,\tau) \bar{d_2}(\rho,\tau) + \tau (N_1 d_1^2 + \bar{N_2}(\rho,\tau) \bar{d_2}(\rho,\tau)^2),
\end{aligned}
\end{align}
for every $(\bar{d_2}(\cdot),\bar{N_2}(\cdot))$ such that $(d_1,\bar{d_2}(\cdot),\bar{N_2}(\cdot))$ is \ARF{}.
\end{definition}

The first observation we make in \eqref{eq: exact} is that any solution $(d_1,d_2(\cdot),N_2(\cdot))$ feasible for scenario $(\rho,\tau)\in Z$ is also feasible for $(\rho,\tau+\epsilon)$ with a higher objective value, if $\epsilon>0$. Therefore, in any worst-case realization it will hold that $\tau = \tau_L$ (see \eqref{eq: Z}). This observation has consequences for what uncertainty sets $Z$ need to be considered. Due to the result \eqref{eq: mizuta-rule}, one can in general distinguish three cases for uncertainty set $Z$ and parameter $\sigma$:
\begin{CC}[leftmargin=*,labelindent=1em]
\item $\sigma \rho_U \leq \tau_L$: According to \eqref{eq: mizuta-rule}, for any realization (with $\tau = \tau_L$) it is optimal to deliver the minimum number of fractions in stage 2.
\item $\sigma \rho_L \geq \tau_L$: According to \eqref{eq: mizuta-rule}, for any realization (with $\tau = \tau_L$) it is optimal to deliver the maximum number of fractions in stage 2. 
\item $\sigma \rho_L < \tau_L < \sigma \rho_U$: In the scenario $(\rho_L,\tau_L)$, it is optimal to deliver the maximum number of fractions in stage 2 according to \eqref{eq: mizuta-rule}. In the scenario $(\rho_U,\tau_L)$, it is optimal to deliver the minimum number of fractions in stage 2 according to \eqref{eq: mizuta-rule}.
\end{CC}
In Cases 1 and 2, \eqref{eq: exact} is easily solved by plugging in the (worst-case) optimal value for $N_2$, and solving the resulting 2-variable optimization problem. Therefore, only Case 3 is of interest and in the remainder of this paper we make the following assumption.
\begin{assumption} \label{ass: omega-rho}
It holds that $\sigma \rho_L < \tau_L < \sigma \rho_U$.
\end{assumption}
In our numerical experiments, we use a lung cancer data set. Recent evidence suggests that for lung cancer \Cref{ass: omega-rho} can indeed hold, i.e., the optimal number of treatment fractions is not always known prior to treatment. Further details are provided in \Cref{sec: setup}. For other tumor sites, such as liver cancer, the tumor $\alpha/\beta$ is generally assumed to be 10 or higher, whereas the $\alpha/\beta$ of normal liver tissue is typically assumed to be $3$ or $4$. Thus, for such tumor sites \Cref{ass: omega-rho} generally does not hold, and hyperfractionation is optimal.

\subsection{Optimal decision rules and worst-case solution} \label{sec: exact-solution}
Problem \eqref{eq: exact} is a 2-stage non-convex mixed-integer ARO problem, which are generally hard to solve. Nevertheless, due to the small number of variables the problem can be solved to optimality. In order to solve \eqref{eq: exact}, we take two steps:
\begin{STP}[leftmargin=*,labelindent=1em]
\item Determine optimal decision rules $d_2(\cdot)$ and $N_2(\cdot)$ for fixed $d_1$.
\item Plug in optimal decision rules and solve for $d_1$.
\end{STP}
In what follows, we give a detailed explanation of both steps. Let $(d_1^{\ast},d_2^{\ast}(\rho,\tau),N_2^{\ast}(\rho,\tau))$ denote an \ARO{} solution to \eqref{eq: exact}.\\

\noindent \noindent \emph{Step 1: Determine optimal decision rules $d_2(\cdot)$ and $N_2(\cdot)$ for fixed $d_1$}\\
\noindent Fix stage-1 variable $d_1$. Similar to the result \eqref{eq: mizuta-rule}, we will show that it is optimal to deliver either the minimum or the maximum number of fractions in stage-2. Moreover, \eqref{eq: exact-2} is the only OAR dose-limiting constraint, so it will hold with equality if this does not violate variable bounds \eqref{eq: exact-4} and \eqref{eq: exact-5}. We will show that the latter is not the case. The theorem below summarizes the result.
\begin{theorem} \label{theorem: exact-stage2} 
Let $d_1$ be the stage-1 decision of \eqref{eq: exact}. The decision rules
\begin{align} \label{eq: exact-N2}
N_2^{\ast}(\rho,\tau) &=
\begin{cases}
N_2^{\text{min}} & \text{ if } \tau \geq \sigma \rho  \\
N_2^{\text{max}} & \text{ otherwise},
\end{cases}
\end{align}
and 
\begin{align} \label{eq: exact-d2}
d_2^{\ast}(d_1;\rho,\tau) &=
\begin{cases}
g(d_1,N_1,N_2^{\text{min}};\rho) & \text{ if } \tau \geq \sigma \rho  \\
g(d_1,N_1,N_2^{\text{max}};\rho) & \text{ otherwise}
\end{cases}
\end{align}
are optimal to \eqref{eq: exact} for the given $d_1$. These provide the unique optimal decisions unless $\tau = \sigma \rho$.
\end{theorem}
\begin{proof}
See \Cref{app: proof-exact-stage2}.
\end{proof}
Clearly, these decision rules are non-linear, and in fact split the uncertainty region in two parts: one where it is optimal to deliver the minimum number of fractions $N_2^{\text{min}}$ in stage 2, and one where it is optimal to deliver the maximum number of fractions $N_2^{\text{max}}$ in stage 2. This suggests splitting the uncertainty set as follows:
\begin{subequations} \label{eq: Zsplit-fixed}
\begin{align}
Z^{\text{min}} &:= Z \cap \{\tau \geq \sigma\rho \}  \label{eq: Zmin}\\
Z^{\text{max}} &:= Z\cap \{ \tau < \sigma\rho \}. \label{eq: Zmax}
\end{align}
\end{subequations}
An illustration is provided in \Cref{fig: Z-box}.\\
\begin{figure}
\centering
\includegraphics{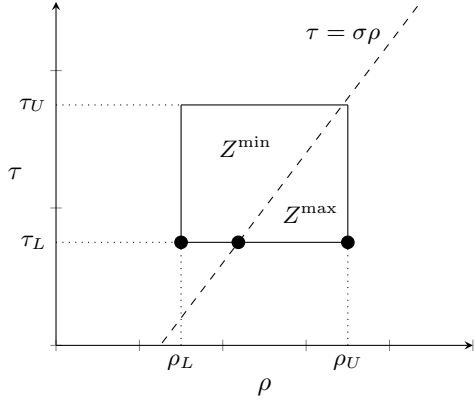}
\caption{ \small Split of uncertainty set $Z$ according to \eqref{eq: Zsplit-fixed}. The circles indicate the locations of the candidate worst-case scenarios for \eqref{eq: exact2}. \label{fig: Z-box}}
\end{figure}

\noindent \noindent \emph{Step 2: Plug in optimal decision rules and solve for $d_1$}\\
In order to find an \ARO{} $d_1$, we introduce the following objective function for given $(\rho,\tau)$:
\begin{align} \label{eq: f}
f(d_1,N_2 ; \rho,\tau) :=%
\begin{cases}
 \!\begin{aligned} &N_1d_1+ N_2 g(d_1,N_1,N_2;\rho) \\
 + &\tau \big(N_1d_1^2+N_2g(d_1,N_1,N_2;\rho)^2\big) \end{aligned} & \text{ if } d_1 \in [0, g(0,0,N_1;\rho)]  \\
  -\infty & \text{ otherwise},
 \end{cases}
\end{align}
where, for given $\rho$, the value $g(0,0,N_1;\rho)$ is the maximum dose that can be delivered in stage 1 due to the nonnegativity restriction on the stage-2 dose. From \Cref{ass: exact-dminmax} it follows that $g(0,0,N_1;\rho) \geq d_1^{\text{max}}$ for all $(\rho,\tau) \in Z$, so $f$ is finite for all feasible $d_1$. According to \Cref{lemma: convexity-f} in \Cref{app: auxiliary-lemmas}, function $f$ is either convex, concave or constant in $d_1$.

Plugging in \eqref{eq: exact-N2} and \eqref{eq: exact-d2} and using definition \eqref{eq: f} allows us to rewrite \eqref{eq: exact} to a problem of only variable $d_1$:
\begin{subequations}\label{eq: exact2}
\begin{align}
\max_{d_1}~&~ \min_{(\rho,\tau) \in Z}~ f(d_1,N_2^{\ast}(\rho,\tau), \rho,\tau) \label{eq: exact2-1} \\
\text{s.t.} ~&~ d^{\text{min}} \leq d_1 \leq d_1^{\text{max}}.
\end{align}
\end{subequations}
As noted in \Cref{sec: exact-problem-formulation}, in any worst-case realization it will hold that $\tau = \tau_L$, so it is sufficient to consider only those observations with $\tau = \tau_L$.

In order to reformulate \eqref{eq: exact2}, we make use of the properties of $g$ and $f$ in \Cref{lemma: fg-omega} in \Cref{app: auxiliary-lemmas}. In particular, \Cref{lemma: f-omega} states that function $f$ is either increasing or decreasing in $\rho$ for fixed $d_1$. Hence, if we move \eqref{eq: exact2-1} to a constraint and split according to \eqref{eq: Zsplit-fixed}, for both $Z^{\text{min}}$ and $Z^{\text{max}}$ it is sufficient to consider the constraint for the highest and lowest value of $\rho$ in the uncertainty set. With $\tau=\tau_L$, this yields the scenarios $(\rho_L,\tau_L)$ and $(\frac{\tau_L}{\sigma},\tau_L)$ for $Z^{\text{min}}$ and $(\frac{\tau_L}{\sigma},\tau_L)$ and $(\rho_U,\tau_L)$ for $Z^{\text{max}}$. 

Therefore, the three candidate worst-case scenarios are $(\rho_L,\tau_L)$, $(\rho_U,\tau_L)$ and $(\frac{\tau_L}{\sigma},\tau_L)$; their locations are indicated in \Cref{fig: Z-box}. By \Cref{lemma: convexity-f}, the objective value in the third scenario is equal to $K$, with
\begin{align} \label{eq: K}
K = \frac{1}{\sigma}B\big(0,0,\frac{\tau_L}{\sigma}\big).
\end{align}
This is the maximum target BED that can be attained if radiation sensitivity parameters are exactly such that fractionation has no influence on the optimal objective value. It is equal to the maximum tolerable OAR BED for these radiation sensitivity parameters, divided by the generalized OAR dose sparing factor $\sigma$. 

Putting everything together, we conclude that if $(d_1^{\ast},d_2^{\ast}(\cdot),N_2^{\ast}(\cdot))$ is \ARO{} to the EDP \eqref{eq: exact} then there exists a $q^{\ast} \in \mathbb{R}_+$ such that $(d_1^{\ast},q^{\ast})$ is an optimal solution to
\begin{subequations}\label{eq: exact3}
\begin{align}
\max_{d_1,q}~&~ q \\
\text{s.t.} ~&~ q \leq f(d_1,N_2^{\text{min}};\rho_L,\tau_L) \label{eq: exact3-2}\\
            ~&~ q \leq f(d_1,N_2^{\text{max}};\rho_U,\tau_L) \label{eq: exact3-3}\\
            ~&~ q  \leq K \label{eq: exact3-4} \\
            ~&~ d^{\text{min}} \leq d_1 \leq d_1^{\text{max}}. \label{eq: exact3-5}
\end{align}
\end{subequations}
Conversely, if $(d_1^{\ast},q^{\ast})$ is an optimal solution to \eqref{eq: exact3} and $N_2^{\ast}(\cdot)$ and $d_2^{\ast}(\cdot)$ are given by \eqref{eq: exact-N2} and \eqref{eq: exact-d2}, respectively, then $(d_1^{\ast},d_2^{\ast}(\cdot),N_2^{\ast}(\cdot))$ is \ARO{} to \eqref{eq: exact}. Hence, \eqref{eq: exact3} and EDP \eqref{eq: exact} are equivalent.

According to \Cref{lemma: convexity-f}, the RHS of \eqref{eq: exact3-2} and \eqref{eq: exact3-3} is convex and concave in $d_1$, respectively. Hence, \eqref{eq: exact3} asks to find the value of $d_1$ that maximizes the minimum of a univariate convex \eqref{eq: exact3-2}, concave \eqref{eq: exact3-3} and constant \eqref{eq: exact3-4} function on a closed interval \eqref{eq: exact3-5}. \Cref{lemma: intersection} in \Cref{app: auxiliary-lemmas} provides information on the intersection points of the functions \eqref{eq: exact3-2}-\eqref{eq: exact3-4}. Consequently, the optimal solution(s) to \eqref{eq: exact3} is/are easily found.

\Cref{fig: exact} illustrates a possible instance of \eqref{eq: exact3}, displaying constraints \eqref{eq: exact3-2}-\eqref{eq: exact3-4}. In this case, the set of optimal solutions is the union of the two intervals for $d_1$ where constraint \eqref{eq: exact3-4} is active. This is indicated in red. Dose constraints \eqref{eq: exact3-5} may cut off part of these intervals. If due to constraint \eqref{eq: exact3-5} both these intervals are infeasible, the optimum is at one of the boundaries for $d_1$.

\begin{figure}
\centering
\includegraphics{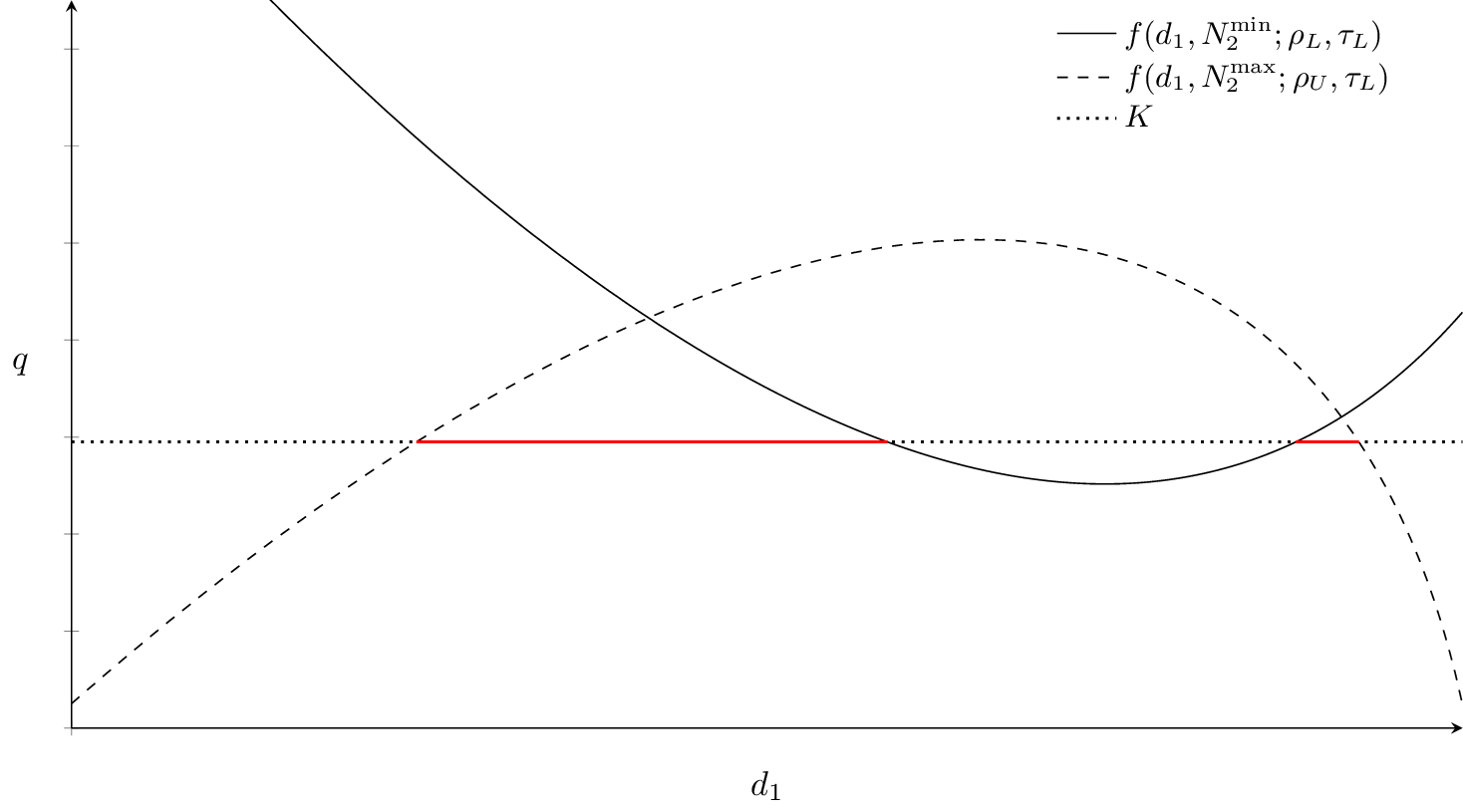}
\caption{\small Schematic illustration of \eqref{eq: exact3}. The solid and dashed curves represent constraints \eqref{eq: exact3-2} and \eqref{eq: exact3-3}, respectively, and the dotted line represent constraint \eqref{eq: exact3-4}. The optimal intervals are indicated in red.\label{fig: exact}}
\end{figure}

\subsection{Pareto adjustable robustly optimal solutions} \label{sec: exact-PARO}
\Cref{fig: exact} illustrates that it is possible that there are multiple optimal solutions to \eqref{eq: exact3}. These solutions are \ARO{} stage-1 solutions to the EDP \eqref{eq: exact}. In general, in case there are multiple \ARO{} solutions these may perform vastly different if a non-worst-case scenario realizes \citep{DeRuiter16}. \citet{Iancu14} studies static robust optimization problems with multiple robustly optimal solutions, and introduce the concept of Pareto robustly optimal (\PRO{}) solutions. A robustly optimal solution is called \PRO{} if there is no other robustly feasible solution that has equal or better objective value for all scenarios in the uncertainty set, while being strictly better for at least one scenario. Non-\PRO{} solutions are dominated by at least one \PRO{} solution and are therefore not desired\footnote{The concept of Pareto robust optimality closely resembles the concept of Pareto efficiency in multi-criteria optimization (MCO). In MCO, Pareto efficient solutions can only be improved in one criteria at the cost of a deterioration in another criteria. Only Pareto efficient solutions are of interest, and the overall goal in MCO is to compute this set of solutions (known as the Pareto surface).}. The concept has previously been applied to breast cancer treatment planning by \citet{Mahmoudzadeh15}.

\citet{Iancu14} study \PRO{} solutions solely for static RO problems; we generalize the concept to 2-stage adjustable robust optimization problems.
\begin{definition}[Pareto adjustable robustly optimal] \label{def: PARO-xy}
An \ARO{} tuple $(d_1,d_2(\cdot),N_2(\cdot))$ is Pareto adjustable robustly optimal (\PARO{}) to \eqref{eq: exact} if there is no tuple $(\bar{d_1},\bar{d_2}(\cdot),\bar{N_2}(\cdot))$ that is \ARO{} to \eqref{eq: exact} and satisfies
\begin{subequations}\label{eq: PARO-conditions}
\small 
\begin{align}
\begin{aligned}
& N_1 d_1 + N_2(\rho,\tau) d_2(\rho,\tau) + \tau (N_1 d_1^2 + N_2(\rho,\tau) d_2(\rho,\tau)^2) \\
& \hspace*{1cm} \leq N_1 \bar{d_1} + \bar{N_2}(\rho,\tau) \bar{d_2}(\rho,\tau) + \tau (N_1 \bar{d_1}^2 + \bar{N_2}(\rho,\tau) \bar{d_2}(\rho,\tau)^2)~~ \forall (\rho, \tau) \in Z 
\end{aligned}\\[1em]
\begin{aligned}
& N_1 d_1 + N_2(\bar{\rho},\bar{\tau}) d_2(\bar{\rho},\bar{\tau}) + \tau (N_1 d_1^2 + N_2(\bar{\rho},\bar{\tau}) d_2(\bar{\rho},\bar{\tau})^2) \\
& \hspace*{1cm} < N_1 \bar{d_1} + \bar{N_2}(\bar{\rho},\bar{\tau}) \bar{d_2}(\bar{\rho},\bar{\tau}) + \bar{\tau} (N_1 \bar{d_1}^2 + \bar{N_2}(\bar{\rho},\bar{\tau}) \bar{d_2}(\bar{\rho},\bar{\tau})^2) ~~\text{for some } (\bar{\rho}, \bar{\tau}) \in Z.
\end{aligned}
\end{align}
\end{subequations}
\end{definition}
\noindent We also define the concept \PARO{} for the stage-1 decision $d_1$ individually.
\begin{definition}[Pareto adjustable robustly optimal $d_1$] \label{def: PARO-x}
A stage-1 decision $d_1$ is \PARO{} to \eqref{eq: exact} if there exist decision rules $N_2(\cdot)$ and $d_2(\cdot)$ such that $(d_1,d_2(\cdot),N_2(\cdot))$ is \PARO{} to \eqref{eq: exact}.
\end{definition}
If there are multiple \ARO{} solutions, we wish to pick one that is \PARO{}. In general, finding PARO solutions is hard, because it requires comparing the performance of both stage-1 decisions and stage-2 decision rules on multiple scenarios simultaneously. However, for the current problem \eqref{eq: exact-N2} and \eqref{eq: exact-d2} are optimal decision rules. Plugging these in conditions \eqref{eq: PARO-conditions} reduces the problem of finding a PARO solution to solely comparing the performance of \ARO{} stage-1 decisions $d_1$ in non-worst-case scenarios.

In \citet{Iancu14} it is shown for linear optimization that, if we optimize over the robustly optimal solutions for a second criterion (scenario) that is in the relative interior of the uncertainty set, the resulting solution(s) are \PRO{}. In a similar fashion \PARO{} solutions to the current problem can be found. Let $X^{\text{ARO}}$ denote the set of \ARO{} stage-1 solutions to \eqref{eq: exact}. It turns out that consecutively optimizing over an auxiliary scenario where hyperfractionation is optimal and an auxiliary scenario where hypofractionation is optimal yields a set of \PARO{} solutions. Let $(\rho^{\text{aux-min}}, \tau^{\text{aux-min}}) \in \text{int} ( Z^{\text{min}} )$, where $\text{int}(\cdot)$ is the interior operator. Define the auxiliary optimization problem for the hypofractionation scenario:
\begin{align}\label{eq: ARO-exact-aux-min}
\max_{d_1 \in X^{\text{ARO}}} f\big(d_1,N_2^{\text{min}};\rho^{\text{aux-min}},\tau^{\text{aux-min}}\big).
\end{align}
Denote the set of optimal solutions to \eqref{eq: ARO-exact-aux-min} by $X^{\text{aux-min}}$. Similarly, let $(\rho^{\text{aux-max}}, \tau^{\text{aux-max}}) \in \text{int} (Z^{\text{max}} )$. Define the auxiliary optimization problem for the hyperfractionation scenario:
\begin{align}\label{eq: ARO-exact-aux-max}
\max_{d_1 \in X^{\text{aux-min}}} f\big(d_1,N_2^{\text{max}};\rho^{\text{aux-max}},\tau^{\text{aux-max}}\big).
\end{align}
Note that it uses $X^{\text{aux-min}}$ as input, i.e., we solve the auxiliary problems consecutively. Denote the set of optimal solutions to \eqref{eq: ARO-exact-aux-max} by $X^{\text{PARO}}$.
\begin{theorem} \label{theorem: PARO}
All solutions in $X^{\text{PARO}}$ are \PARO{} to \eqref{eq: exact}.
\end{theorem}
\begin{proof}
See \Cref{app: proof-PARO}. 
\end{proof}
Solving \eqref{eq: ARO-exact-aux-min} or \eqref{eq: ARO-exact-aux-max} entails maximizing a strictly convex or strictly concave function over a feasible set consisting of a small number of intervals or points. Hence, these auxiliary problems are easily solved. Note that the second auxiliary problem is only relevant if the first auxiliary problem has multiple optimal solutions. Switching their order, and optimizing \eqref{eq: ARO-exact-aux-min} over the set $X^{\text{aux-max}}$ may lead to different solutions, and these are also \PARO{}. Thus, in general $X^{\text{PARO}}$ does not contain all \PARO{} solutions. The two-step approach is necessary; numerical results show that optimizing over only one auxiliary scenario may indeed yield non-\PARO{} solutions.

\section{ARO: Biomarkers provide inexact information} \label{sec: inexact}
In this section we present an adjustable robust optimization approach to solve a more realistic version of the adaptive treatment-length problem. Because in practice it is impossible to exactly determine the $\alpha/\beta$ parameters from biomarker data, any values for the $\alpha/\beta$ parameters obtained during treatment are inexact. This section presents a model that accounts for uncertainty in biomarker information.

\subsection{Problem formulation}
The setup for the ARO problem with inexact data is based on \citet{DeRuiter17}. After $N_1$ fractions we obtain an estimate $(\hat{\rho},\hat{\tau})$ for $(\rho,\tau)$, the inverse $\alpha/\beta$ parameters for the OAR and the tumor. It is still assumed that \Cref{ass: omega-rho} holds for uncertainty set $Z$. Furthermore, we assume that $(\rho,\tau),~(\hat{\rho},\hat{\tau}) \in Z$ (as defined in \eqref{eq: Z}) and that $(\hat{\rho},\hat{\tau})-(\rho,\tau) \in \hat{Z}$, with
\begin{align}\label{eq: Zhat}
\hat{Z} = \{(\varepsilon^{\rho}, \varepsilon^{\tau}) \in \mathbb{R}^2: |\varepsilon^{\rho}| \leq r^{\rho}, |\varepsilon^{\tau}| \leq r^{\tau}\}.
\end{align}
Here $r^{\rho}$ and $r^{\tau}$ are parameters that define the accuracy of the estimate/observation $(\hat{\rho},\hat{\tau})$. Set $\hat{Z}$ is the uncertainty set around the inexact observation. This can also be written as $(\rho,\tau) \in \{(\hat{\rho},\hat{\tau})\} + \hat{Z}$, which is the Minkowski sum of a singleton and a set. This new set need not be contained in the original uncertainty set $Z$. Define
\begin{align} \label{eq: U}
U := \{(\rho,\tau, \hat{\rho}, \hat{\tau}): (\rho,\tau),~(\hat{\rho},\hat{\tau}) \in Z,~(\hat{\rho},\hat{\tau})-(\rho,\tau) \in \hat{Z}\},
\end{align}
and
\begin{align} \label{appeq: stage2-Z}
Z_{(\hat{\rho},\hat{\tau})} := \big( \{(\hat{\rho},\hat{\tau})\} + \hat{Z} \big) \cap Z.
\end{align}
The set $U$ contains all possible observation-realization pairs and $Z_{(\hat{\rho},\hat{\tau})}$ contains all possible realizations after observation of $(\hat{\rho},\hat{\tau})$. For given observation $(\hat{\rho},\hat{\tau})$, the new upper and lower bounds for $(\rho,\tau)$ are given by
\begin{subequations} \label{eq: bounds-given-obs}
\begin{align}
\hat{\tau}_L = \max\{\tau_L, \hat{\tau}-r^{\tau}\}, ~~& \hat{\tau}_U = \min\{\tau_U, \hat{\tau}+r^{\tau}\} \\
\hat{\rho}_L = \max\{\rho_L, \hat{\rho}-r^{\rho}\}, ~~& \hat{\rho}_U = \min\{\rho_U, \hat{\rho}+r^{\rho}\}.
\end{align}
\end{subequations}
Compared to \Cref{sec: exact}, we remove \Cref{ass: exact-dminmax} and impose a different (slightly stricter) assumption on the relation between $d^{\text{min}}$, $d_1^{\text{max}}$ and the bounds on $N_2$.
\begin{assumption}\label{ass: inexact-dminmax}
It holds that $d^{\text{min}} \leq d_1^{\text{max}}$ and 
\begin{equation}
\scalebox{0.9}{$ d_1^{\text{max}} \leq \min \Big\{g(d_1^{\text{min}},N_2^{\text{min}},N_1; \rho_L), g(d_1^{\text{min}},N_2^{\text{max}},N_1; \max\{\rho_L,\frac{\tau_L}{\sigma}-2r^{\rho}\}), g(d_1^{\text{min}},N_2^{\text{max}},N_1; \rho_U) \Big\}.$}
\end{equation}
\end{assumption}
\noindent The inexact data problem (IDP) analogous to \eqref{eq: exact} is given by
\begin{subequations}\label{eq: inexact}
\begin{align}
\max_{d_1,d_2(\hat{\rho},\hat{\tau}),N_2(\hat{\rho},\hat{\tau})}~&~ \min_{(\rho,\tau, \hat{\rho}, \hat{\tau}) \in U} N_1 d_1 + N_2(\hat{\rho},\hat{\tau}) d_2(\hat{\rho},\hat{\tau}) + \tau (N_1 d_1^2 + N_2(\hat{\rho},\hat{\tau}) d_2(\hat{\rho},\hat{\tau})^2), \label{eq: inexact-1} \\
\begin{split}
\text{s.t.}~~~~ ~&~ \sigma (N_1 d_1 + N_2(\hat{\rho},\hat{\tau}) d_2(\hat{\rho},\hat{\tau}))  + \rho \sigma^2 (N_1 d_1^2 + N_2(\hat{\rho},\hat{\tau}) d_2(\hat{\rho},\hat{\tau})^2) \\
& ~~~ \leq \text{BED}_{\text{tol}}(\rho),~~ \forall (\rho,\tau, \hat{\rho}, \hat{\tau}) \in U
\end{split} \label{eq: inexact-2} \\[0.5em]
            ~&~ N_2(\hat{\rho},\hat{\tau}) \in \{ N_2^{\text{min}},\dotsc, N_2^{\text{max}} \}, ~~\forall (\rho,\tau, \hat{\rho}, \hat{\tau}) \in U \label{eq: inexact-3} \\
            ~&~ d_2(\hat{\rho},\hat{\tau}) \geq d^{\text{min}}, ~~\forall (\rho,\tau, \hat{\rho}, \hat{\tau}) \in U \label{eq: inexact-4} \\
            ~&~ d^{\text{min}} \leq d_1 \leq d_1^{\text{max}}. \label{eq: inexact-5}
\end{align}
\end{subequations}

For stage-2 variables $d_2(\hat{\rho},\hat{\tau})$ and $N_2(\hat{\rho},\hat{\tau})$ it is indicated that they are a function of the observations $(\hat{\rho},\hat{\tau})$ instead of the uncertain parameters $(\rho,\tau)$. Similar to \Cref{sec: exact}, we formally define several properties of solutions. Let $X(\rho,\tau, \hat{\rho}, \hat{\tau})$ denote the feasible region defined by constraints \eqref{eq: inexact-2}-\eqref{eq: inexact-5} for fixed $(\rho,\tau, \hat{\rho}, \hat{\tau})$. 

\begin{definition}[Adjustable robust feasibility] \label{def: inexact-feasibility}
A tuple $(d_1,d_2(\cdot),N_2(\cdot))$ is adjustable robustly feasible (\ARF{}) to \eqref{eq: inexact} if $(d_1,d_2(\hat{\rho},\hat{\tau}),N_2(\hat{\rho},\hat{\tau})) \in X(\rho,\tau,\hat{\rho},\hat{\tau})$ for all $(\rho,\tau,\hat{\rho},\hat{\tau}) \in U$.
\end{definition}
\noindent Optimality of a decision rule is defined as follows.
\begin{definition}[Optimal decision rule] \label{def: inexact-optimal-DR}
For a given $d_1$, a decision rule pair $(d_2(\cdot),N_2(\cdot))$ is optimal to \eqref{eq: inexact} if $(d_1,d_2(\cdot),N_2(\cdot))$ is \ARF{} and for any $(\hat{\rho},\hat{\tau}) \in Z$ it holds that
\begin{align}
\begin{aligned}
\min_{(\rho,\tau) \in Z_{(\hat{\rho},\hat{\tau})}} N_1 d_1 + N_2(\hat{\rho},\hat{\tau}) d_2(\hat{\rho},\hat{\tau}) + \tau (N_1 d_1^2 + N_2(\hat{\rho},\hat{\tau}) d_2(\hat{\rho},\hat{\tau})^2)\\
\geq \min_{(\rho,\tau) \in Z_{(\hat{\rho},\hat{\tau})}} N_1 d_1 + \bar{N_2}(\hat{\rho},\hat{\tau}) \bar{d_2}(\hat{\rho},\hat{\tau}) + \tau (N_1 d_1^2 + \bar{N_2}(\hat{\rho},\hat{\tau}) \bar{d_2}(\hat{\rho},\hat{\tau})^2)
\end{aligned}
\end{align}
for every $(\bar{d_2}(\cdot),\bar{N_2}(\cdot))$ such that $(d_1,\bar{d_2}(\cdot),\bar{N_2}(\cdot))$ is \ARF{}.
\end{definition}
Note that for exact data, an optimal decision rule yields the optimal decision for any \emph{realization} in the uncertainty set $Z$ (given $d_1$). For inexact data, we call a decision rule optimal if it yields the maximum worst-case (guaranteed) objective value for any \emph{observation} in the uncertainty set $Z$.

\subsection{Optimal decision rules and conservative approximation} \label{sec: inexact-solution}
Depending on both the observed $(\hat{\rho}, \hat{\tau})$ and the quality of the biomarker information (i.e., $r^{\rho}$ and $r^{\tau})$, we may be able to immediately determine the optimal value for $N_2$. Therefore, we split the uncertainty set for the observations $(\hat{\rho},\hat{\tau})$. Define
\begin{subequations} \label{eq: def-Zi}
\begin{align}
Z_{\text{ID}}^{\text{min}} & = \{(\hat{\rho},\hat{\tau}) \in Z :  \hat{\tau}_L \geq \sigma \hat{\rho}_U  \} \label{eq: def-Z1}\\
Z_{\text{ID}}^{\text{int}} & = \{(\hat{\rho},\hat{\tau}) \in Z :  \sigma \hat{\rho}_L < \hat{\tau}_L < \sigma \hat{\rho}_U \} \label{eq: def-Z2}\\
Z_{\text{ID}}^{\text{max}} & = \{(\hat{\rho},\hat{\tau}) \in Z : \hat{\tau}_L \leq \sigma \hat{\rho}_L  \}, \label{eq: def-Z3}
\end{align}
\end{subequations}
so that $Z = Z_{\text{ID}}^{\text{min}} \cup Z_{\text{ID}}^{\text{int}} \cup Z_{\text{ID}}^{\text{max}}$. \Cref{fig: Z-split} provides an illustration. The split is such that if $(\hat{\rho},\hat{\tau}) \in Z_{\text{ID}}^{\text{min}}$ or $(\hat{\rho},\hat{\tau}) \in Z_{\text{ID}}^{\text{max}}$ only $N_2^{\text{min}}$ resp. $N_2^{\text{max}}$ fractions can be optimal in stage 2. Subset $Z_{\text{ID}}^{\text{int}}$ is the area between the dash-dotted lines. If $(\hat{\rho},\hat{\tau}) \in Z_{\text{ID}}^{\text{int}}$ both $N_2^{\text{min}}$ and $N_2^{\text{max}}$ fractions in stage 2 may be optimal for the true $(\rho,\tau)$. The following theorem states the optimal stage-2 decision rules for a given value of $d_1$.
\begin{theorem} \label{theorem: inexact-stage2}
Let $d_1$ be the stage-1 decision of \eqref{eq: inexact}. The decision rules
\begin{flalign} \label{eq: inexact-stage2-N2}
N_2^{\ast}(d_1;\hat{\rho},\hat{\tau}) =
\begin{cases}
N_2^{\text{min}} & \text{ if } (\hat{\rho},\hat{\tau}) \in Z_{\text{ID}}^{\text{min}} \\
\scriptstyle \argmax\limits_{N_2 \in \{N_2^{\text{min}},\dotsc,N_2^{\text{max}}\}} \min \{f(d_1,N_2;\hat{\rho}_L,\hat{\tau}_L), f(d_1,N_2;\hat{\rho}_U,\hat{\tau}_L) \} & \text{ if } (\hat{\rho},\hat{\tau}) \in Z_{\text{ID}}^{\text{int}} \\
N_2^{\text{max}} & \text{ if } (\hat{\rho},\hat{\tau}) \in Z_{\text{ID}}^{\text{max}},
\end{cases}
\end{flalign}
and
\begin{align} \label{eq: inexact-stage2-d2}
d_2^{\ast}(d_1;\hat{\rho},\hat{\tau}) =  \min \{g(d_1, N_1,N_2^{\ast}(d_1;\hat{\rho},\hat{\tau}); \hat{\rho}_L), g(d_1, N_1,N_2^{\ast}(d_1;\hat{\rho},\hat{\tau}); \hat{\rho}_U) \},
\end{align}
are optimal to \eqref{eq: inexact} for the given $d_1$.
\end{theorem}
\begin{proof}
See \Cref{app: proof-inexact-stage2}. 
\end{proof}
The worst-case optimal decision rule \eqref{eq: inexact-stage2-N2} may yield a value unequal to $N_2^{\text{min}}$ and $N_2^{\text{max}}$ if $(\hat{\rho},\hat{\tau}) \in Z_{\text{ID}}^{\text{int}}$.
\begin{figure}
\centering
\includegraphics{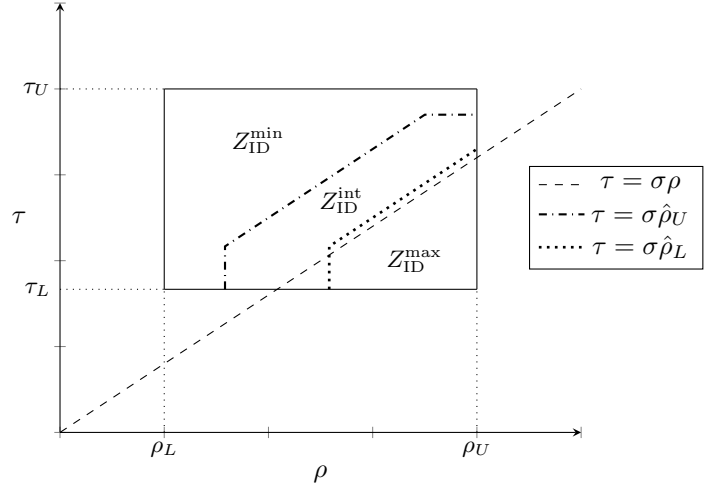}
\caption{\small The uncertainty set $Z$ (solid lines) for $(\hat{\rho},\hat{\tau})$ is split into $Z_{\text{ID}}^{\text{min}}$, $Z_{\text{ID}}^{\text{int}}$, $Z_{\text{ID}}^{\text{max}}$, according to \eqref{eq: def-Zi}. Subset $Z^{\text{int}}$ is the area between the dotted and dash-dotted curves. If $(\hat{\rho},\hat{\tau}) \in Z_{\text{ID}}^{\text{int}}$ both $N_2^{\text{min}}$ and $N_2^{\text{max}}$ fractions in stage 2 may be optimal for the true $(\rho,\tau)$. If $(\hat{\rho},\hat{\tau}) \in Z_{\text{ID}}^{\text{min}}$ or $(\hat{\rho},\hat{\tau}) \in Z_{\text{ID}}^{\text{max}}$ only $N_2^{\text{min}}$ resp. $N_2^{\text{max}}$ fractions can be optimal in stage 2. \label{fig: Z-split}}
\end{figure}
If $r^{\rho}$ and $r^{\tau}$ are zero, i.e., we have exact data, then it holds that $\hat{\tau}_L = \tau$ and $\hat{\rho}_L = \hat{\rho}_U = \rho$. Hence, the two functions $f$ in the RHS of \eqref{eq: inexact-stage2-N2} are equal, and the optimal $N_2^{\ast}$ is the one that maximizes the resulting function. One can verify that this does not depend on $d_1$. Hence, in case of exact data \Cref{theorem: inexact-stage2} reduces to \Cref{theorem: exact-stage2}.

It turns out that, after plugging in \eqref{eq: inexact-stage2-N2} and \eqref{eq: inexact-stage2-d2}, and splitting the uncertainty set according to \eqref{eq: def-Zi}, it is not apparent how to determine the optimal stage-1 decision $d_1^{\ast}$ for \eqref{eq: inexact}. In \Cref{app: proof-inexact-reformulation} the following lower bound problem to \eqref{eq: inexact} is derived, named the Approximate Inexact Data Problem (AIDP):
\begin{subequations}\label{eq: inexact2}
\begin{align}
\max_{d_1,q}~&~ q \\
\text{s.t.}~&~ q \leq f(d_1,N_2^{\text{min}};\rho_L,\tau_L) \label{eq: inexact2-2} \\
            ~&~ q \leq f(d_1,N_2^{\text{max}};\rho_U,\tau_L) \label{eq: inexact2-3} \\
            ~&~ q \leq K \label{eq: inexact2-4} \\
~&~ q \leq p(d_1) \label{eq: inexact2-5} \\
~&~ d^{\text{min}} \leq d_1 \leq d_1^{\text{max}}. \label{eq: inexact2-6}
\end{align}
\end{subequations}

\begin{figure}
\centering
\includegraphics{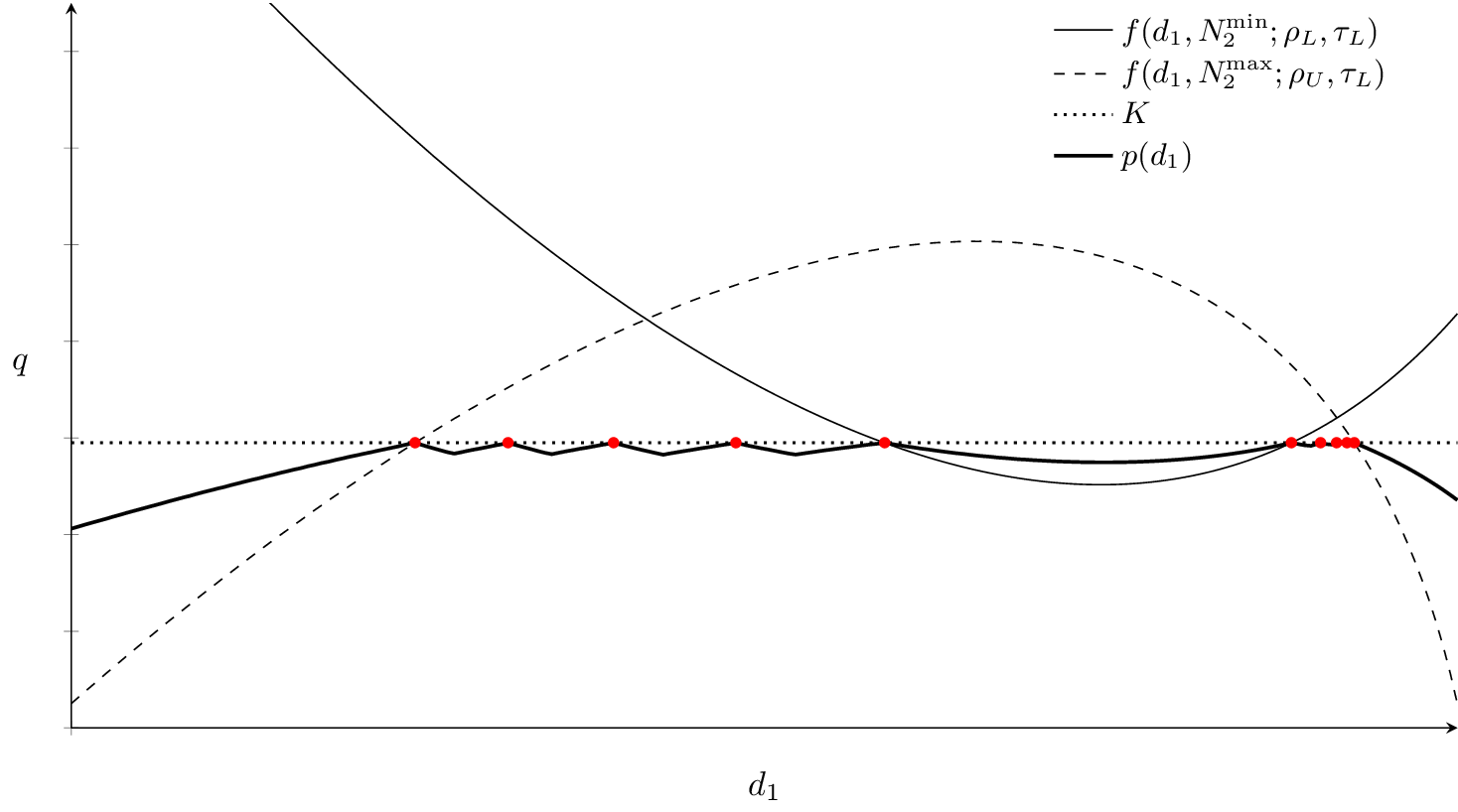}
\caption{\small Schematic illustration of \eqref{eq: inexact2}. Compared to the case with exact data (\Cref{fig: exact}), the thick black curve (constraint \eqref{eq: inexact2-5}) is extra. The solid, dashed and dotted lines/curves represent constraints \eqref{eq: inexact2-2}, \eqref{eq: inexact2-3} and \eqref{eq: inexact2-4}, respectively. Optimal solutions are indicated by red circles.\label{fig: inexact}}
\end{figure}

The AIDP is best explained using an example. \Cref{fig: inexact} illustrates a possible instance of \eqref{eq: inexact2}, displaying constraints \eqref{eq: inexact2-2}-\eqref{eq: inexact2-5}. Compared to \eqref{eq: exact3} for exact biomarker information (see \Cref{fig: exact}), problem \eqref{eq: inexact2} has the added constraint \eqref{eq: inexact2-5}; a piecewise convex-concave function $p(d_1)$ defined by \eqref{appeq: p} in \Cref{lemma: i=2-reform} (\Cref{app: auxiliary-lemmas}). It can be interpreted as follows. If $(\hat{\rho},\hat{\tau}) \in Z_{\text{ID}}^{\text{int}}$, the optimal number of stage-2 fractions can be inbetween $N_2^{\text{min}}$ and $N_2^{\text{max}}$, as shown in \Cref{theorem: inexact-stage2}. In those cases, the optimal number of fractions also depends on the already delivered stage-1 dose $d_1$. The upper kinks (black dots) in the piecewise convex-concave function $p(d_1)$ in \Cref{fig: inexact} indicate values of $d_1$ where the worst-case scenario changes. The lower kinks indicate values where the optimal number of stage-2 fractions changes. The exact expression for $p(d_1)$ does not provide additional insight and is therefore omitted here. 

In \Cref{fig: inexact}, optimal solutions are locations where \eqref{eq: inexact2-4} and \eqref{eq: inexact2-5} are both binding, indicated by red circles. Dose constraints \eqref{eq: inexact2-6} may cut off (some of) these points. If due to constraint \eqref{eq: inexact2-6} none of the circles are feasible, the optimum is at one of the boundaries for $d_1$. Constraint \eqref{eq: inexact2-5} is the only conservative constraint in \eqref{eq: inexact2}. Hence, only if the feasible values for $d_1$ are such that none of the circles in \Cref{fig: inexact} are feasible and constraint \eqref{eq: inexact2-5} is binding, it is possible that the optimal objective value of \eqref{eq: inexact2} is strictly worse than that of \eqref{eq: inexact}.

\Cref{lemma: convexity-f,lemma: intersection,lemma: i=2-reform} in \Cref{app: auxiliary-lemmas} provide information on the shape and intersection points of constraint functions \eqref{eq: inexact2-2}-\eqref{eq: inexact2-5}. Consequently, the optimal solution(s) of \eqref{eq: inexact} is/are easily obtained. If $(d_1^{\ast},q^{\ast})$  is optimal to AIDP \eqref{eq: inexact2}, and $N_2^{\ast}(\cdot)$ and $d_2^{\ast}(\cdot)$ are given by \eqref{eq: inexact-stage2-N2} and \eqref{eq: inexact-stage2-d2} then $(d_1^{\ast},d_2^{\ast}(\cdot),N_2^{\ast}(\cdot))$ is \ARF{} to the original IDP \eqref{eq: inexact}. It is \ARF{} because AIDP is a conservative approximation of IDP.

\subsection{Pareto robustly optimal solutions to conservative approximation} \label{sec: inexact-PARO}
\Cref{fig: inexact} also illustrates that it is possible that there are multiple optimal solutions to the AIDP \eqref{eq: inexact2}. Because the AIDP provides a conservative approximation to IDP \eqref{eq: inexact}, optimizing over auxiliary scenario(s) as in \Cref{sec: exact} does not necessarily yield a stage-1 decision $d_1$ that is \PARO{} to the original IDP. It turns out that a \PRO{} solution to AIDP is obtained from the set of robustly optimal solutions to AIDP if we consecutively optimize for two auxiliary observations such that any worst-case realization is in the interior of set $Z^{\text{min}}$ resp. $Z^{\text{max}}$. Two important remarks are in place here. First, a \PRO{} solution to AIDP need not be a \PARO{} solution to IDP, even if it is \ARO{} to IDP. Second, the required auxiliary scenarios need not exist; their existence depends on the values of $r^{\rho}$ and $r^{\tau}$. Hence, further details are omitted.

\section{Numerical results} \label{sec: results}
This section presents numerical results of the methods presented in Sections \ref{sec: exact} and \ref{sec: inexact}. First, Section \ref{sec: benchmark} describes the benchmark methods against which we compare the ARO method for EDP and IDP, and Section \ref{sec: setup} describes the setup of the numerical experiments.

\subsection{Benchmark static and folding horizon methods} \label{sec: benchmark}
We analyze the performance of the static and folding horizon nominal method (NOM and NOM-FH), the static and folding horizon robust optimization method (RO and RO-FH) and the adjustable robust optimization method (ARO). In the folding horizon approaches only the stage-1 decisions are implemented, and the model is re-optimized for the second stage once the biomarker information is revealed.

The static method NOM optimizes for the nominal parameter values $(\bar{\rho},\bar{\tau})$ and disregards any uncertainty and adaptability. This method is the same for both EDP and IDP. In stage 2, NOM-FH solves the nominal problem under the assumption that the obtained biomarker estimate is exact (which is an invalid assumption for IDP). This method does not guarantee robustly feasible solution (feasible for all $(\rho,\tau) \in Z$) nor a robustly optimal solution (\RO{}; (static) optimal for the worst-case realization $(\rho,\tau) \in Z$). The static method RO optimizes for the worst-case realization of $(\rho,\tau)$ in the uncertainty set $Z$, and disregards adaptability. For EDP the method RO-FH solves the same nominal problem as NOM-FH in stage 2; for IDP it solves a static robust optimization problem in stage 2, for which the uncertainty set is determined by the accuracy of the biomarker information. RO and RO-FH both guarantee an \RO{} solution.

One may add a folding horizon component to ARO (for either EDP or AIDP). This may improve the results in case a suboptimal stage-2 decision rule is used. However, as shown in Sections \ref{sec: exact-solution} and \ref{sec: inexact-solution}, the used stage-2 decision rules are optimal for any realized scenario (and for given stage-1 decision $d_1$ in case of inexact information). Hence, adding a folding horizon component will not change results.

Table \ref{table: solution-properties} provides an overview of the guaranteed solution properties of the methods. It is important to note that in case of inexact biomarker information (IDP) the methods RO and RO-FH guarantee an \RO{} solution, whereas ARO guarantees only an \ARF{} solution via solving the approximate problem AIDP. Depending on the approximation quality, the \ARF{} solution may be close or equal to an \ARO{} solution.
\begin{table}
\centering
\begin{tabular}{c | c c c c c} \toprule
         & \multicolumn{5}{c}{Method} \\ 
 Problem & NOM   & NOM-FH  & RO       & RO-FH    & ARO \\ \midrule
EDP     &  -    &   -     & \RO{}      & \RO{}      & \PARO{} \\ 			
IDP     &  -    &   -     & \RO{}        & \RO{} & \ARF{} \\ \bottomrule
\end{tabular}
\caption{\small Guaranteed solution properties of the five methods.  \label{table: solution-properties}}
\end{table}
Next to these five methods, we also report the results for the perfect information optimum (PI). This is the attainable optimum if from the start of the first fraction the true $(\rho,\tau)$ is exactly known. It can be formulated by taking the nominal problem and replacing the nominal parameter values by their true values. While in practice not a viable strategy, PI provides information on the value of perfect information, and allows us to put the performance of and differences between the other methods in perspective.

\subsection{Study setup} \label{sec: setup}
We use a data set of 30 non-small cell lung cancer (NSCLC) patients, treated with either photon or proton therapy. The mathematical models in Sections \ref{sec: exact} and \ref{sec: inexact} are based on the assumption that there is a single dose restricting OAR. We assume that the single dose restricting OAR is the normal lung itself\footnote{This is in line with clinical practice wherein normal lung is treated as the most important normal tissue and the treatment is designed as to minimize the radiation exposure to normal lung.}. For the models in Sections \ref{sec: exact} and \ref{sec: inexact}, an instance is defined by a tuple $(\sigma, \varphi, D, T, N_1, N^{\text{min}}, N^{\text{max}},d^{\text{min}},d_1^{\text{max}})$ and the relevant uncertainty sets.

Clinically, the number of treatment fractions varied between 33 and 37 fractions, with the majority of patients receiving 37 fractions. We set $N^{\text{min}} = 30$ and $N^{\text{max}} = 40$, to allow for slight deviations from the clinical standard. We assume the biomarker acquisition is made once $N_1=10$ fractions have been administered. This implies $N_2^{\text{min}} = 20$ and $N_2^{\text{max}} = 30$. Mean lung dose tolerance is $D=20$Gy, and we set $T = 37$ as that is the clinically standard regimen. The patients differ in $(\sigma,\varphi)$, which characterize the spatial dose distribution. Using the clinically delivered dose distribution, we derive for each normal lung voxel its dose sparing factor $s_i$ (see \Cref{sec: fractionation}). The dose shape factor $\varphi$ and the generalized dose sparing factor $\sigma$ for mean OAR BED are given by
\begin{subequations}
\begin{align}
\varphi &= \frac{n \sum_{i=1}^n s_i^2}{\big(\sum_{i=1}^n s_i \big)^2}, \\[0.5em]
\sigma  &= \frac{\sum_{i=1}^n s_i^2}{\sum_{i=1}^n s_i},
\end{align}
\end{subequations}
see \citet{Perko18} for details.

\citet{Cox86} estimate normal lung tissue $\alpha/\beta$ to be between 2.4 and 6.3. We set the nominal value at the midpoint $4.35$. The $\alpha/\beta$ of NSCLC lung tumors has traditionally been assumed to be above 10 Gy. However, recent NSCLC hypofractionation trials show promising results, indicating that NSCNC cells are more sensitive to fraction size than previously assumed, i.e., have a lower $\alpha/\beta$ than 10. \citet{Santiago16} find values between 2.2 and 9.0. We set the nominal value at the midpoint $5.6$. Put together, we get the following uncertainty set for the inverse $\alpha/\beta$ ratios:
\begin{align}
Z = \{(\rho,\tau) : 1/6.3 \leq \rho \leq 1/2.4,~1/9.0 \leq \tau \leq 1/2.2 \},
\end{align}
and the nominal scenario is $(\bar{\rho},\bar{\tau}) = (1/4.35, 1/5.6)$. With this uncertainty set, 20 out of 30 patient cases satisfy \Cref{ass: omega-rho}: these are used in the numerical experiments. For the remaining ten patients the optimal number of treatment fractions can be determined prior to treatment, so these are removed. 

To discriminate between multiple \ARO{} solutions, we follow the procedure detailed in Section \ref{sec: exact-PARO} in the case of exact biomarker information. The auxiliary scenarios are sampled uniformly from $\text{int}(Z^{\text{min}})$ and $\text{int}(Z^{\text{max}})$. In the case of inexact biomarker information, the procedure discussed in Section \ref{sec: inexact-PARO} is followed if the required auxiliary observations exist. If such observations exist, we sample uniformly from $Z$ until we have found two auxiliary observations for which any worst-case realization is in $\text{int}(Z^{\text{min}})$ resp. $\text{int}(Z^{\text{max}})$. If such observations do not exist, the robustly optimal solution to AIDP with lowest stage-1 dose is selected. The method RO (and therefore also RO-FH) may also find multiple robustly optimal solutions. For the obtained set of robustly optimal solutions we again follow the procedure detailed in Section \ref{sec: exact-PARO}. It turns out that for RO, the robustly optimal solutions often perform identical in non-worst-case scenarios. We optimize over the auxiliary scenarios consecutively; the first auxiliary scenario is the scenario corresponding to $\text{int}(Z^{\text{min}})$.

The minimum dose per fraction is $d^{\text{min}} = 1.5$ Gy and the maximum stage-1 dose per fraction is $d_1^{\text{max}} =3$. This satisfies \Cref{ass: exact-dminmax} (for EDP) and \ref{ass: inexact-dminmax} (for IDP). Using these parameter values, it is feasible to deliver $N_2^{\text{max}}$ fractions with dose $d^{\text{min}}$ in \emph{all} scenarios in $Z$. This means that stage-1 decisions cannot render stage 2 infeasible for RO, NOM (and their FH counterparts) or PI. Numerical results indicate that results are not sensitive to the choice of $d^{\text{min}}$ and $d_1^{\text{max}}$.

First, Section \ref{sec: results-exact} presents and discusses the results for the problem with exact biomarker information (EDP) of Section \ref{sec: exact}. After that, Section \ref{sec: results-inexact} presents and discusses the results for the problem with inexact biomarker information (IDP) of Section \ref{sec: inexact}. Lastly, Section \ref{sec: results-N1} again considers the inexact biomarker information case, and varies parameter $N_1$ in order to determine the optimal moment of biomarker acquisition.

We consider a sample of 200 scenarios for $(\rho,\tau)$ from $Z$. For each scenario, we compute the the average tumor BED over 20 patients is computed, thus creating a tumor BED distribution for the `average' patient. For this tumor BED distribution we report the mean, $5\%$ quantile and worst-case value. Next to this, we report the true worst-case tumor BED over $Z$ (averaged over 20 patients). Note that the true worst-case scenario can differ per patient, so the true worst-case BED is is typically not attained in the sample. For OAR violations, we report the percentage by which the OAR BED tolerance is exceeded (i.e., percentage overdose). The maximum violation is the maximum value found over all patients and scenarios. All reported decision variable statistics are averaged over all patients and scenarios.

\subsection{Results exact biomarker information} \label{sec: results-exact}
\Cref{table: numexp-res-unif} presents the results. Altogether, the results indicate that the value of exact information is high. NOM-FH performs very similar to ARO. This illustrates that ignoring uncertainty and adaptability in stage 1 neither compromises worst-case or mean performance, nor does it lead to OAR constraint violations if treatment can be adapted based on exact biomarker information. In fact, NOM-FH outperforms RO-FH, indicating that accounting for uncertainty in stage-1 is overly conservative.

NOM is the only method that is not worst-case optimal, but yields the highest mean tumor BED accross the sample. However, it is the only method that results in OAR constraint violations. In the nominal scenario $(\bar{\rho},\bar{\tau})$ it is optimal to hypofractionate for all patients, so the mean $N_2$ equals $20$ for NOM. The other static method, RO, is worst-case optimal, but yields lower tumor BED accross the entire sample. It delivers significantly more fractions on average, i.e., it decides to hyperfractionate more often.  

NOM-FH adds a folding horizon component to NOM, and this results in zero violations and worst-case optimality. It does have slightly lower sample mean tumor BED. RO-FH adds a folding horizon component to RO, and this results in improved performance accross the entire sample. It chooses to hypofractionate more often than RO. ARO is worst-case optimal and performs very similar to NOM-FH. Excluding NOM (for OAR constraint violations) and PI (not implementable), NOM-FH, ARO, RO-FH and RO yield the (possibly joint) highest objective value in $83.4\%$, $76.4\%$, $22.1\%$ and $0.6\%$ of all (scenario, patient) instances, respectively.

The results of \Cref{table: numexp-res-unif} show that the methods have different stage-1 decisions $d_1$; this indicates the existence of multiple worst-case optimal stage-1 solutions. As indicated in \Cref{sec: setup}, RO, RO-FH and ARO optimize over auxiliary scenarios in this case. According to \Cref{theorem: PARO}, ARO finds a \PARO{} solution this way. Overall, methods that deliver a relatively low dose in stage-1 perform better than the methods that deliver a higher dose. This may be data set-specific. From PI we see that for the majority of patients and scenarios hypofractionation is optimal (average $N_2=22.2$). We emphasize that for different data sets, where for the majority of scenarios and patients hyperfractionation is optimal, a higher stage-1 dose (which allows for $N^{\text{max}}$ constant-dose fractions) may perform better, such as the result of RO and RO-FH. 

\begin{table}[htb!]
\small
\centering
\begin{tabular}{l | c c c c c c} \toprule
       & \multicolumn{6}{c}{Method} \\
      						&    NOM & NOM-FH &     RO &  RO-FH &    ARO &    PI \\ \midrule
Tumor BED - sample mean  (Gy)              & 162.75 & 161.44 & 156.57 & 160.14 & 161.40 & 161.49 \\
		 Tumor BED - sample $5\%$ quantile (Gy)               & 151.98 & 150.94 & 147.57 & 150.04 & 150.90 & 151.04 \\
		 Tumor BED - sample wc (Gy) & 145.98 & 146.33 & 142.53 & 145.58 & 146.32 & 146.39 \\
		 		 \textbf{Tumor BED - wc over $Z$} (\textbf{Gy}) & \textbf{114.72} & \textbf{116.19} & \textbf{116.19} & \textbf{116.19} & \textbf{116.19} & \textbf{116.19}\\
		 OAR violation - mean ($\%$) 				     &   1.25 &   0    &   0    &   0    &   0    &   0    \\
	     OAR violation - max  ($\%$) 				     &  4.22 &   0    &   0    &   0    &   0    &   0    \\
         Stage-1 dose $d_1$ (Gy) 					     &   1.50 &   1.50 &   2.29 &   2.29 &   1.51 &   1.66 \\ 
         Stage-2 dose $d_2$ (Gy) 					     &   3.45 &   3.24 &   2.48 &   2.95 &   3.24 &   3.19 \\ 
         Stage-2 fractions $N_2$ 					     &   20.0 &   22.2 &   27.2 &   22.2 &   22.2 &   22.2 \\ \bottomrule
\end{tabular}

\caption{\small Results for experiments with exact biomarker information and uniform sampling of $(\rho,\tau)$ over $Z$ (200 scenarios). For each scenario, results are averaged over 20 patients$^{\ast}$. All methods optimize for worst-case tumor BED in $Z$, which is displayed in bold.\\
$^{\ast}$: the maximum OAR violation is computed over all patients and scenarios.\label{table: numexp-res-unif}}
\end{table}
\begin{figure}[htb!]
\centering
\includegraphics[scale=0.8]{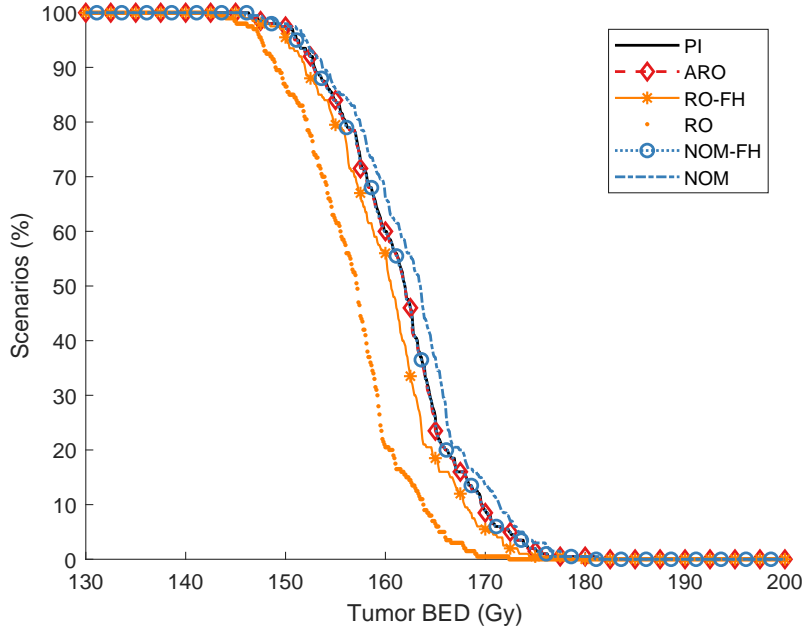}
\caption{\small Cumulative scenario-tumor BED graph for experiments with exact biomarker information and uniform sampling of $(\rho,\tau)$ over $Z$ (200 scenarios). A point $(x,y)$ indicates that in $y\%$ of scenarios the tumor BED (averaged over 20 patients) is at least $x$ Gy. ARO and NOM-FH are very close to PI. \label{fig: numexp-exact}}
\end{figure}

\Cref{fig: numexp-exact} shows the complete cumulative scenario-tumor BED graph. A point $(x, y)$ in \Cref{fig: numexp-exact} can be interpreted as follows: for the average patient, in $y\%$ of scenarios the tumor BED is at least $x$ Gy. The results clearly demonstrate that RO and RO-FH are outperformed by the other methods. Both NOM-FH and ARO are visually almost indistinguishable from PI. NOM performs even better accross the entire sample (except the first percent of the sample), at the cost of OAR BED violations. 

To see the difference in mean performance between the multiple worst-case optimal solutions, we compare \PARO{} solution found by the ARO method to the \ARO{} solution that performs worst in the two auxiliary scenarios. \Cref{table: numexp-ARO-best-worst} shows the results, OAR constraint violations are zero in all cases. The worst-performing \ARO{} solution has a considerably higher stage-1 dose. This implies that (for the current parameter settings) delivering a high stage-1 dose does not allow as much adjustment possibilities in stage 2 as a low stage-1 dose, but it does allow for adjustments to reach the worst-case optimum. Relative to the results of \Cref{table: numexp-res-unif}, the difference between the best and worst ARO solution is considerable: the worst-performing ARO solution performs worse than the RO-FH solution.
\begin{table}[htb!]
\small
\centering
\begin{tabular}{l | c c } \toprule
       & \multicolumn{2}{c}{Method} \\
      						&  \ARO{}$_{\text{worst}}$ & \ARO{}$_{\text{best}}$ \\ \midrule
Tumor BED - sample mean  (Gy)              & 159.87 & 161.40\\
Tumor BED - sample $5\%$ quantile (Gy)     & 149.86 & 150.90 \\
Tumor BED - sample worst-case (Gy) 		   & 145.34 & 146.32\\
		 \textbf{Tumor BED - wc over $Z$} (\textbf{Gy}) & \textbf{116.19} & \textbf{116.19} \\
         Stage-1 dose $d_1$ (Gy) 			    &   2.67 &   1.51 \\ 
         Stage-2 dose $d_2$ (Gy) 			    &   2.79 &   3.24 \\ 
         Stage-2 fractions $N_2$ 			    &   22.2 &   22.2 \\ \bottomrule
\end{tabular}
\caption{\small Comparison between the best (\PARO{}) and worst performing \ARO{} solutions, for uniform sampling of $(\rho,\tau)$ over $Z$ (200 scenarios). For each scenario, results are averaged over 20 patients.  All methods optimize for worst-case tumor BED in $Z$, which is displayed in bold. OAR constraint violations are zero in all cases.\label{table: numexp-ARO-best-worst}}
\end{table}

\Cref{app: out-of-sample} reports the results of an auxiliary experiment where the $(\rho,\tau)$ samples are drawn from a superset of $Z$, to compare the out-of-sample performance of the methods. NOM remains the only method with OAR constraint violations. Compared to \Cref{table: numexp-res-unif}, static methods NOM and RO have poor performance. The relative performance of the adaptive methods remains mostly unchanged. 

Altogether, the results of \Cref{sec: results-exact} demonstrate that if exact biomarker information is available mid-treatment, most stage-1 decisions allow for sufficient adaptation space in stage 2, also with realizations outside of $Z$. Different stage-1 decisions yield the worst-case optimum, have good performance on the scenario sample and have no OAR BED violations. We note that all presented differences in tumor BED are of relatively small magnitude. One reason for this is that the number of stage-2 fractions is restricted to $[N_2^{\text{min}}, N_2^{\text{max}}] = [20, 30]$. If the minimum number of fractions represents a `true' hypofractionation case, the dose per fraction can vary more, and the difference in performance between hypo- and hyperfractionation strategies is amplified.

\subsection{Results inexact biomarker information} \label{sec: results-inexact}
In case of in inexact biomarker information (IDP), we do not obtain the true parameter values $(\rho,\tau)$ after $N_1=10$ fractions, but only an estimate $(\hat{\rho},\hat{\tau})$. As discussed in Section \ref{sec: inexact}, we specify a new uncertainty set $\hat{Z}$ such that $(\hat{\rho},\hat{\tau})-(\rho,\tau) \in \hat{Z}$. Let $\text{DQ} \in [0,1]$ indicate the data quality. Then we set $\hat{Z}$ such that the width of the new uncertainty intervals for $\tau$ and $\rho$ is $(1-\text{DQ})$ times the width of the original intervals $[\tau_L,\tau_U]$ and $[\rho_L,\rho_U]$. That is, $\text{DQ} \cdot 100\%$ can be interpreted as the percentage by which the uncertainty intervals can be reduced due to the observation. The relation with the accuracy parameter $r^{\rho}$ (or similarly $r^{\tau}$) is given by
\begin{align}
r^{\rho} = \frac{1}{2}(\rho_U-\rho_L)(1-\text{DQ}).
\end{align}
Note that even $\text{DQ}=0$ has some value as the new interval is centered around the observation, which already cuts off part of the original uncertainty set $Z$. We pick $\text{DQ}=2/3$, so the obtained information after fraction $N_1$ reduces the size of the interval by $66.7\%$ around the new observation. Variations for $\text{DQ}$ are considered in Section \ref{sec: results-N1}. For all 20 patients the required auxiliary scenarios for the method of Section \ref{sec: inexact-PARO} can be found. 

Table \ref{table: numexp-res-unif-id} shows the results. The robust methods RO, RO-FH and ARO are all worst-case optimal. This indicates that, although not theoretically guaranteed, ARO finds an ARO solution in all considered scenarios. The mean performance of RO-FH and ARO is further away from PI than in the case with exact biomarker information (\Cref{table: numexp-res-unif}). This is as expected, as due to inexact observations the possibility for ARO and RO-FH to make adjustments is less valuable, whereas PI is not in influenced by this. On the other hand, NOM and NOM-FH are not worst-case optimal, but have better performance on the sample of scenarios, at the cost of OAR violations.

\begin{table}[htb!]
\small
\centering
\begin{tabular}{l | c c c c c c} \toprule
       & \multicolumn{6}{c}{Method} \\
      						&    NOM & NOM-FH &     RO &  RO-FH &    ARO &    PI \\ \midrule
Tumor BED - sample mean  (Gy)                   & 162.52&	161.03	&156.38	&158.76	&159.46	&161.16\\
Tumor BED - sample $5\%$ quantile (Gy)          & 151.36&	150.12	&147.04	&148.51	&148.94	&150.21\\
Tumor BED - sample wc (Gy) 						& 147.79&	146.04	&144.01	&145.02	&145.33	&146.18\\
\textbf{Tumor BED - wc over $Z$} (\textbf{Gy})  & \textbf{114.72} &	\textbf{115.96}	&\textbf{116.19}	&\textbf{116.19}	&\textbf{116.19}	&\textbf{116.19}\\
OAR violation - mean ($\%$) 					& 1.25	&0.16	&0	& 0	&0	&0 \\
OAR violation - max  ($\%$) 					& 4.23	&1.49	&0	& 0	&0	&0 \\
Stage-1 dose $d_1$ (Gy) 						& 1.50	&1.50	&2.29	&2.29	&1.79	&1.65\\
Stage-2 dose $d_2$ (Gy) 						& 3.45	&3.25	&2.48	&2.78	&2.92	&3.20\\
Stage-2 fractions $N_2$ 						& 20.0	&22.0	&27.2	&24	&24.3	&22.1\\\bottomrule
\end{tabular}
\caption{\small Results for experiments with inexact biomarker information (data quality $\text{DQ} = 2/3$) and uniform sampling of $(\rho,\tau)$ over $Z$ (200 scenarios). All results are averages over a sample of $20$ patients. For each scenario, results are averaged over 20 patients$^{\ast}$. All methods optimize for worst-case tumor BED in $Z$, which is displayed in bold.\\
$^{\ast}$: the maximum OAR violation is computed over all patients and scenarios. \label{table: numexp-res-unif-id}}
\end{table}

ARO is the only method (together with PI) that has a different stage-1 decision than in
the case with exact biomarker information. This is because it is the only method that takes inexactness of biomarker information into account at the start of stage 1. The average stage-1 dose $d_1$ differs considerably between ARO and RO-FH, whereas their worst-case performance is equal on average (and equal to PI). This demonstrates the existence of multiple worst-case optimal solutions. Whereas optimizing worst-case optimal solutions for ARO over two auxiliary scenarios does not guarantee a PARO solution (\Cref{sec: inexact-PARO}), results in \Cref{table: numexp-res-unif-id} indicate that it does yield solutions that perform slightly better on average than RO-FH.

For ARO, it is noteworthy that the average number of stage-2 fractions (24.3 fx) differs from that of PI (22.1 fx). Although ARO uses optimal decision rules for stage 2, these are optimal for the worst-case scenario in the new uncertainty set $Z_{(\hat{\rho},\hat{\tau})}$, and need not be optimal for the `true' realization in this set. In fact, NOM-FH treats the inexact biomarker information as the `true' parameter values, and administers 22.0 fractions, on average, which is closer than that of PI. Although the fractionation decision of NOM-FH is not worst-case optimal, \Cref{table: numexp-res-unif-id} shows that it performs better on the samples of scenarios.

\Cref{fig: numexp-inexact} shows the complete cumulative scenario-tumor BED graph for the `average patient'. Whereas in case of exact biomarker information (\Cref{fig: numexp-exact}), the ARO line was very close to PI, here a clear difference can be observed. NOM and NOM-FH outperform ARO (and RO and RO-FH) over the entire distribution. 

\begin{figure}[htb!]
\centering
\includegraphics[scale=0.8]{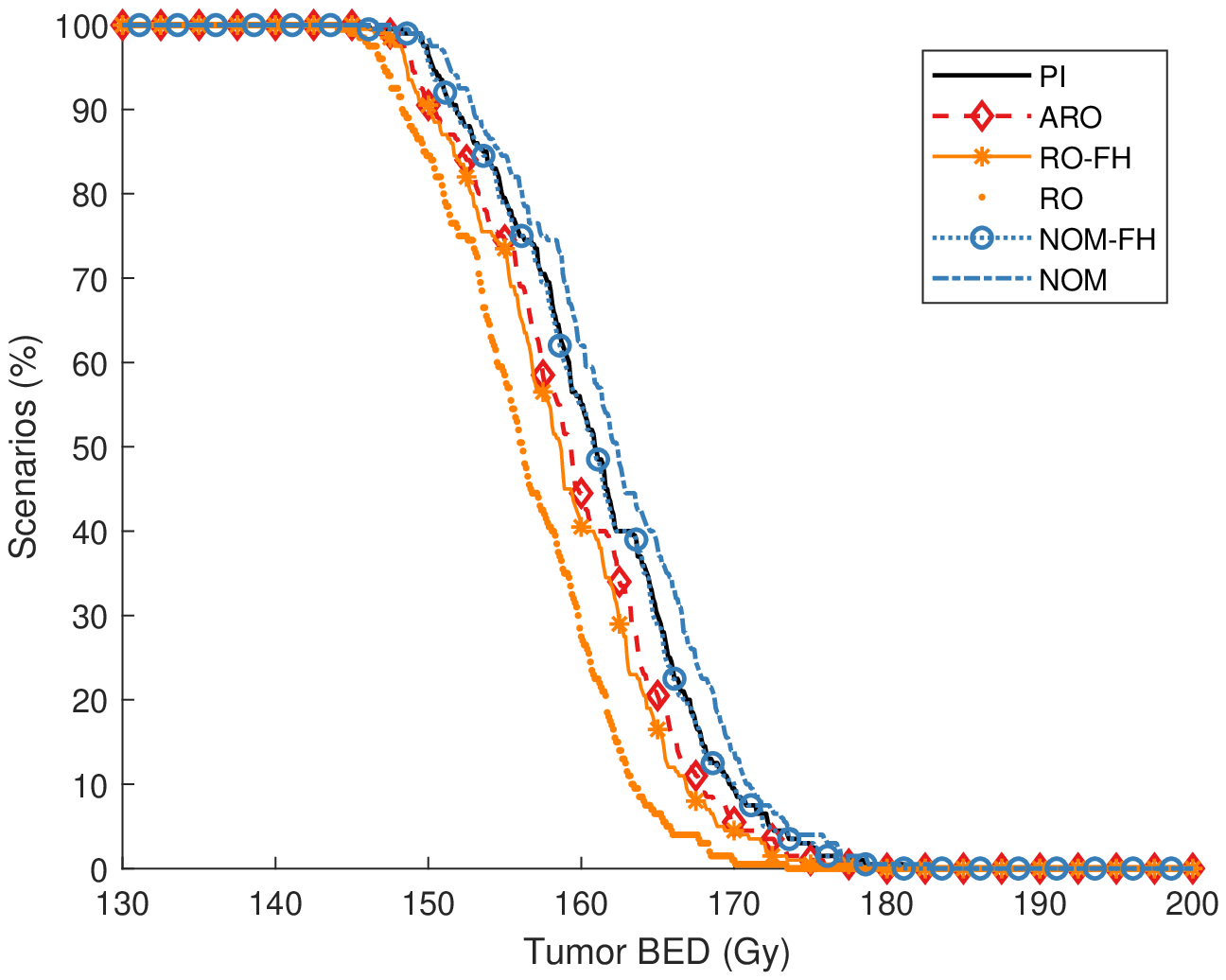}
\caption{\small  Cumulative scenario-tumor BED graph for experiments with inexact biomarker information (data quality $\text{DQ} = 2/3$) and uniform sampling of $(\rho,\tau)$ over $Z$ (200 scenarios). A point $(x,y)$ indicates that in $y\%$ of scenarios the tumor BED (averaged over 20 patients) is at least $x$ Gy. NOM-FH is very close to PI. \label{fig: numexp-inexact}}
\end{figure}

The good performance of NOM and NOM-FH in terms of sample mean tumor BED does come at the cost of OAR violations. However, these are relatively minor. The reason for this is that the number of stage-2 fractions is relatively high (between 20 and 30fx), so any method delivers reasonably low dose per fraction in stage 2. Consequently, the quadratic term in the BED model is smaller, and so is the influence of the $\alpha/\beta$ parameters. With higher dose per fraction, the use of incorrect (e.g., nominal) $\alpha/\beta$ parameter values may yield higher OAR constraint violations. Preliminary experiments for stereotactic body radiation therapy (SBRT, an RT modality that uses around five high dose fractions) indeed result in slightly higher OAR constraint violations for NOM and NOM-FH. In any case, a trade-off can be observed between higher tumor BED attained by NOM and NOM-FH and associated OAR constraint violations.

\subsection{Optimal moment of biomarker acquisition} \label{sec: results-N1}
The moment of biomarker observation, need not be fixed. Part of the decision making process then involves choosing this observation moment such that it maximally improves treatment quality. Late observation may result in limited possibilities for treatment adaptation, whereas with too early observation one cannot yet reliably observe the true individual patient response. Although one can incorporate $N_1$ as a decision variable in the mathematical model, the small decision space allows to simply vary its value in numerical experiments. We assume a (hypothetical) mathematical relationship between information point $N_1$ and the data quality parameter DQ. With $N^{\text{max}}$ the maximum number of fractions, we consider the following three data quality functions:
\begin{subequations}
\begin{align}
\text{DQ}_1(N_1) ~&= \Bigg(\frac{N_1}{N^{\text{max}}} \Bigg)^4 \\
\text{DQ}_2(N_1) ~&= \frac{N_1}{N^{\text{max}}} \\
\text{DQ}_3(N_1) ~&= \Bigg(\frac{N_1}{N^{\text{max}}} \Bigg)^{1/4}.
\end{align}
\end{subequations}
Hence, $\text{DQ}_1$, $\text{DQ}_2$ and $\text{DQ}_3$ describe a convex, linear and concave relationship between observation moment and data quality, respectively. Figure \ref{fig: data-quality-shapes} shows the graphs of the three functions. Whether $\text{DQ}_1$, $\text{DQ}_2$ or $\text{DQ}_3$ is most realistic depends on the specific biomarker(s) that is/are used, see \Cref{sec: biomarkers} for details.

\begin{figure}
\centering
\includegraphics[scale=1]{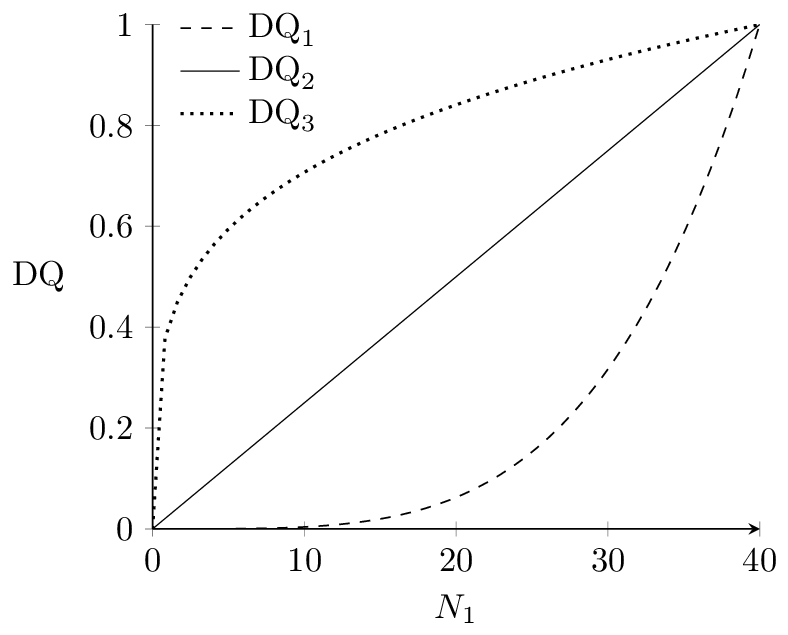}
\caption{\small The biomarker data quality is a function of the number of treatment fractions $N_1$ after which it is acquired. We consider three functions $\text{DQ}_i(N_1)$, $i=1,2,3$. \label{fig: data-quality-shapes}}
\end{figure}

\begin{figure}[ht!]
\hspace*{-0.6cm}
\begin{subfigure}{0.55\textwidth}
\includegraphics[scale=0.58]{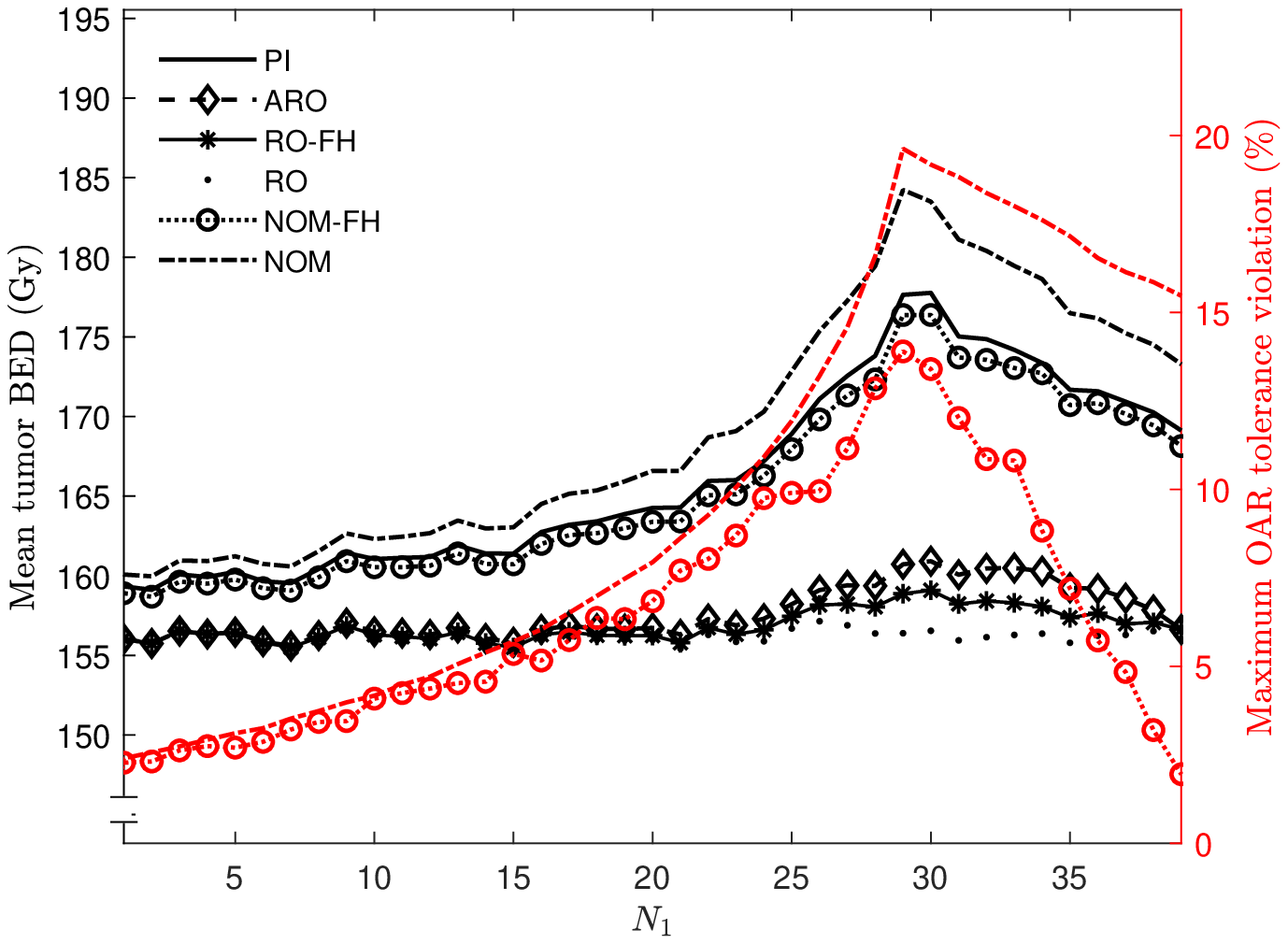}%
\caption{\small $\text{DQ}_1$ (convex) \label{fig: DQplots-convex}}
\end{subfigure}
\begin{subfigure}{0.55\textwidth}
\includegraphics[scale=0.58]{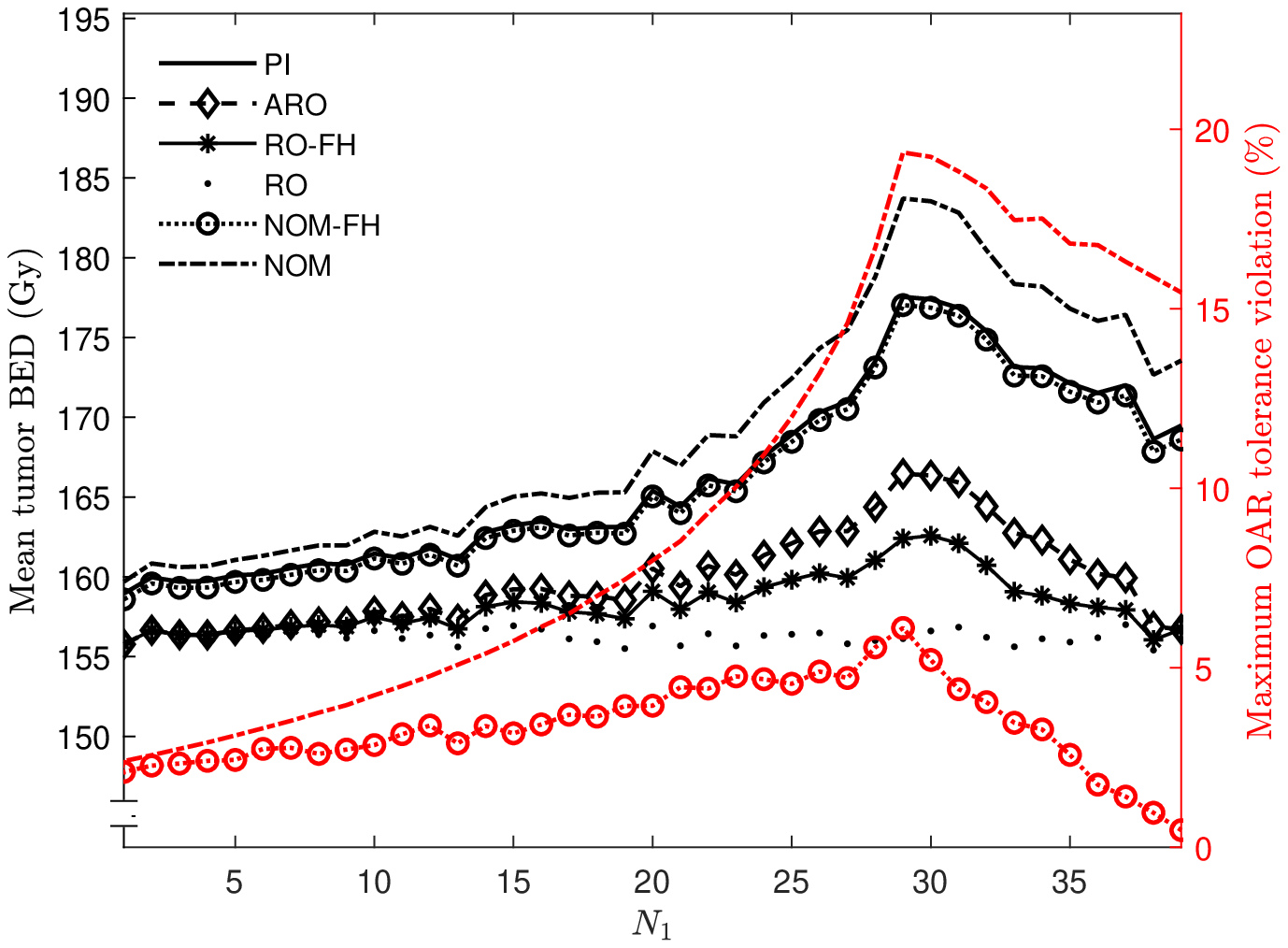}%
\caption{\small $\text{DQ}_1$ (linear) \label{fig: DQplots-linear}}
\end{subfigure}
\begin{subfigure}{\textwidth}
\centering
\includegraphics[scale=0.58]{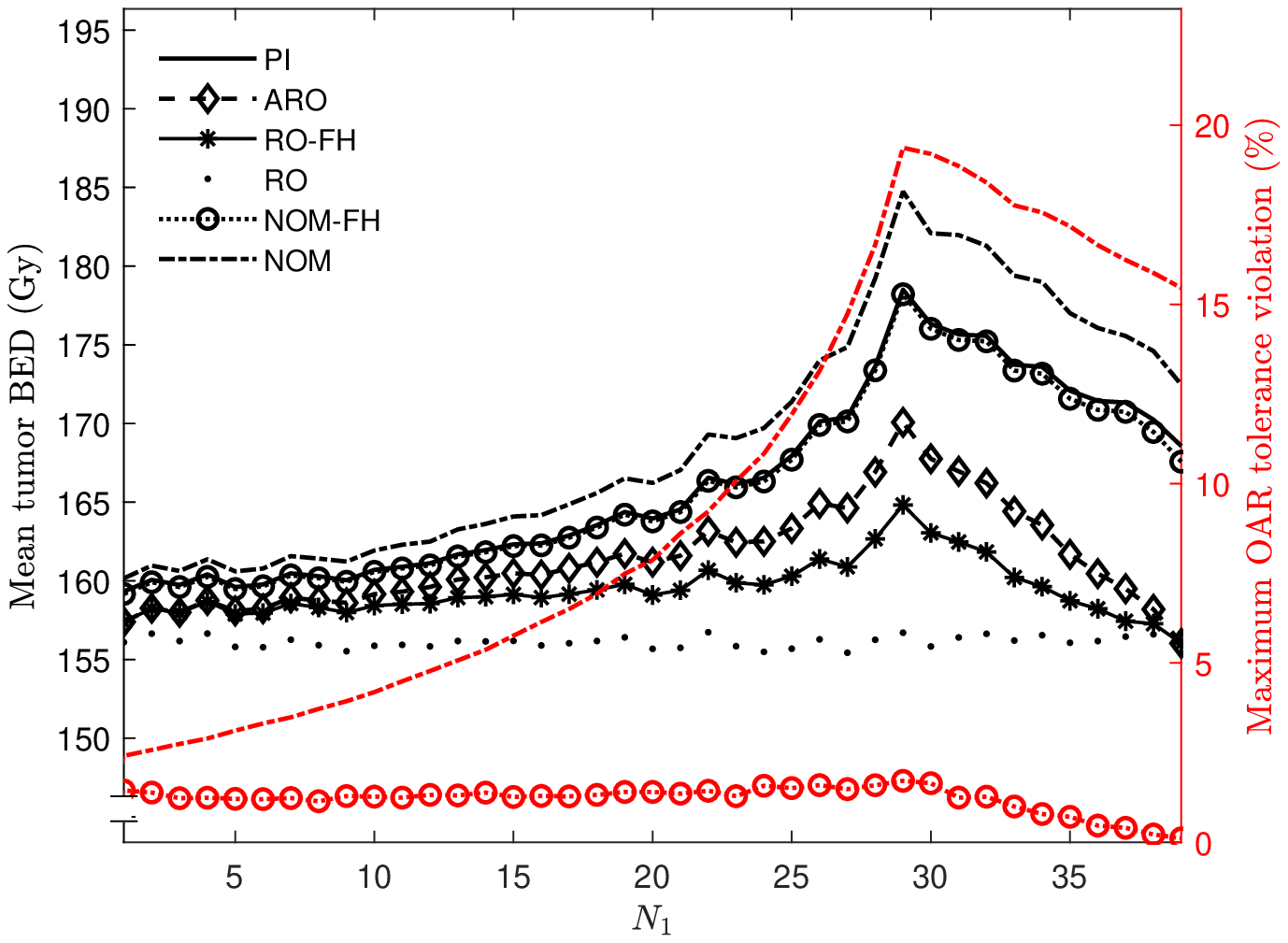}%
\caption{\small $\text{DQ}_2$ (concave) \label{fig: DQplots-concave}}
\end{subfigure}
\caption{\small Results for varying the information point $N_1$ from $1$ to $N^{\text{max}}-1$, for data quality functions $DQ_i(N_1)$, $i=1,2,3$.  The maximum OAR BED constraint violation ($\%$) of NOM (dash-dotted) and NOM-FH (dotted, circle marker) is measured against the right vertcial axis. Note that the left vertical axis measures displays the mean tumor BED (averaged over all patients and scenarios), while the methods maximize the worst-case tumor BED per patient. \label{fig: DQplots}}
\end{figure}

Figure \ref{fig: DQplots} shows the results for the numerical experiments where we vary the information point $N_1$ from $0$ to $N^{\text{max}}-1$ for data quality functions $\text{DQ}_i(N_1)$, $i=1,2,3$. The left vertical axis indicates the mean tumor BED (averaged over all patients and scenarios), the right vertical axis indicates the maximum OAR tolerance violation for NOM and NOM-FH. It is important to note that as $N_1$ increases past $N^{\text{min}}=30$, this also increases the minimum number of fractions correspondingly. Moreover, the dose per fraction is constant per treatment stage, so the choice of $N_1$ also influences the types of treatments that can be delivered. 

For these reasons the curve for PI is not constant, even though it does not actually use biomarker information. The optimal moment of biomarker acquisition for PI is $N_1 =  29$. This is because the minimum number of fractions is $N^{\text{min}} =30$. Hence, if hypofractionation is optimal we can deliver one more fraction with high dose, and deliver a low dose in stage 1. If hyperfractionation is optimal we can deliver $11$ more fractions (and get the total maximum of $40$) with low dose. Having $N_1 > 29$ forces the use of more than $N_2^{\text{min}} =30$ fractions, which is disadvantageous for those (patients, scenario) cases where hypofractionation is optimal. 

For the latter reason the NOM curve is also not constant. It results in a higher tumor BED than PI for any value of $N_1$, at the cost of OAR violations of up to $20\%$. NOM-FH yields a sample mean tumor BED close to PI for all three DQ functions, and the OAR violations depend on the DQ function. With poor data quality (\Cref{fig: DQplots-convex}) and an observation moment close to $N_1=29$,  OAR violations over $10\%$ are possible, despite the fact that NOM-FH is an adaptive method. On the other hand, with good (\Cref{fig: DQplots-convex}) data quality, the violations remain below $2\%$. The OAR tolerance violations are highest near $N_1 = 29$. This is because in case of hypofractionation in stage 2 the influence of the $\alpha/\beta$ parameters is highest, as was noted in \Cref{sec: results-inexact}.

The robust methods RO, RO-FH and ARO do not result in any OAR violations, by construction. The better the data quality, the larger the differences between RO, RO-FH and ARO. This implies that, if robustness is required, there is value in (i) adapting based on inexact information, (ii) taking adaptability into account when planning the stage-1 dose. The good performance of NOM-FH shows that this value diminishes if OAR violations are allowed. NOM-FH does not account for adaptability, and does not take inexactness of biomarker information into account. Nevertheless, it results in higher sample mean tumor BED for any $N_1$, and the difference increases from poor (convex) to good (concave) data quality. Thus, \Cref{fig: DQplots-concave} illustrates the trade-off between higher sample mean tumor BED and possible OAR violations that was also observed in \Cref{sec: results-inexact} (for the entire sample distribution).

The shape of the data quality function influences the optimal moment of biomarker observation only slightly. For all adaptive methods, we find that the peak is more pronounced for high data quality (concave) than low data quality (convex), but it is centered around $N_1=29$. In case of convex data quality the peak is relatively flat, indicating a trade-off between observing at $N_1=29$ (giving maximum adaptation flexibility) and postponing (waiting for higher data accuracy).

\section{Concluding remarks}\label{sec: conclusion}
In this paper we have presented an ARO approach to optimally adapt the treatment length of radiation therapy treatments, using mid-treatment biomarker information. Using an ARO approach, adaptability is taken into account prior to treatment and it provides insight into the optimal stage-2 decisions. 

In the case of exact biomarker information, there is sufficient space to adapt, and numerical results show that taking into account both robustness and adaptability is not necessary. In the case of inexact biomarker information, adaptive strategies can use only parameter estimates instead of true parameter values, and may still yield violations if this uncertainty is not accounted for. Accounting for adaptability and inexactness of biomarker information is particularly beneficial when robustness (w.r.t. OAR violations) is of high importance. If minor OAR violations are allowed, NOM-FH is a good performing alternative, which can outperform ARO. NOM-FH and ARO thus yield a trade-off between higher performance and OAR violations. Both the difference in performance and the magnitude of OAR violations of NOM-FH are highly influenced by the data quality (i.e., accuracy of parameter estimates).

The current setting can be extended in several ways. In practice the tumor and OAR $\alpha/\beta$ values would have to be estimated from actual biomarkers (e.g., imaging, blood-based biomarkers, genotyping), which can be incorporated in the model. Furthermore, the approach can be extended to heterogeneous tumor response (different $\alpha/\beta$ ratios for different tumor subvolumes), or time-dependent response parameters. Other RT applications may also benefit from ARO, such as re-optimization to account for organ motion or setup errors, optimization using the MR-linac or combining RT with chemotherapy.

\bibliographystyle{apalike}
\bibliography{ReferenceList}

\section*{Appendix}
\appendix
\renewcommand{\theequation}{\thesection.\arabic{equation}}
\renewcommand{\thefigure}{\thesection.\arabic{figure}}
\renewcommand{\thetable}{\thesection.\arabic{table}}

\setcounter{equation}{0}
\setcounter{figure}{0}
\setcounter{table}{0}
\sectionfont{\normalsize}
\subsectionfont{\small}

\section{Results exact biomarker information: out-of-sample performance} \label{app: out-of-sample}
\small
To investigate the out-of-sample performance of the methods, we assume a uniform distribution for $(\rho,\tau)$ over a larger set than $Z$. We can write $Z$ as
\begin{align}
Z = \{(\rho,\tau): \rho_L \leq \rho \leq \rho_U, \tau_L \leq \tau \leq \tau_U \} = \{(\rho,\tau): | \bar{\rho} - \rho | \leq \varepsilon_{\rho},~| \bar{\tau} - \tau | \leq \varepsilon_{\tau} \},
\end{align}
where $(\varepsilon_{\rho},\varepsilon_{\tau})$ is the maximum deviation from the nominal scenario $(\bar{\rho},\bar{\tau})$. This allows us to define
\begin{align}
Z_c = \{(\rho,\tau): | \bar{\rho} - \rho | \leq c\varepsilon_{\rho},~| \bar{\tau} - \tau | \leq c\varepsilon_{\tau} \},
\end{align}
where $c>0$ is a parameter. We assume a uniform distribution over the new set $Z_c$. If $c=1$, we have $Z_c = Z$, so we sample exactly from $Z$. If $c>1$, we sample from an interval that is $c^2$ times as large as $Z$ ($c$ times larger for both $\rho$ and $\tau$). For $c=2$ we obtain the results in Table \ref{table: numexp-res-unif2}. The stage-1 dose $d_1$ is the same as in \Cref{table: numexp-res-unif} for all methods except PI, because PI is the only method that is aware that the samples are not taken from uncertainty set $Z$ but from $Z_2$. For NOM, the maximum violation percentage has increased slightly. All other methods are able to deal with the out-of-sample realizations and do not have any OAR constraint violations.

\begin{table}[htb!]
\small
\centering
\begin{tabular}{l | c c c c c c} \toprule
       & \multicolumn{6}{c}{Method} \\
      						&    NOM & NOM-FH &     RO &  RO-FH &    ARO &    PI \\ \midrule                         
Tumor BED - sample mean  (Gy)                   & 156.49&	158.92&	151.30&	156.83&	158.87&	159.06 \\
Tumor BED - sample $5\%$ quantile (Gy)          & 140.14&	142.11&	137.55&	141.11&	142.11&	142.32 \\
Tumor BED - sample wc (Gy) 						& 134.06&	136.80&	132.56&	136.09&	136.79&	137.12 \\
\textbf{Tumor BED - wc over $Z$} (\textbf{Gy})  & \textbf{114.72}&	\textbf{116.19}&	\textbf{116.19}&	\textbf{116.19}&	\textbf{116.19}&	\textbf{116.19} \\
OAR violation - mean ($\%$) 					& 1.16	&	0	&	0	&	0	&	0	&	0 \\
OAR violation - max  ($\%$) 				    & 5.30	&	0	&	0	&	0	&	0	&	0 \\
Stage-1 dose $d_1$ (Gy) 					    & 1.50	&	1.50&	2.29&	2.29&	1.51&	1.73 \\
Stage-2 dose $d_2$ (Gy) 					    & 3.45	&	3.21&	2.48&	2.92&	3.21&	3.15 \\
Stage-2 fractions $N_2$ 					    & 20	&	22.9&	27.2&	22.9&	22.9&	22.9 \\ \bottomrule         
\end{tabular}
\caption{\small Results for experiments with exact biomarker information and uniform sampling of $(\rho,\tau)$ over $Z_2$. For each scenario, results are averaged over 20 patients$^{\ast}$. All methods optimize for worst-case tumor BED in $Z$, which is displayed in bold.\\
$^{\ast}$: the maximum OAR violation is computed over all patients and scenarios
\label{table: numexp-res-unif2}}
\end{table}

Due to the larger sampling space (the area of $Z_2$ is four times the area of $Z$), the difference between sample mean and sample worst-case performance is much larger than in \Cref{table: numexp-res-unif} for all methods. The true worst-case objective value in $Z$ is still lower than the sample worst-case in $Z_2$. The reason for this is that the true worst-case scenario can differ per patient. The relative performance of the adaptive methods remains mostly unchanged.
\begin{figure}[htb!]
\centering
\includegraphics[scale=0.8]{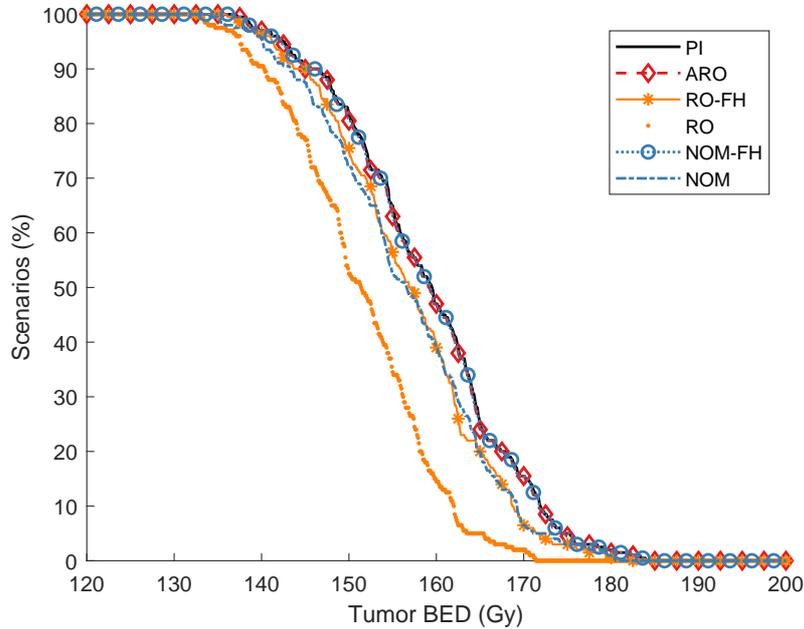}
\caption{\small Cumulative scenario-tumor BED graph for experiments with exact biomarker information and uniform sampling of $(\rho,\tau)$ over $Z_2$ (200 scenarios). A point $(x,y)$ indicates that in $y\%$ of scenarios the tumor BED (averaged over 20 patients) is at least $x$ Gy. ARO and NOM-FH are very close to PI. \label{fig: numexp-exact-outofsample}}
\end{figure}

NOM-FH and ARO are near optimal for the worst-case sample scenario, and are also close to PI in the sample $5\%$ quantile and sample mean. The relative performance of the adaptive methods remains mostly unchanged, RO-FH performs slightly worse than NOM-FH and ARO, similar to \Cref{table: numexp-res-unif}. Compared to \Cref{table: numexp-res-unif}, NOM and RO have poor performance across the sample. This indicates bad performance of the static methods on scenarios outside of $Z$. 

\Cref{fig: numexp-exact-outofsample} shows the complete cumulative scenario-tumor BED graph for the `average patient'. Compared to \Cref{fig: numexp-exact}, the main difference is the decrease in performance of NOM. Naturally, the performance of static nominal optimization is directly related to the magnitude of possible deviations from the nominal scenario, which is higher in $Z_2$ than in $Z$.

\setcounter{equation}{0}
\setcounter{figure}{0}
\setcounter{table}{0}
\section{Extra analyses and proofs} \label{app: proofs}
\small
For convenience, we repeat the definitions of functions $B$, $g$ and $f$:
\begin{subequations}
\small
\begin{flalign}
B(d',N';\rho) &:= \varphi D \Big( 1 + \frac{\varphi D}{T} \rho\Big) - \sigma d_1 N_1 - \sigma^2\rho d_1^2 N_1& \label{appeq: B}\\
g(d',N',N''; \rho) &:= \frac{-1 + \sqrt{1 + \frac{4\rho}{N''}B(d',N';\rho)}}{2\sigma \rho}& \label{appeq: g}\\
f(d_1,N_2 ; \rho,\tau) &:=
\begin{cases}
 N_1d_1+ N_2 g(d_1,N_1,N_2;\rho) + \tau \big(N_1d_1^2+N_2g(d_1,N_1,N_2;\rho)^2\big) & \text{ if } d_1 \in [0, g(0,0,N_1;\rho)]  \\
  -\infty & \text{ otherwise},
 \end{cases} & \label{appeq: f}
\end{flalign}
\end{subequations}
see \eqref{eq: B}, \eqref{eq: g} and \eqref{eq: f}.

\subsection{Proof \Cref{theorem: exact-stage2}} \label{app: proof-exact-stage2}
First, we show that for fixed $d_1$, feasible to \eqref{eq: exact}, and given $(\rho,\tau)$, it is optimal to minimize the number of stage-2 fractions if $\tau \geq \sigma \rho$, and it is optimal to maximize the number of stage-2 fractions otherwise. After that, we show that with stage-2 dose $d_2$ such that \eqref{eq: exact-2} holds with equality, $N_2(\rho,\tau)=N_2^{\text{min}}$ is feasible if $\tau \geq \sigma \rho$ and $N_2(\rho,\tau)=N_2^{\text{max}}$ is feasible otherwise.

Consider problem \eqref{eq: exact}. At the start of stage 2, we have delivered $N_1$ fractions with dose $d_1$ per fraction. Let $(\rho,\tau)$ be the realization of the uncertain parameters. The stage-2 problem reads
\begin{subequations}\label{appeq: exact-stage-2-problem}
\begin{align}
N_1d_1 + \tau N_1 d_1^2 + \max_{d_2,N_2} ~&~ N_2d_2 + \tau  N_2d_2^2 \label{appeq: exact-stage-2-problem1} \\
\text{s.t.}    ~&~ \sigma N_2d_2 + \rho \sigma^2 N_2d_2^2\leq B(d_1,N_1,\rho)  \label{appeq: exact-stage-2-problem2} \\
               ~&~ d_2 \geq d_{\text{min}} \\
               ~&~ N_2 \in \{N_2^{\text{min}},\dotsc,N_2^{\text{max}} \}.
\end{align}
\end{subequations}
This is a static fractionation problem. Constraint \eqref{appeq: exact-stage-2-problem2} will hold with equality at the optimum, because it is the only dose-limiting constraint. This yields
\begin{align} \label{appeq: d2-function}
d_2^{\ast}(d_1,N_2; \rho) = g(d_1,N_1,N_2; \rho).
\end{align}
Secondly, this allows us to rewrite the objective to
\begin{align}
\max_{d_2,N_2} ~&~ N_2d_2 \Big( \frac{\sigma \rho - \tau}{\sigma \rho} \Big) + \frac{\tau B(d_1,N_1,\rho)}{\sigma^2 \rho},
\end{align}
which implies that if $\tau > \sigma \rho$ it is optimal to minimize $d_2N_2$. If $\tau < \sigma \rho$ it is optimal to maximize $d_2N_2$, and if $\tau = \sigma \rho$ the objective value is independent of the value of $N_2$. Similar results are obtained in \citet{Mizuta12,Bortfeld15}. As given in \Cref{sec: exact}, at the optimum
\begin{align}
N_2d_2^{\ast}(d_1,N_2;\rho) = N_2g(d_1,N_1,N_2;\rho) = \frac{-N_2 + \sqrt{N_2^2 + 4 N_2 \rho B(d_1,N_1;\rho)}}{2\sigma \rho},
\end{align}
and it is straightforward to show that
\begin{align}
\frac{\partial N_2g(d_1,N_1,N_2;\rho) }{\partial N_2} \geq 0.
\end{align}
Hence, if $\tau > \sigma \rho$, it is optimal to minimize the number of fractions, and if $\tau < \sigma \rho$ it is optimal to maximize the number of fractions. If $\tau = \sigma \rho$, every feasible number of fractions is optimal.

For the second part, we must show that for any $(\rho,\tau) \in Z \cap \{\tau \geq \sigma\rho \}$ resp. $(\rho,\tau) \in Z\cap \{ \tau < \sigma\rho \}$, it is indeed possible to deliver $N_2^{\text{min}}$ resp. $N_2^{\text{max}}$ fractions with dose according to \eqref{appeq: d2-function} in stage 2. That is, we must show
\begin{subequations}
\begin{align}
& g(d_1,N_1,N_2^{\text{min}};\rho) \geq d^{\text{min}},~~\forall (\rho,\tau) \in  Z \cap \{\tau \geq \sigma\rho \} \\
& g(d_1,N_1,N_2^{\text{max}};\rho) \geq d^{\text{min}},~~\forall (\rho,\tau) \in Z\cap \{ \tau < \sigma\rho \},
\end{align}
\end{subequations}
which is equivalent to
\begin{subequations} \label{appeq: fractionation-feasible}
\begin{align}
& d_1 \leq g(d^{\text{min}},N_2^{\text{min}},N_1;\rho),~~\forall (\rho,\tau) \in Z \cap \{\tau \geq \sigma\rho \} \\
& d_1 \leq g(d^{\text{min}},N_2^{\text{max}},N_1;\rho),~~\forall (\rho,\tau) \in Z\cap \{ \tau < \sigma\rho \}.
\end{align}
\end{subequations}
\Cref{lemma: g-omega} states that $g$ is increasing or decreasing in $\rho$ for a fixed first argument. Hence, it is sufficient to consider only the largest and smallest value of $\rho$ in either subset of $Z$. Therefore, \eqref{appeq: fractionation-feasible} is equivalent to
\begin{subequations} \label{appeq: fractionation-feasible2}
\begin{align}
& d_1 \leq g(d^{\text{min}},N_2^{\text{min}},N_1;\rho_L) \\
& d_1 \leq g(d^{\text{min}},N_2^{\text{min}},N_1;\min\{\frac{\tau_U}{\sigma},\rho_U\})  \label{appeq: fractionation-feasible2-2}\\
& d_1 \leq g(d^{\text{min}},N_2^{\text{max}},N_1;\frac{\tau_L}{\sigma}) \\
& d_1 \leq g(d^{\text{min}},N_2^{\text{max}},N_1;\rho_U).
\end{align}
\end{subequations}
From \eqref{appeq: g} we see that function $g$ is decreasing in its second argument, so \eqref{appeq: fractionation-feasible2-2} is redundant. The remaining three conditions in \eqref{appeq: fractionation-feasible2} hold true due to \Cref{ass: exact-dminmax}. Hence, an optimal decision rule for $N_2(\cdot)$ is given by
\begin{align}
N_2^{\ast}(\rho,\tau) =
\begin{cases}
N_2^{\text{min}} & \text{ if } \tau \geq \sigma \rho \\
N_2^{\text{max}} & \text{ otherwise,}
\end{cases}
\end{align}
and 
\begin{align} 
d_2^{\ast}(d_1;\rho,\tau) &=
\begin{cases}
g(d_1,N_1,N_2^{\text{min}};\rho) & \text{ if } \tau \geq \sigma \rho  \\
g(d_1,N_1,N_2^{\text{max}};\rho) & \text{ otherwise}.
\end{cases}
\end{align}
are optimal decision rules for $N_2(\cdot)$ and $d_2(\cdot)$, respectively. For $\tau \neq \sigma \rho$, these give the unique optimal decisions. For $\tau = \sigma \rho$ any $N_2 \in \{N_2^{\text{min}},\dotsc,N_2^{\text{max}}\}$ is optimal, and the corresponding optimal $d_2$ follows according to \eqref{appeq: d2-function}.

\subsection{Proof \Cref{theorem: PARO}} \label{app: proof-PARO}
Due to \Cref{theorem: exact-stage2} a stage-1 decision $d_1$ is \PARO{} according to \Cref{def: PARO-x} if conditions \eqref{eq: PARO-conditions} hold with $(d_2^{\ast}(\cdot),N_2^{\ast}(\cdot))$ plugged in. Thus, we must show that for any $d_1 \in X^{\text{PARO}}$ there is no \ARO{} $\bar{d_1}$ such that
\begin{subequations}
\begin{align}
f(d_1, N_2^{\ast}(\rho,\tau); \rho, \tau) &\leq f(\bar{d_1}, N_2^{\ast}(\rho,\tau); \rho, \tau) ~~\forall (\rho,\tau) \in Z\\
f(d_1, N_2^{\ast}(\bar{\rho},\bar{\tau}); \bar{\rho}, \bar{\tau}) &< f(\bar{d_1}, N_2^{\ast}(\bar{\rho},\bar{\tau}); \bar{\rho}, \bar{\tau})~~\text{ for some } (\bar{\rho},\bar{\tau}) \in Z.
\end{align}
\end{subequations}
If $|X^{\text{PARO}}| = 1$, then the single element yields a strictly better objective value than all other elements in $X^{\text{ARO}}$ in either scenario $(\rho^{\text{aux-min}},\tau^{\text{aux-min}})$ or $(\rho^{\text{aux-max}},\tau^{\text{aux-max}})$ or both, so it is \PARO{}. For the remainder of this proof we assume $|X^{\text{PARO}}| \geq 2$.

Consider $X^{\text{aux-min}}$. By construction of $(\rho^{\text{aux-min}}, \tau^{\text{aux-min}})$ it holds that $\tau^{\text{aux-min}} \neq \sigma \rho^{\text{aux-min}}$. Hence, according to \Cref{lemma: twins}, there can be at most two values for $d_1$ in $X^{\text{aux-min}}$ that yield the same objective value $f$ in scenario $(\rho^{\text{aux-min}}, \tau^{\text{aux-min}})$. Hence, $|X^{\text{aux-min}}| = |X^{\text{PARO}} | = 2$. Denote the two elements of $X^{\text{PARO}}$ by $d_1'$ and $d_1''$, let $d_1' < d_1''$. Solutions $d_1'$ and $d_1''$ are both optimal to \eqref{eq: ARO-exact-aux-min} and \eqref{eq: ARO-exact-aux-max}. Hence, according to \Cref{lemma: twins}, it holds that
\begin{subequations}
\begin{align}
d_1'' & = t(d_1'; \rho^{\text{aux-min}},\tau^{\text{aux-min}}) \label{appeq: d''-min} \\
d_1'' & = t(d_1'; \rho^{\text{aux-max}},\tau^{\text{aux-max}}). \label{appeq: d''-max}
\end{align}
\end{subequations}
From the definition of $t$ (see \eqref{appeq: t}) we derive for $\sigma \rho \neq \tau$:
\begin{align}
\frac{\partial t(d_1; \rho,\tau)}{\partial \rho} = \frac{2 N_2^{\ast}(\rho,\tau)}{N_1 + N_2^{\ast}(\rho,\tau)} \frac{\partial g(d_1,N_1, N_2^{\ast}(\rho,\tau) ; \rho)}{\partial \rho},
\end{align}
because $N_2^{\ast}(\rho,\tau)$ is constant in $\rho$ unless $\sigma \rho = \tau$. According to \Cref{lemma: g-omega}, if for given $N_2$ it holds that $d_1 \neq d_1^{-}(N_2)$ and $d_1 \neq d_1^{+}(N_2)$ (defined in \eqref{appeq: d1-roots}), then function $g(d_1,N_1,N_2,\rho)$ is strictly increasing or decreasing in $\rho$. By construction, it holds that $d_1^{+}(N_2) = t(d_1^{-}(N_2);\rho,\tau)$ for any $\rho$. According to \Cref{lemma: intersection2}, we have $d_1^{-}(N_2^{\text{min}}) \neq d_1^{-}(N_2^{\text{max}})$, so $d_1'$ cannot be equal to both. Additionally, it cannot hold that $d_1' = d_1^{+}(N_2^{\text{min}})$ or $d_1' = d_1^{+}(N_2^{\text{max}})$, because it would imply $d'' \leq d'$. Hence, either $d_1' \notin \{d_1^{-}(N_2^{\text{min}}),d_1^{+}(N_2^{\text{min}})\}$ or $d_1' \notin \{d_1^{-}(N_2^{\text{max}}), d_1^{+}(N_2^{\text{max}})\}$ holds (or both).

We show that in either case, we can construct two new scenarios where $d_1'$ outperforms $d_1''$ in one scenario, and vice versa in the other. Suppose $d_1' \notin \{d_1^{-}(N_2^{\text{min}}),d_1^{+}(N_2^{\text{min}})\}$. In this case, it holds that
\begin{align} \label{appeq: partial-t-nonzero}
\frac{\partial t(d_1'; \rho^{\text{aux-min}},\tau^{\text{aux-min}})}{\partial \rho} \neq 0.
\end{align}
We consider two new scenarios. Let $\varepsilon >0$ be a sufficiently small number and define
\begin{subequations}
\begin{align}
(\rho_1,\tau_1) &:= (\rho^{\text{aux-min}} - \varepsilon,\tau^{\text{aux-min}}) ~~~~~~~~(\in \text{int}(Z^{\text{min}}))\\
(\rho_2,\tau_2) &:= (\rho^{\text{aux-min}} + \varepsilon,\tau^{\text{aux-min}}). ~~~~~~~(\in \text{int}(Z^{\text{min}}))
\end{align}
\end{subequations}
This is visualized in \Cref{fig: scenarios-epsilon}. Due to \eqref{appeq: partial-t-nonzero} and \eqref{appeq: d''-min}, it holds that
\begin{align} \label{appeq: clauses}
\big( t(d_1'; \rho_1,\tau_1) > d''  \land  t(d_1'; \rho_2,\tau_2) < d'' \big) \lor \big( t(d_1'; \rho_1,\tau_1) < d''  \land  t(d_1'; \rho_2,\tau_2) > d'' \big).
\end{align}
If the first clause is is true, we obtain
\begin{subequations}
\begin{align}
f(d_1',N_2^{\text{min}}; \rho_1,\tau_1) > f(d_1'',N_2^{\text{min}}; \rho_1,\tau_1) \\
f(d_1',N_2^{\text{min}}; \rho_2,\tau_2) < f(d_1'',N_2^{\text{min}}; \rho_2,\tau_2),
\end{align}
\end{subequations}
where we used convexity of $f(d_1,N_2^{\text{min}};\rho,\tau)$ for $(\rho,\tau) \in \text{int}(Z^{\text{min}})$. Similarly, if the second clause of \eqref{appeq: clauses} is true, we obtain
\begin{subequations}
\begin{align}
f(d_1',N_2^{\text{min}}; \rho_1,\tau_1) < f(d_1'',N_2^{\text{min}}; \rho_1,\tau_1) \\
f(d_1',N_2^{\text{min}}; \rho_2,\tau_2) > f(d_1'',N_2^{\text{min}}; \rho_2,\tau_2).
\end{align}
\end{subequations}
In either case, there is a scenario in $Z^{\text{min}}$ where $d_1'$ outperforms $d_1''$ and a scenario in $Z^{\text{min}}$ where $d_1''$ outperforms $d_1'$. Hence, both $d_1'$ and $d_1''$ are \PARO{}. Using similar arguments, we can show that in case $d_1' \notin \{d_1^{-}(N_2^{\text{max}}), d_1^{+}(N_2^{\text{max}})\}$ also both $d_1'$ and $d_1''$ are \PARO{}.
\begin{figure}
\centering
\includegraphics{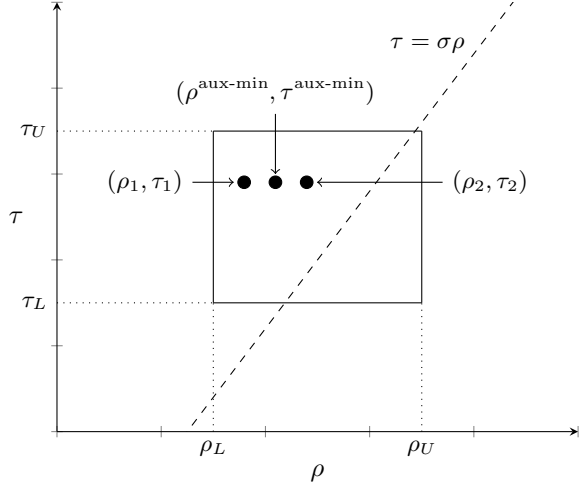}
\caption{\small Case $d_1' \neq d_1^{-}(N_2^{\text{min}})$. Construction of two new scenarios $(\rho_1,\tau_1)$ and $(\rho_2,\tau_2)$ from scenario $(\rho^{\text{aux-min}},\tau^{\text{aux-min}})$. Solution $d_1'$ outperforms $d_1''$ at one scenario, vice versa at the other. \label{fig: scenarios-epsilon}}
\end{figure}

\subsection{Proof \Cref{theorem: inexact-stage2}} \label{app: proof-inexact-stage2}
Consider problem \eqref{eq: inexact}. At the start of stage 2, we have delivered $N_1$ fractions with dose $d_1$ per fraction. Let $(\hat{\rho},\hat{\tau})$ be the observation. The resulting stage-2 problem for \eqref{eq: inexact} reads
\begin{subequations}\label{appeq: inexact-stage2p1}
\begin{align}
\max_{d_2,N_2} ~&~ \min_{(\rho,\tau) \in Z_{(\hat{\rho},\hat{\tau})}}~ (N_1d_1 + N_2d_2) + \tau ( N_1d_1^2 + N_2d_2^2) \label{appeq: inexact-stage2p1-1} \\
\text{s.t.}    ~&~ \sigma N_2d_2 + \rho \sigma^2 N_2d_2^2 \leq B(d_1,N_1,\rho)~~\forall (\rho,\tau) \in Z_{\hat{\rho},\hat{\tau}}  \label{appeq: inexact-stage2p1-2} \\
               ~&~ d_2 \geq d_{\text{min}} \label{appeq: inexact-stage2p1-3}\\
               ~&~ N_2 \in \{N_2^{\text{min}},\dotsc,N_2^{\text{max}} \}.
\end{align}
\end{subequations}
This is a static robust optimization problem. Constraint \eqref{appeq: inexact-stage2p1-2} will hold with equality at the optimum, because it is the only dose-limiting constraint. Solving for $d_2$ yields the constraint
\begin{align}
d_2 = g(d_1,N_1,N_2;\rho),~~\forall (\rho,\tau) \in Z_{(\hat{\rho},\hat{\tau})},
\end{align}
and this is used to rewrite \eqref{appeq: inexact-stage2p1-1} and \eqref{appeq: inexact-stage2p1-2} in terms of functions $f$ and $g$. Problem \eqref{appeq: inexact-stage2p1} is equivalent to
\begin{subequations} \label{appeq: inexact-stage2p2}
\begin{align}
\max_{N_2}  ~&~  \min_{(\rho,\tau) \in Z_{(\hat{\rho},\hat{\tau})}}~ f(d_1,N_2, \rho, \tau) \label{appeq: inexact-stage2p2-1} \\
\text{s.t.} ~&~ g(d_1,N_1,N_2;\rho) \geq d^{\text{min}}, ~~\forall (\rho,\tau) \in Z_{(\hat{\rho},\hat{\tau})} \label{appeq: inexact-stage2p2-2}\\
            ~&~ N_2 \in \{N_2^{\text{min}},\dotsc,N_2^{\text{max}} \}.
\end{align}
\end{subequations}
Similar to the exact case (\Cref{sec: exact}), in any worst-case realization it will hold that $\tau$ is at its lowest value, so it is sufficient to consider only those observations $(\rho,\tau) \in Z_{(\hat{\rho},\hat{\tau})}$ with $\tau = \hat{\tau}_L$. Additionally, according to \Cref{lemma: fg-omega} functions $f$ and $g$ are increasing or decreasing in $\rho$. Hence, there are two candidate worst-case scenarios: $(\hat{\rho}_L,\hat{\tau}_L)$ and $(\hat{\rho}_U,\hat{\tau}_L)$. We can rewrite \eqref{appeq: inexact-stage2p2} to
\begin{subequations} \label{appeq: inexact-stage2p3}
\begin{align}
\max_{N_2}  ~&~ \min \big\{ f(d_1,N_2, \hat{\rho}_L,\hat{\tau}_L),~f(d_1,N_2, \hat{\rho}_U,\hat{\tau}_L) \big\} \label{appeq: inexact-stage2p3-1} \\
\text{s.t.} ~&~ g(d_1,N_1,N_2;\hat{\rho}_L) \geq d^{\text{min}} \label{appeq: inexact-stage2p3-2}\\
            ~&~ g(d_1,N_1,N_2;\hat{\rho}_U) \geq d^{\text{min}} \label{appeq: inexact-stage2p3-3}\\
            ~&~ N_2 \in \{N_2^{\text{min}},\dotsc,N_2^{\text{max}} \}.
\end{align}
\end{subequations}
We distinguish three cases:
\begin{itemize}
\item Case $(\hat{\rho},\hat{\tau}) \in Z_{\text{ID}}^{\text{max}}$: Analogous to the proof of \Cref{theorem: exact-stage2}, one can show that for any realization $(\rho,\tau) \in Z_{(\hat{\rho},\hat{\tau})}$ it is optimal to maximize the number of fractions in stage 2. We plug in $N_2^{\ast}(\rho,\tau) = N_2^{\text{max}}$ and show that it is feasible. Constraints \eqref{appeq: inexact-stage2p3-2} and \eqref{appeq: inexact-stage2p3-3} reduce to
    \begin{align}
    \min \big\{ g(d_1,N_1,N_2^{\text{max}};\hat{\rho}_L), g(d_1,N_1,N_2^{\text{max}};\hat{\rho}_U) \big\} \geq d^{\text{min}},
    \end{align}
    which is equivalent to
    \begin{align} \label{appeq: inexact-d1bound1}
    d_1 \leq \min \big\{ g(d^{\text{min}}, N_2^{\text{max}}, N_1; \hat{\rho}_L), g(d^{\text{min}}, N_2^{\text{max}}, N_1; \hat{\rho}_U) \big\}.
    \end{align}
    It holds that $\hat{\rho}_L \geq \frac{\hat{\tau}_L}{\sigma} \geq \frac{\tau_L}{\sigma}$, and $\hat{\rho}_U \leq \rho_U$. According to \Cref{lemma: g-omega} function $g$ is either increasing or decreasing in $\rho$ for other arguments fixed. Hence, by \Cref{ass: inexact-dminmax} condition \eqref{appeq: inexact-d1bound1} holds. Hence, $N_2^{\ast}(\rho,\tau) = N_2^{\text{max}}$ is feasible and optimal. Thus, \eqref{appeq: inexact-stage2p2} equals
    \begin{align}
    \min \big\{ f(d_1,N_2^{\text{max}}, \hat{\rho}_L,\hat{\tau}_L),~f(d_1,N_2^{\text{max}}, \hat{\rho}_U,\hat{\tau}_L) \big\}.
    \end{align}
    By definition of $f$, this implies
    \begin{align}
    d_2 = \min \big\{ g(d_1,N_1,N_2^{\text{max}};\hat{\rho}_L),  g(d_1,N_1,N_2^{\text{max}};\hat{\rho}_U) \big\}.
    \end{align}

\item Case $(\hat{\rho},\hat{\tau}) \in Z_{\text{ID}}^{\text{min}}$: Similar to the previous case. Analogous to the proof of \Cref{theorem: exact-stage2}, one can show that for any realization $(\rho,\tau) \in Z_{(\hat{\rho},\hat{\tau})}$ it is optimal to minimize the number of fractions in stage 2. We plug in $N_2^{\ast}(\rho,\tau) = N_2^{\text{min}}$ and show that it is feasible. Similar to the previous case, constraints \eqref{appeq: inexact-stage2p3-2} and \eqref{appeq: inexact-stage2p3-3} reduce to
    \begin{align} \label{appeq: inexact-d1bound2}
    d_1 \leq \min \big\{ g(d^{\text{min}}, N_2^{\text{min}}, N_1; \hat{\rho}_L), g(d^{\text{min}}, N_2^{\text{min}}, N_1; \hat{\rho}_U) \big\}.
    \end{align}
    It holds that $\hat{\rho}_L \geq \rho_L$, and $\hat{\rho}_U \leq \rho_U$. Hence, by \Cref{ass: inexact-dminmax}, \Cref{lemma: g-omega} and using the fact that function $g$ is decreasing in its second argument, condition \eqref{appeq: inexact-d1bound2} holds. Hence, $N_2^{\ast}(\rho,\tau) = N_2^{\text{min}}$ is feasible and optimal. Similar to the previous case, we find
    \begin{align}
    d_2 = \min \big\{ g(d_1,N_1,N_2^{\text{min}};\hat{\rho}_L),  g(d_1,N_1,N_2^{\text{min}};\hat{\rho}_U) \big\}.
    \end{align}

\item Case $(\hat{\rho},\hat{\tau}) \in Z_{\text{ID}}^{\text{int}}$: The optimal number of fractions in stage-2 is not known a priori. By definition of $Z_{\text{ID}}^{\text{int}}$, it holds that $\hat{\rho}_L \geq \max\{\rho_L,\frac{\tau_L}{\sigma} - 2r^{\rho})\}$ and $\hat{\rho}_U \leq \rho_U$. By \Cref{ass: inexact-dminmax} it holds that
    \begin{align} \label{appeq: assumption-int}
        d_1 \leq \min  \big\{ g(d^{\text{min}}, N_2^{\text{max}}, N_1; \max\{\rho_L,\frac{\tau_L}{\sigma} - 2r^{\rho})\}, g(d^{\text{min}}, N_2^{\text{max}}, N_1; \rho_U) \big\}.
    \end{align}
    \Cref{lemma: g-omega}, the fact that function $g$ is decreasing in its third argument and  \eqref{appeq: assumption-int} together imply that \eqref{appeq: inexact-stage2p3-2} and \eqref{appeq: inexact-stage2p3-3} hold for any feasible $N_2$. Hence, from problem \eqref{appeq: inexact-stage2p3} we derive
    \begin{align}
    N_2^{\ast}(d_1; \hat{\rho},\hat{\tau}) = \argmax_{N_2 \in \{N_2^{\text{min}},\dotsc,N_2^{\text{max}} \}} \min \big\{ f(d_1,N_2, \hat{\rho}_L,\hat{\tau}_L),~f(d_1,N_2, \hat{\rho}_U,\hat{\tau}_L) \big\},
    \end{align}
    and by definition of $f$ the corresponding value for $d_2$ is
    \begin{align}
    d_2 = \min \big\{ g(d_1,N_1,N_2^{\ast}(d_1; \hat{\rho},\hat{\tau});\hat{\rho}_L),  g(d_1,N_1,N_2^{\ast}(d_1; \hat{\rho},\hat{\tau});\hat{\rho}_U) \big\}.
    \end{align}
\end{itemize}
Combining the above three cases, we arrive at the optimal decision rules \eqref{eq: inexact-stage2-N2} and \eqref{eq: inexact-stage2-d2} for fixed $d_1$.

\subsection{Extra analysis to \Cref{sec: inexact}} \label{app: proof-inexact-reformulation}
This analysis makes use of the lemmas in \Cref{app: auxiliary-lemmas}. Consider problem \eqref{eq: inexact}. For given $d_1$, the optimal stage-2 decision rules are given by \Cref{theorem: inexact-stage2}. As stated in \Cref{sec: inexact}, we split the uncertainty set $Z$ into three subsets. This enables us to exploit the fact that depending on $(\hat{\rho},\hat{\tau})$ the value $N_2^{\ast}(d_1; \hat{\rho},\hat{\tau})$ may be known in advance. The split \eqref{eq: def-Zi} is repeated here for convenience
\begin{subequations} \label{appeq: def-Zi}
\begin{align}
Z_{\text{ID}}^{\text{min}} & = \{(\hat{\rho},\hat{\tau}) \in Z :  \hat{\tau}_L \geq \sigma \hat{\rho}_U  \} \label{appeq: def-Z1}\\
Z_{\text{ID}}^{\text{int}} & = \{(\hat{\rho},\hat{\tau}) \in Z :  \sigma \hat{\rho}_L < \hat{\tau}_L < \sigma \hat{\rho}_U \} \label{appeq: def-Z2}\\
Z_{\text{ID}}^{\text{max}} & = \{(\hat{\rho},\hat{\tau}) \in Z : \hat{\tau}_L \leq \sigma \hat{\rho}_L  \}, \label{appeq: def-Z3}
\end{align}
\end{subequations}
so that $Z = Z_{\text{ID}}^{\text{min}} \cup Z_{\text{ID}}^{\text{int}} \cup Z_{\text{ID}}^{\text{max}}$. The associated sets of observation-realization pairs $(\rho,\tau,\hat{\rho},\hat{\tau})$ are given by
\begin{align}
U^i &= U \cap \{(\rho,\tau,\hat{\rho},\hat{\tau}): (\hat{\rho},\hat{\tau}) \in Z_{\text{ID}}^i \},~~i\in \{\text{min},\text{int},\text{max} \},
\end{align}
so it holds that $U = U^{\text{min}} \cup U^{\text{int}} \cup U^{\text{max}}$. Set $U^i$ can be interpreted as the set of observation-realization pairs for which the observation $(\hat{\rho},\hat{\tau})$ is in set $Z_{\text{ID}}^i$. \Cref{fig: Z-split} illustrates the subsets $Z_{\text{ID}}^i$. Set $U^{\text{min}}$ consists of those observation-realization pairs $(\rho,\tau,\hat{\rho},\hat{\tau})$ for which $N_2^{\ast}(d_1; \hat{\rho},\hat{\tau}) = N_2^{\text{max}}$. If $(\rho,\tau,\hat{\rho},\hat{\tau}) \in U^{\text{int}}$, then based on the observation $(\hat{\rho},\hat{\tau})$ it is unclear what fractionation is worst-case optimal. Last, if $(\rho,\tau,\hat{\rho},\hat{\tau}) \in U^{\text{max}}$ we know $N_2^{\ast}(d_1; \hat{\rho},\hat{\tau}) = N_2^{\text{min}}$. Problem \eqref{eq: inexact} is equivalent to
\begin{subequations}\label{appeq: inexact2}
\begin{align}
\max_{d_1,q}~&~ q \\
\text{s.t.}~&~ q \leq f(d_1, N_2^{\text{min}};\rho,\tau), ~~\forall (\rho,\tau,\hat{\rho},\hat{\tau}) \in U^{\text{min}} \label{appeq: inexact2-1} \\
            ~&~ q \leq f(d_1, N_2^{\ast}(d_1;\hat{\rho},\hat{\tau});\rho,\tau), ~~\forall (\rho,\tau,\hat{\rho},\hat{\tau}) \in U^{\text{int}} \label{appeq: inexact2-2} \\
            ~&~ q \leq f(d_1, N_2^{\text{max}};\rho,\tau), ~~\forall (\rho,\tau,\hat{\rho},\hat{\tau}) \in U^{\text{max}} \label{appeq: inexact2-3} \\
            ~&~ d^{\text{min}} \leq d_1 \leq d_1^{\text{max}}. \label{appeq: inexact2-4}
\end{align}
\end{subequations}
Similar to the exact case (\Cref{sec: exact}), in any worst-case realization it will hold that
$\tau = \tau_L$. Therefore, any observation with $\hat{\tau} - r^{\tau} > \tau_L$ cannot yield the worst-case realization.
Define
\begin{align}
U_L^i = U_i \cap \{(\rho,\tau,\hat{\rho},\hat{\tau}): \hat{\tau} - r^{\tau} \leq \tau_L \},~~i\in \{\text{min},\text{int},\text{max} \},
\end{align}
which is the subset of $U^i$ of observation-realization pairs for which $\tau_L$ is a possible realization of $\tau$. Constraints \eqref{appeq: inexact2-1}-\eqref{appeq: inexact2-4} can be replaced by 
\begin{subequations}
\begin{align}
q &\leq f(d_1, N_2^{\text{min}};\rho,\tau), ~~\forall (\rho,\tau,\hat{\rho},\hat{\tau}) \in U_L^{\text{min}} \label{appeq: inexactL-1} \\
            q &\leq f(d_1, N_2^{\ast}(d_1;\hat{\rho},\hat{\tau});\rho,\tau), ~~\forall (\rho,\tau,\hat{\rho},\hat{\tau}) \in U_L^{\text{int}} \label{appeq: inexactL-2} \\
q &\leq f(d_1, N_2^{\text{max}};\rho,\tau), ~~\forall (\rho,\tau,\hat{\rho},\hat{\tau}) \in U_L^{\text{max}} \label{appeq: inexactL-3}
\end{align}
\end{subequations}
For \eqref{appeq: inexactL-1} and \eqref{appeq: inexactL-3} it remains to find the worst-case realization of $\rho$ for which the observation-realization pair is in $U_L^{\text{min}}$ and $U_L^{\text{max}}$, respectively. According to \Cref{lemma: f-omega}, function $f$ is increasing or decreasing in $\rho$ for fixed $d_1$, so it is sufficient to check the maximum and minimum realization of $\rho$ for which the observation-realization pair is in those sets. These are
\begin{subequations}
\begin{align}
 &\min \{\rho : (\rho,\tau,\hat{\rho},\hat{\tau}) \in U_L^{\text{min}} \} = \rho_L,~~ \max \{\rho : (\rho,\tau,\hat{\rho},\hat{\tau}) \in U_L^{\text{min}} \} = \frac{\tau_L}{\sigma} \\
 &\min \{\rho : (\rho,\tau,\hat{\rho},\hat{\tau}) \in U_L^{\text{max}} \} = \frac{\tau_L}{\sigma},~~ \max \{\rho : (\rho,\tau,\hat{\rho},\hat{\tau}) \in U_L^{\text{max}} \} = \rho_U.
\end{align}
\end{subequations}
Plugging in $\rho = \frac{\tau_L}{\sigma}$ in \eqref{appeq: inexactL-1} and \eqref{appeq: inexactL-3} yields $q \leq K$, with $K$ defined by \eqref{eq: K}. \Cref{lemma: i=2-reform} provides a conservative approximation of constraint \eqref{appeq: inexactL-2}. Putting everything together, the optimum of the following problem is a lower bound to the optimum of \eqref{appeq: inexact2} (or, equivalently, \eqref{eq: inexact}):
\begin{subequations}\label{appeq: inexact3}
\begin{align}
\max_{d_1,q}~&~ q \\
\text{s.t.}~&~ q \leq f(d_1,N_2^{\text{min}};\rho_L,\tau_L) \label{appeq: inexact3-2} \\
            ~&~ q \leq f(d_1,N_2^{\text{max}};\rho_U,\tau_L) \label{appeq: inexact3-3} \\
            ~&~ q \leq K \label{appeq: inexact3-4} \\
~&~ q \leq p(d_1) \label{appeq: inexact3-5} \\
~&~ d_{\text{min}} \leq d_1 \leq d_1^{\text{max}}, \label{appeq: inexact3-6}
\end{align}
\end{subequations}
with $p(d_1)$ defined by \eqref{appeq: p} in \Cref{app: auxiliary-lemmas}. Constraint \eqref{appeq: inexact3-5} is the only conservative constraint, all other constraints are exact reformulations. In particular, this means that if for a solution the objective value equals $K$, it is certain that this is an optimal solution. It is easy to obtain other straightforward conservative approximations of \eqref{appeq: inexact2-2}. For instance, a policy that delivers $N_2^{\text{min}}$ or $N_2^{\text{max}}$ fractions (or any number in between, for that matter) for any observation $(\hat{\rho},\hat{\tau}) \in Z_{\text{ID}}^{\text{int}}$ is a conservative approximation. However, these perform less good and do not use all available information, as explained in the proof of \Cref{lemma: i=2-reform}.

\setcounter{equation}{0}
\setcounter{figure}{0}
\setcounter{table}{0}
\section{Extra lemmas} \label{app: auxiliary-lemmas}
\small
This appendix states and proves several frequently used properties of functions $g$ and $f$.
\begin{lemma}[Convexity/concavity $f$ w.r.t. $d_1$] \label{lemma: convexity-f}
Let $\rho>0$, $\tau>0$ and $N_1,N_2 \in \mathbb{N}_+$. Let $d_1 \in [0,g(0,0,N_1;\rho)]$. The following properties hold for function $f$:
\begin{itemize}
\item Function $f(d_1,N_2;\rho,\tau)$ is strictly convex in $d_1$ if $\tau  > \rho\sigma$, with unique minimizer $g(0,0,N_1+N_2;\rho)$;
\item Function $f(d_1,N_2;\rho,\tau)$ is strictly concave in $d_1$ if $\tau  < \rho\sigma$, with unique maximizer $g(0,0,N_1+N_2;\rho)$;
\item Function $f(d_1,N_2;\rho,\tau)$ is constant in $d_1$ if $\tau = \rho\sigma$, with value $\frac{1}{\sigma}B(0,0,\frac{\tau}{\sigma})$.
\end{itemize}
\end{lemma}

\begin{proof}
The partial derivative of $f$ w.r.t. $d_1$ is given by \small
\begin{align} \label{appeq: f-partial-d1}
\hspace*{-0.2cm}
\frac{\partial f(d_1,N_2;\rho,\tau)}{\partial d_1} =  N_1 + N_2\frac{\partial g(d_1,N_1,N_2;\rho)}{\partial d_1} + \tau\Big(2N_1d_1 + 2N_2g(d_1,N_1,N_2;\rho)\frac{\partial g(d_1,N_1,N_2;\rho)}{\partial d_1}\Big),
\end{align} 
where the partial derivative of $g$ w.r.t. $d_1$ is given by
\begin{align} \label{appeq: g-partial-d1}
\frac{\partial g(d_1,N_1,N_2;\rho)}{\partial d_1} = -\frac{1}{N_2}(N_1+2N_1d_1\sigma\rho) \Big(1+\frac{4 \rho}{N_2} B(d_1,N_1;\rho)\Big)^{-\frac{1}{2}}.
\end{align}
Define $h(d_1,N_2;\rho) = 1 + 4 \frac{\rho}{N_2}B(d_1,N_1;\rho)$. Then, plugging \eqref{appeq: g-partial-d1} in \eqref{appeq: f-partial-d1}, we obtain
\begin{align*}
\frac{\partial f(d_1,N_2;\rho,\tau)}{\partial d_1} = ~&~ (N_1-(N_1+2N_1d_1\sigma\rho) h(d_1,N_2;\rho)^{-\frac{1}{2}} \\
~&~ + \tau \Big(2N_1d_1 - \frac{2}{N_2}(N_1+2N_1d_1\sigma\rho) h(d_1,N_2;\rho)^{-\frac{1}{2}} N_2\frac{-1 + h(d_1,N_2;\rho)^{\frac{1}{2}}}{2\sigma\rho } \Big)\\[0.5em]
= ~&~ \frac{ N_1}{\sigma\rho} \big(h(d_1,N_2;\rho)^{-\frac{1}{2}}(2\sigma\rho d_1+1)-1\big)(\tau -\rho\sigma).
\end{align*}
Further elementary math shows that $h(d_1,N_2;\rho)^{-\frac{1}{2}}(2\sigma\rho d_1+1) -1 = 0$ if and only if $d_1=g(0,0,N_1+N_2;\rho)$. For the second derivative of $f$ w.r.t. $d_1$ we obtain:
\begin{align*}
\frac{\partial^2 f(d_1,N_2;\rho,\tau)}{\partial d_1^2} ~&~= \Big(\frac{\tau -\rho\sigma}{\sigma \rho}N_1 \Big) \frac{\partial}{\partial d_1} h(d_1,N_2;\rho)^{-\frac{1}{2}}(2\sigma \rho d_1+1) \\
=~&~\Big(\frac{\tau -\rho\sigma}{\sigma \rho} N_1 \Big) \Bigg[h(d_1,N_2;\rho)^{-\frac{1}{2}}2\sigma\rho + \frac{2\rho}{N_2} (2\sigma \rho d_1 + 1) h(d_1,N_2,\rho)^{-\frac{3}{2}}(\sigma N_1 + 2\rho\sigma^2d_1N_1)   \Bigg],
\end{align*}
and the second part of this product is positive. Hence, its sign depends only on the term $\tau - \rho\sigma$. Combining the result for the first and second derivative, we obtain
\begin{itemize}
\item Function $f(d_1,N_2;\rho,\tau)$ is strictly convex in $d_1$ if $\tau > \rho\sigma$, with unique minimizer $g(0,0,N_1+N_2;\rho)$;
\item Function $f(d_1,N_2;\rho,\tau)$ is strictly concave in $d_1$ if $\tau < \rho\sigma$, with unique maximizer $g(0,0,N_1+N_2;\rho)$;
\item Function $f(d_1,N_2;\rho,\tau)$ is constant in $d_1$ otherwise.
\end{itemize}
If $\tau = \rho \sigma$, we can rewrite $f(d_1,N_2; \frac{\tau}{\sigma},\tau)$ to
\begin{align*}
f(d_1,N_2;\frac{\tau}{\sigma},\tau) = ~& \max_{d_2} \Big\{ d_1N_1 + d_2N_2 + \tau(d_1^2N_1 + d_2^2 N_2)~|~
\sigma (d_1N_1 + d_2N_2) + \rho \sigma^2(d_1^2N_1 + d_2^2 N_2) \leq B(0,0;\rho) \Big\}\\[0.5em]
=~& \max_{d_2} \Big\{d_1N_1 + d_2N_2 + \tau(d_1^2N_1 + d_2^2 N_2) ~|~
\sigma (d_1N_1 + d_2N_2) + \rho \sigma^2(d_1^2N_1 + d_2^2 N_2) = B(0,0;\rho) \Big\} \\[0.5em]
		= ~& \max_{d_2}\Big\{  d_1N_1 + d_2N_2 + \tau(d_1^2N_1 + d_2^2 N_2) ~|~
d_1N_1 + d_2N_2 + \tau(d_1^2N_1 + d_2^2 N_2) = \frac{1}{\sigma}B(0,0,\frac{\tau}{\sigma}) \Big\}\\[0.5em]		
					=~&~ \frac{1}{\sigma}B(0,0,\frac{\tau}{\sigma}).
\end{align*}

\end{proof}

Define
\begin{subequations}\label{appeq: d1-roots}
\begin{align}
d_1^{-}(N_2) =
\begin{cases}
\displaystyle \frac{\varphi D - \varphi D \Big(1 + (N_1+N_2) \frac{N_2-T}{N_1 T} \Big)^{\frac{1}{2}}}{\sigma(N_1+N_2)}  & \text{ if } N_1 + N_2 \geq T \land N_1 \leq T \\
 - \infty  & \text{ otherwise }
\end{cases}\label{appeq: d1-rootA}\\
d_1^{+}(N_2) =
\begin{cases}
\displaystyle  \frac{\varphi D + \varphi D \Big(1 + (N_1+N_2) \frac{N_2-T}{N_1 T} \Big)^{\frac{1}{2}}}{\sigma(N_1+N_2)} & \text{ if } N_1 + N_2 \geq T \land N_2 \leq T \\
+ \infty & \text{ otherwise. }
\end{cases} \label{appeq: d1-rootB}
\end{align}
\end{subequations}
If two functions $f$ with equal $N_2$ but different $\rho$ intersect, $d_1$ takes value $d_1^{-}(N_2)$ or $d_1^{+}(N_2)$. The following lemma provides information on the existence and location of these intersection points. We consider only those values for $d_1$ where function $f(d_1,N_2; \rho,\tau)$ is finite for all $(\rho,\tau)\in Z$. Let
\begin{align}
d_{\text{UB}} = \min_{(\rho,\tau) \in Z} g(0,0,N_1;\rho).
\end{align}

\begin{lemma}[Properties $d_1^{-}$ and $d_1^{+}$] \label{lemma: intersection}
Let $N_1, T \in \mathbb{N}_+$.
\begin{lemlist}
\item Let $N_2 \in \mathbb{N}_+$. If $\rho_1 \neq \rho_2$, the equation
\begin{align} \label{appeq: f-equality}
f(d_1,N_2; \rho_1,\tau) = f(d_1,N_2; \rho_2,\tau),
\end{align}
has the following real roots for $d_1$ on the interval $[0, d_{\text{UB}} ]$:
\begin{subequations} \label{appeq: summ-roots}
\begin{align}
&\bullet \qquad d_1^{-}(N_2) \text{ and } d_1^{+}(N_2) &\text{ if } N_1 + N_2 \geq T,~N_2 \leq T \text{ and } N_1 \leq T \label{appeq: summ-roots1}\\
&\bullet \qquad d_1^{-}(N_2) &\text{ if } N_1 + N_2 \geq T,~ N_2 \leq T \text{ and } N_1 > T \label{appeq: summ-roots2}\\
&\bullet \qquad d_1^{+}(N_2) &\text{ if } N_1 + N_2 \geq T ,~ N_2 > T \text{ and } N_1 \leq T \label{appeq: summ-roots3}\\
&\bullet \qquad \text{no roots on interval} & \text{ if } N_1 + N_2 \geq T ,~~ N_2 > T \text{ and } N_1 > T\\
&\bullet \qquad \text{no real roots } & \text{ otherwise}.\label{appeq: summ-roots4}
\end{align}
\end{subequations} \label{lemma: intersection1}
\item Let $N_2^A,N_2^B \in \mathbb{N}_+$ such that $N_2^A < N_2^B$. It holds that
\begin{enumerate}[label=(\emph{\roman*})]
\item If $d_1^{-}(N_2^A)$ and $d_1^{-}(N_2^B)$ are both finite, then $d_1^{-}(N_2^A) > d_1^{-}(N_2^B)$;
\item If $d_1^{+}(N_2^A)$ and $d_1^{+}(N_2^B)$ are both finite, then $d_1^{+}(N_2^A) \leq d_1^{+}(N_2^B)$.
\end{enumerate} \label{lemma: intersection2}
\end{lemlist}
\end{lemma}

\begin{proof}
Both parts of the lemma are proved individually.\\
\noindent \emph{Proof \Cref{lemma: intersection1}}\\
By definition of $f$, the equation $f(d_1,N_2;\rho_1,\tau) = f(d_1,N_2;\rho_2,\tau)$ reduces to $g(d_1,N_1,N_2;\rho_1) = g(d_1,N_1,N_2;\rho_2)$ with $d_1 \in [0, \min\{g(0,0,N_1;\rho_1),g(0,0,N_1;\rho_2)\}]$. By construction of $g$, this means we are interested in the pairs $(d_1,d_2)$ that solve the system
\begin{subequations} \label{appeq: BED-cond}
\begin{align}
& \sigma (N_1d_1 + N_2 d_2) + \rho_1 \sigma^2(N_1d_1^2 + N_2 d_2^2) = \varphi D (1 + \rho_1 \frac{D}{T}\varphi) \label{appeq: BED-cond-1}\\
& \sigma (N_1d_1 + N_2 d_2) + \rho_2 \sigma^2(N_1d_1^2 + N_2 d_2^2) = \varphi D (1 + \rho_2 \frac{D}{T}\varphi)\label{appeq: BED-cond-2}\\
& d_1\geq 0,~d_2 \geq 0.
\end{align}
\end{subequations}
We subtract $\frac{\rho_2}{\rho_1}$ times \eqref{appeq: BED-cond-1} from \eqref{appeq: BED-cond-2} and solve for $d_1$ to obtain
\begin{align} \label{appeq: cond-d1}
d_1 = \frac{\varphi D - \sigma N_2 d_2}{\sigma N_1}.
\end{align}
We know that $d_2 = g(d_1,N_1,N_2;\rho_1)$. Plugging in \eqref{appeq: cond-d1} in this expression and simplifying yields the following roots for $d_2$:
\begin{subequations}\label{appeq: r2-roots}
\begin{align}
r_2^{-}(N_2) = \frac{\varphi D + \varphi D \Big(1 + (N_1+N_2) \big(\frac{N_1}{T} -1\big)/N_2 \Big)^{\frac{1}{2}}}{\sigma(N_1+N_2)} \\
r_2^{+}(N_2) = \frac{\varphi D - \varphi D \Big(1 + (N_1+N_2) \big(\frac{N_1}{T} -1\big)/N_2 \Big)^{\frac{1}{2}}}{\sigma(N_1+N_2)}.
\end{align}
\end{subequations}
Plugging \eqref{appeq: r2-roots} in \eqref{appeq: cond-d1} and simplifying yields the following roots for $d_1$:
\begin{subequations}\label{appeq: r1-roots}
\begin{align}
r_1^{-}(N_2) = \frac{\varphi D - \varphi D \Big(1 + (N_1+N_2) \frac{N_2-T}{N_1 T} \Big)^{\frac{1}{2}}}{\sigma(N_1+N_2)} \label{appeq: r1-rootA}\\
r_1^{+}(N_2) = \frac{\varphi D + \varphi D \Big(1 + (N_1+N_2) \frac{N_2-T}{N_1 T} \Big)^{\frac{1}{2}}}{\sigma(N_1+N_2)}. \label{appeq: r1-rootB}
\end{align}
\end{subequations}
These roots need not be real-valued, nor in the interval $[0, \min\{g(0,0,N_1;\rho_1),g(0,0,N_1;\rho_2)\}]$. For both $r_1^{-}(N_2)$ and $r_1^{+}(N_2)$ to be real-valued, we require that
\begin{align*}
&1 + (N_1+N_2) \big(\frac{N_1}{T} -1\big)/N_2 \geq 0,
\end{align*}
which reduces to $N_1 + N_2 \geq T$. Furthermore, for nonnegativity of $r_1^{-}(N_2)$ and $r_1^{+}(N_2)$ it suffices to check nonnegativity of the former. This is equivalent to
\begin{align*}
\varphi D - \sigma N_2 d_2^{-} \geq 0,
\end{align*}
which reduces to $N_2\leq T$. Moreover, it needs to hold that $r_1^{+}(N_2) \leq g(0,0,N_1;\rho_1)$ and $r_1^{+}(N_2) \leq g(0,0,N_1;\rho_2)$. This is equivalent to $r_2^{+}(N_2) \geq 0$, which can be rewritten to
\begin{align*}
1 + (N_1+N_2) \big(\frac{N_1}{T} -1\big)/N_2 \leq 1,
\end{align*}
and this reduces to $N_1\leq T$. Parameters $d_1^{-}(N_2)$ resp. $d_1^{+}(N_2)$ (see \eqref{appeq: d1-roots}) take the values of $r_1^{-}(N_2)$ resp. $r_1^{-}(N_2)$ if they are a root of \eqref{appeq: f-equality}, and $-\infty$ resp. $+\infty$ otherwise. All together, we obtain the cases in \eqref{appeq: summ-roots}.

It remains to show that the obtained roots are in the interval $[0, d_{\text{UB}}]$. It is already shown that, if they are (real-valued) roots to \eqref{appeq: f-equality}, then $d_1^{-}(N_2), d_1^{+}(N_2)$ are nonnegative. Furthermore, in that case $d_1^{-}(N_2) \leq d_1^{+}(N_2)$. It holds that
\begin{align*}
\frac{\partial g(0,0,N_1;\rho)}{\partial \rho} \leq 0 \Leftrightarrow N_1 \leq T.
\end{align*}
Hence, if $d_1^{+}(N_2)$ is a real-valued root to \eqref{appeq: f-equality} it follows that
\begin{align*}
d_{\text{UB}} = \min_{(\rho,\tau) \in Z} g(0,0,N_1;\rho) \geq \lim_{\rho \rightarrow +\infty} g(0,0,N_1;\rho) = \frac{\varphi D}{\sigma \sqrt{N_1 T}} \geq d_1^{+}(N_2),
\end{align*}
where the second equality follows from the definition of $g$. This implies that indeed $d_1^{-}(N_2), d_1^{+}(N_2) \in [0,d_{\text{UB}}]$.\\

\noindent \emph{Proof \Cref{lemma: intersection2}}\\
Assume $N_2^A,N_2^B \in \mathbb{N}_+$ such that $N_2^A \leq N_2^B$, and assume $N_1 + N_2^A \geq T$. Statements $(i)$ and $(ii)$ are proved individually.\\

\noindent \emph{Proof part (i)}\\
Assume $d_1^{-}(N_2^A)$ and $d_1^{-}(N_2^B)$ are both finite. The denominator of $d_1^{-}(N_2)$ (see \eqref{appeq: d1-rootA}) is increasing in $N_2$. The derivative (w.r.t. $N_2$) of the part within the square root in the numerator of \eqref{appeq: d1-rootA} is given by $(N_1T)^{-1}(N_1+2N_2 - T) \geq 0$, because $N_1 +N_2 \geq T$. Hence, the numerator is decreasing in $N_2$, while the denominator is increasing in $N_2$. This implies $d_1^{-}(N_2^A) > d_1^{-}(N_2^B)$.\\

\noindent \emph{Proof part (ii)}\\
Assume $d_1^{+}(N_2^A)$ and $d_1^{+}(N_2^B)$ are both finite. One can show that
\begin{align}
\frac{\partial d_1^{+}(N_2)}{\partial N_2} = \varphi D \frac{(N_1+N_2)(N_1+2N_2-T) - 2N_1T \sqrt{\frac{N_2(N_1+N_2-T)}{N_1T}} - 2N_2(N_1+N_2-T)}{2N_1T \sigma (N_1+N_2)^2 \sqrt{\frac{N_2(N_1+N_2-T)}{N_1T}}}.
\end{align}
This implies
\begin{align}
& \frac{\partial d_1^{+}(N_2)}{\partial N_2} \geq 0 \nonumber\\
\Leftrightarrow~ & (N_1+N_2)(N_1+2N_2-T) - 2N_1T \sqrt{\frac{N_2(N_1+N_2-T)}{N_1T}} - 2N_2(N_1+N_2-T) \geq 0 \nonumber\\
\Leftrightarrow~ & \frac{N_1(N_1+N_2-T) + TN_2}{2N_1T} \geq \sqrt{\frac{N_2(N_1+N_2-T)}{N_1T}} \nonumber \\
\Leftrightarrow~ & \frac{\big(N_1(N_1+N_2-T) + TN_2 \big)^2}{4N_1^2T^2} \geq \frac{N_2(N_1+N_2-T)}{N_1T} \nonumber \\
\Leftrightarrow~ & \frac{(N_1-T)^2(N_1+N_2)^2}{4N_1^2T^2} \geq 0,\label{appeq: d1B-deriv-eta}
\end{align}
where the fourth line is obtained by using the fact that $N_1 + N_2 \geq T$, and squaring on both sides. The last line follows from simple algebraic manipulations. Condition \eqref{appeq: d1B-deriv-eta} clearly holds, so $d_1^{+}(N_2^A) \leq d_1^{+}(N_2^B)$.
\end{proof}

\begin{lemma}[Derivative $f$ and $g$ w.r.t. $\rho$] \label{lemma: fg-omega}
Let $(\rho,\tau) \in Z$.
\begin{lemlist}
\item Let $N',N'' \in \mathbb{N}_+$. Let $d' \in [0,d_{\text{UB}}]$. If $N' + N'' < T$, then
\begin{align}
\frac{\partial g(d',N',N'';\rho)}{\partial \rho} < 0~~\text{ for all } d' \in [0,d_{\text{UB}}].
\end{align}
If $N' + N'' \geq T$, then
\begin{align} \label{eq: g-omega}
\frac{\partial g(d',N',N'';\rho)}{\partial \rho}
\begin{cases}
< 0 & \text{ if }~ d' \in [0,d_1^{-}(N'')) \cup (d_1^{+}(N''), d_{\text{UB}}]\\
= 0 & \text{ if }~ d' \in [0,d_{\text{UB}}] \cap \{d_1^{-}(N''), d_1^{+}(N'')\} \\
> 0 & \text{ if }~ d' \in [0,d_{\text{UB}}] \cap (d_1^{-}(N''), d_1^{+}(N'')). \\
\end{cases}
\end{align}\label{lemma: g-omega}

\item Let $N_1,N_2 \in \mathbb{N}_+$. Let $d_1 \in [0,d_{\text{UB}}]$. If $N_1 + N_2 < T$, then
\begin{align}
\frac{\partial f(d_1,N_2;\rho,\tau)}{\partial \rho} < 0~~\text{ for all } d_1 \in [0,d_{\text{UB}}].
\end{align}
If $N_1 + N_2 \geq T$, then
\begin{align}
\frac{\partial f(d_1,N_2;\rho,\tau)}{\partial \rho}
\begin{cases}
< 0 & \text{ if }~ d_1 \in [0,d_1^{-}(N_2)) \cup (d_1^{+}(N_2), d_{\text{UB}}]\\
= 0 & \text{ if }~ d_1 \in [0,d_{\text{UB}}] \cap \{d_1^{-}(N_2), d_1^{+}(N_2)\} \\
> 0 & \text{ if }~ d_1 \in [0,d_{\text{UB}}] \cap (d_1^{-}(N_2), d_1^{+}(N_2)). \\
\end{cases}
\end{align} \label{lemma: f-omega}
\end{lemlist}
\end{lemma}

\begin{proof}
We first prove \Cref{lemma: g-omega}, after that the result of \Cref{lemma: f-omega} is easily obtained.

It holds that
\begin{align}
\frac{\partial g(d',N',N'';\rho)}{\partial \rho} =  \frac{\sqrt{1+\frac{4\rho}{N''} B(d',N';\rho)}  - 1 +\frac{2\rho}{N''} (d'N'\sigma- \varphi D )}{2\sigma \rho^2 \sqrt{1 + \frac{4\rho}{N''} B(d',N';\rho)}},
\end{align}
so
\begin{align} \label{appeq: g-inequality}
\frac{\partial g(d',N',N'';\rho)}{\partial \rho} \geq 0 \Leftrightarrow \sqrt{(N'')^2+4\rho N'' B(d',N';\rho)} \geq N'' + 2\rho(\varphi D - d'N' \sigma).
\end{align}
We distinguish 2 cases:
\begin{itemize}
\item $\varphi D \geq d'N'\sigma$. In this case, squaring \eqref{appeq: g-inequality} on both sides and simplifying results in
\begin{align} \label{appeq: condition-g-w}
- \sigma^2 N' (N' + N'')d'^2 + 2\varphi D N' \sigma d' + \big(\frac{N''}{T} - 1\big)\varphi^2 D^2  \geq 0,
\end{align}
which is a condition independent of $\rho$. If $N' + N'' < T$, this inequality has no roots for $d'$, and \eqref{appeq: condition-g-w} holds for all $d' \in [0, \frac{\varphi D}{N'\sigma}]$. If $N' + N'' \geq T$ one can verify that $d_1 = d_1^{-}(N'')$ and $d_1 = d_2^{+}(N'')$ are the roots of this concave parabola if they are finite. The smaller root, $d_1^{-}(N'')$, is finite if and only if $N'' \leq T$. The larger root, $d_1^{+}(N'')$ is finite if and only if $N' \leq T$. 
\item $\varphi D < d'N'\sigma$. In this case, $B(d',N';\rho) >0 $ only if $N' > T$. In this case, the delivered dose exceeds the dose that is used to set the BED tolerance, which is only possible if the number of fractions $N'$ is strictly larger than the reference number of fractions $T$. Condition \eqref{appeq: g-inequality} clearly holds, so $g(d',N',N'';\rho)$ is increasing in $\rho$.
    Using the fact that $N' > T$ it is easily shown that $d_1^{-} < \frac{\varphi D}{ \sigma N'} < d'$. Additionally, it can be shown that $d_{\text{UB}} < d_1^{+}(N'')$. Hence this case satisfies \eqref{eq: g-omega}. Putting all of the above together yields the required result for $g$, i.e., \Cref{lemma: g-omega}.
\end{itemize}

The partial derivative of $f$ w.r.t. $\rho$ is given by
\begin{align}
\frac{\partial f(d_1,N_2;\rho,\tau)}{\partial \rho} = \frac{\partial g(d_1,N_1,N_2;\rho)}{\partial \rho} \Big(N_2 + 2\tau N_2 g(d_1,N_1,N_2;\rho) \Big).
\end{align}
Hence, the sign of the partial derivative of $f$ w.r.t. $\rho$ is equal to the sign of the partial derivative of $g$ w.r.t. $\rho$. The result of \Cref{lemma: f-omega} immediately follows.
\end{proof}

For given $(\rho,\tau)$ such that $\tau \neq \sigma \rho$, define the \emph{twin point} of $d_1 \in W(\rho,\tau)$ as
\begin{align} \label{appeq: t}
t(d_1;\rho,\tau) := \frac{\big(N_1 - N_2^{\ast}(\rho,\tau)\big) d_1 + 2N_2^{\ast}(\rho,\tau) g\big(d_1,N_1,N_2^{\ast}(\rho,\tau);\rho\big)}{N_1 + N_2^{\ast}(\rho,\tau)},
\end{align}
where
\begin{align}
\hspace*{-0.3cm} W(\rho,\tau) := \Big [ \max \{0, t(g(0,0,N_1;\rho);\rho,\tau) \},~\min\{t(0;\rho,\tau), g(0,0,N_1;\rho) \} \Big] \backslash \{g(0,0,N_1+N_2;\rho)\}.
\end{align}
\Cref{fig: twin} illustrates the relation between $d_1$ and $t(d_1;\rho,\tau)$. Set $W$ can be interpreted as the points $d_1$ for which there exists another point the graph of $f$ that has the same value, we refer to such points as twin points. The following lemma proves that for fixed $(\rho,\tau)$ any $d_1$ in the set $W(\rho,\tau)$ has a twin point $t(d_1;\rho,\tau)$ that is also in the set $W(\rho,\tau)$, and their objective values are equal.

\begin{figure}
\centering
\begin{subfigure}{0.45\textwidth}
\includegraphics[scale=0.9]{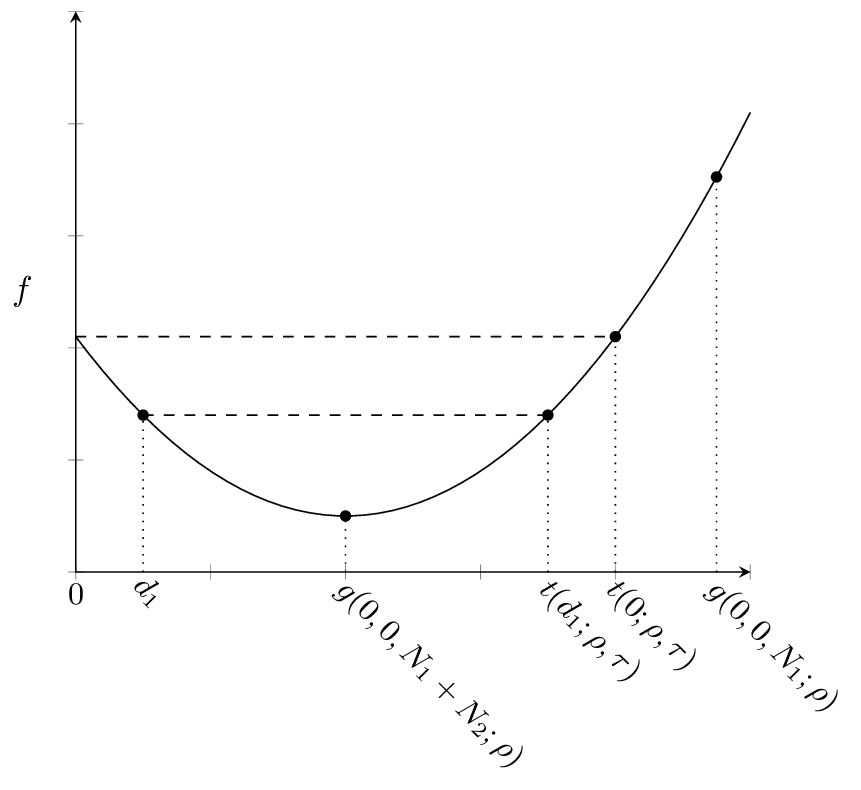}%
\caption{\small Case $t(0;\rho,\tau) < g(0,0,N_1;\rho)$.}
\end{subfigure}
\qquad
\begin{subfigure}{0.45\textwidth}
\includegraphics[scale=0.9]{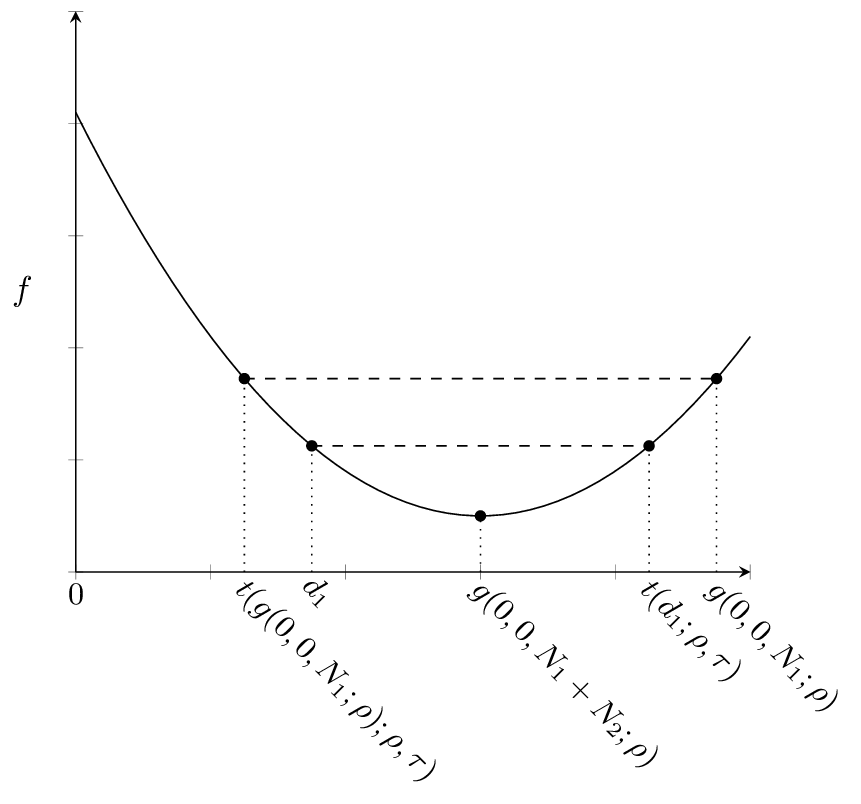}%
\caption{\small Case $t(0;\rho,\tau) > g(0,0,N_1;\rho)$.}
\end{subfigure}
\caption{\small Illustration of $d_1$ and $t(d_1;\rho,\tau)$, for convex $f$. \label{fig: twin}}
\end{figure}

\begin{lemma}\label{lemma: twins}
Let $(\rho,\tau) \in Z$ such that $\tau \neq \sigma\rho$, let $N_2 = N_2^{\ast}(\rho,\tau)$ and let $d_1 \in [0,g(0,0,N_1;\rho)]$. The equation
\begin{align} \label{appeq: equation-t}
f(d_1,N_2;\rho,\tau) = f(d_1',N_2;\rho,\tau)
\end{align}
has a solution $d_1' \in [0,g(0,0,N_1;\rho)]$ unequal to $d_1$ if and only if $d_1 \in W(\rho,\tau)$. In that case, there is a unique solution $d_1' = t(d_1;\rho,\tau) \in W(\rho,\tau)$, and it holds that $d_1 = t(t(d_1;\rho,\tau);\rho,\tau)$.
\end{lemma}
\begin{proof}
Let $(\rho,\tau)$ such that $\tau \neq \sigma \rho$, let $N_2 = N_2^{\ast}(\rho,\tau)$. We first show that if there exists a solution $d_1'$ to \eqref{appeq: equation-t} unequal to $d_1$, then this solution is $d_1' = t(d_1;\rho,\tau)$ and that $d_1 = t(t(d_1;\rho,\tau);\rho,\tau)$. Subsequently, we show that this solution exists only if $d_1 \in W(\tau,\rho)$ and that in that case also $t(d_1;\rho,\tau) \in W(\rho,\tau)$.\\

\noindent \emph{First part:}\\
Let $L, Q \in \mathbb{R}_+$ denote the linear and quadratic contribution of $d_1$ to $f$, i.e.,
\begin{subequations} \label{appeq: LQ}
\begin{align}
L(d_1,N_2,\rho) &:= N_1d_1 + N_2g(d_1,N_1,N_2;\rho)  \label{appeq: L}\\
Q(d_1,N_2,\rho) &:= N_1d_1^2 + N_2 g(d_1,N_1,N_2;\rho)^2.\label{appeq: Q}
\end{align}
\end{subequations}
We show that for any $d_1 \in [0,g(0,0,N_1;\rho)]$ unequal to $g(0,0,N_1+N_2,\rho)$ there is a unique solution $d_1'$ such that
\begin{subequations} \label{appeq: LQ-condition}
\begin{align}
L(d_1,N_2;\rho) &= L(d_1',N_2;\rho) \\
Q(d_1,N_2;\rho) &= Q(d_1',N_2;\rho).
\end{align}
\end{subequations}
Any such $d_1'$ has exactly the same objective value as $d_1$. Plug \eqref{appeq: LQ} in \eqref{appeq: LQ-condition}, rewrite the first equation to eliminate $g$, and plug this in the second equation to get
\begin{align}
d_1' &~= \frac{L(d_1,N_2;\rho) \pm \sqrt{\frac{N_2}{N_1} \big((N_1 + N_2) Q(d_1,N_2;\rho) - L^2(d_1,N_2;\rho)\big)}}{N_1 + N_2} \nonumber \\
&~= \frac{N_1d_1 + N_2 g(d_1,N_1,N_2;\rho) \pm N_2(d_1 - g(d_1,N_1,N_2;\rho))}{N_1 + N_2}. \label{appeq: d1-plusminus}
\end{align}
The `+' solution to \eqref{appeq: d1-plusminus} returns $d_1' = d_1$, and the `--' solution returns
\begin{align}
d_1' = \frac{(N_1-N_2)d_1 + 2N_2 g(d_1,N_1,N_2,\rho)}{N_1 + N_2},
\end{align}
and we denote this solution by $t(d_1;\rho,\tau)$. By construction, it holds that $d_1 = t(t(d_1;\rho,\tau);\rho,\tau)$. Because $\tau \neq \sigma \rho$, function $f(d_1,N_2;\rho,\tau)$ is a strictly convex or concave function according to \Cref{lemma: convexity-f}, so $f(d_1,N_2;\rho,\tau) = z$ for some constant $z \in \mathbb{R}$ has either 0, 1 or 2 solutions. In particular, $d_1 = t(d_1;\rho,\tau)$ if and only if $d_1$ equals minimizer $g(0,0,N_1+N_2;\rho)$. Hence if there exists a solution $d_1'$ to \eqref{appeq: equation-t} unequal to $d_1$, then this solution is $d_1' = t(d_1;\rho,\tau)$.\\

\noindent \emph{Second part:}
\begin{itemize}
\item Suppose $d_1 \notin W(\rho,\tau)$. We distinguish three cases. Case (i): `$d_1 = g(0,0,N_1+N_2;\rho)$'. Because this is the unique minimizer of $f$, there does not exist a $d_1'$ with equal objective value. Case (ii). `$d_1 > \min \{t(0;\rho,\tau),~g(0,0,N_1;\rho) \}$'. Because $d_1 \in [0,g(0,0,N_1;\rho)]$, this implies $d_1 > t(0,\rho,\tau)$. As shown in the first part of the proof, it holds that $d_1 = t(t(d_1;\rho,\tau);\rho,\tau)$. Hence, $d_1 > t(0;\rho,\tau)$ is equivalent to $t(t(d_1;\rho,\tau),\rho,\tau) >t(0;\rho,\tau)$. Because $t(d_1;\rho,\tau)$ is decreasing in $d_1$, this implies $t(d_1;\rho,\tau) < 0$, so according to \eqref{appeq: f} it holds that $f(t(d_1;\rho,\tau), N_2 ; \rho,\tau) = -\infty$ and we have a contradiction. Case (iii): `$d_1 < \max \{0,~ t(g(0,0,N_1;\rho);\rho,\tau) \}$'. Similar to case (ii), one can show that $f(t(d_1;\rho,\tau), N_2 ; \rho,\tau) = -\infty$.
\item Suppose $d_1 \in W(\rho,\tau)$. From \eqref{appeq: d1-plusminus} one can see that, because the term $d_1 - g(d_1,N_1,N_2,\rho)$ is increasing in $d_1$, the function $t(d_1;\rho,\tau)$ is decreasing in $d_1$. Consequently,
\begin{align}
d_1 \leq \min\{t(0;\rho,\tau), g(0,0,N_1;\rho) \} \Leftrightarrow d_1' \geq \max \{0, t(g(0,0,N_1;\rho);\rho,\tau) \}.
\end{align}
Furthermore, using the same argument,
\begin{align}
d_1 \geq \max \{0, t(g(0,0,N_1;\rho);\rho,\tau) \} \Leftrightarrow d_1' \leq \min\{t(0;\rho,\tau), g(0,0,N_1;\rho) \}.
\end{align}
Therefore, it holds that $d_1' \in W(\rho,\tau)$.
\end{itemize}
\end{proof}

\noindent In the following lemma, let $I(\cdot|S)$ denote the indicator function for a set $S$:
\begin{align}
I(x | S) =
\begin{cases}
1 & \text{ if } x\in S \\
0 & \text{ otherwise. }
\end{cases}
\end{align}
\begin{lemma}\label{lemma: i=2-reform}
For given $q \in \mathbb{R}_+$ and given $d_1 \in [0, d_{\text{UB}}]$,
\begin{align} \label{appeqB: inexact-con2}
q \leq f(d_1, N_2^{\ast}(d_1; \hat{\rho},\hat{\tau});\rho,\tau), ~~\forall (\rho,\tau,\hat{\rho},\hat{\tau}) \in U_L^{\text{int}},
\end{align}
holds if $\tau \leq p(d_1)$, with
\begin{align} \label{appeq: p}
\begin{aligned}
p(d_1) =& \sum_{\eta \in \{N_2^{\text{min}},\dotsc,N_2^{\text{max}}\}} \max \big \{ f(d_1,\eta;\rho^{\text{int}}_L,\tau_L),~f(d_1,\eta-1;\rho^{\text{int}}_U,\tau_L)\big\} I(d_1 | S_{\eta}) \\
        & + f(d_1,N_2^{\text{min}};\rho^{\text{int}}_L,\tau_L) I(d_1 | S_{\text{min}}) + f(d_1,N_2^{\text{max}};\rho^{\text{int}}_U,\tau_L) I(d_1 | S_{\text{max}}),
\end{aligned}
\end{align}
where sets $S_{\text{min}}$, $S_{\text{max}}$ and $S_{\eta}$ are defined in \eqref{appeq: S-min}, \eqref{appeq: S-max} and \eqref{appeq: S-eta}, respectively, and $\rho^{\text{int}}_L$, $\rho^{\text{int}}_U$ are defined in \eqref{appeq: omega-int}.
\end{lemma}

\begin{proof}
By definition of $N_2^{\ast}(d_1; \hat{\rho},\hat{\tau})$ and $U_L^{\text{int}}$, it holds that
\begin{align}
q \leq f(d_1, N_2^{\ast}(d_1; \hat{\rho},\hat{\tau});\rho,\tau), ~~\forall (\rho,\tau,\hat{\rho},\hat{\tau}) \in U_L^{\text{int}},
\end{align}
is equivalent to
\begin{align} \label{appeq: i=2-equation}
q \leq \max_{\tilde{\eta} \in \{N_2^{\text{min}},\dotsc,N_2^{\text{max}}\}} \min\{ f(d_1, \tilde{\eta} ; \hat{\rho}_L,\hat{\tau}_L), f(d_1, \tilde{\eta} ; \hat{\rho}_U,\hat{\tau}_L)\}, ~~\forall (\hat{\rho},\hat{\tau}) \in Z_{\text{ID}}^{\text{int}} \cap \{(\hat{\rho},\hat{\tau}) : \hat{\tau} \leq \tau_L + r^{\tau} \},
\end{align}
and because function $f$ is increasing in $\tau$, we need to consider only those observations $(\hat{\rho},\hat{\tau})$ with $\hat{\tau}_L = \tau_L$. For the first part of the proof, we fix the observation $(\hat{\rho},\hat{\tau})$, plug in $\hat{\tau}_L = \tau_L$, and rewrite \eqref{appeq: i=2-equation} for this fixed observation.

Because $(\hat{\rho},\hat{\tau}) \in Z_{\text{ID}}^{\text{int}}$, it holds that $\sigma \hat{\rho}_L < \tau_L < \sigma \hat{\rho}_U$. Hence, by \Cref{lemma: convexity-f}, function $f(d_1,\eta, \hat{\rho}_L,\tau_L)$ is convex and $f(d_1,\eta, \hat{\rho}_U,\tau_L)$ is concave in $d_1$ for any $\eta \in \mathbb{N}_+$. We make use of results of \Cref{lemma: intersection}. Define
\begin{subequations}
\begin{align}
E^{-} & = \{ \eta : N_1 + \eta \geq T, \eta \leq T \} \cap \{N_2^{\text{min}},\dotsc,N_2^{\text{max}}\}\\
E^{+} & = \{ \eta : N_1 + \eta \geq T, N_1 \leq T \} \cap \{N_2^{\text{min}},\dotsc,N_2^{\text{max}}\},
\end{align}
\end{subequations}
and let $\eta_{\text{min}}^{-}$, $\eta_{\text{max}}^{-}$, $\eta_{\text{min}}^{+}$ and $\eta_{\text{max}}^{+}$, denote the smallest and largest elements of $E^{-}$ and $E^{+}$, respectively. If $\eta \in E^{-}$ respectively $\eta \in E^{+}$, then, according to \Cref{lemma: intersection1}, $d_1=d_1^{-}(\eta)$ respectively $d_1=d_1^{+}(\eta)$ is a nonnegative real root of
\begin{align*}
f(d_1,\eta; \hat{\rho}_L,\tau_L) = f(d_1,\eta; \hat{\rho}_U,\tau_L),
\end{align*}
and the corresponding objective value equals $K$. From \Cref{lemma: intersection2} we know that
\begin{align}
d_1^{-}(N_2^{\text{max}}) < \dotsc < d_1^{-}(N_2^{\text{min}}) \leq d_1^{+}(N_2^{\text{min}}) \leq \dotsc \leq d_1^{+}(N_2^{\text{max}}).
\end{align}
We use this to split the domain $[0,d_{\text{UB}}]$ as follows:
\begin{subequations}
\begin{align}
& S_{\text{min}} =
\begin{cases}
(d_1^{-}(\eta_{\text{min}}^{-}), d_1^{+}(\eta_{\text{min}}^{+})) &\text{ if } E^{-} \neq \emptyset, E^{+} \neq \emptyset \\
[0, d_1^{+}(\eta_{\text{min}}^{+})) & \text{ if } E^{-} = \emptyset, E^{+} \neq \emptyset \\
(d_1^{-}(\eta_{\text{min}}^{-}), d_{\text{UB}}] & \text{ if } E^{-} \neq \emptyset, E^{+} = \emptyset \\
\emptyset &\text{ if }  N_1 + N_2^{\text{max}} < T \\
[0, d_{\text{UB}}] &\text{ otherwise} \\
\end{cases} \label{appeq: S-min}\\
& S_{\eta}^{-} :=
\begin{cases}
[d_1^{-}(\eta),d_1^{-}(\eta-1)] & \text{ if } \eta_{\text{min}}^{-} \leq \eta-1 < \eta \leq \eta_{\text{max}}^{-} \\
\emptyset & \text{ otherwise }\\
\end{cases} ~~\forall \eta \in \{N_2^{\text{min}}+1,\dotsc,N_2^{\text{max}}\}\\
& S_{\eta}^{+} :=
\begin{cases}
[d_1^{+}(\eta-1),d_1^{+}(\eta)] & \text{ if } \eta_{\text{min}}^{+} \leq \eta-1 < \eta \leq \eta_{\text{max}}^{+} \\
\emptyset & \text{ otherwise }\\
\end{cases} ~~\forall \eta \in \{N_2^{\text{min}}+1,\dotsc,N_2^{\text{max}}\}\\
& S_{\text{max}} =
\begin{cases}
[0, d_1^{-}(\eta_{\text{max}}^{-})) \cup (d_1^{+}(\eta_{\text{max}}^{+}), d_{\text{UB}}] &\text{ if } E^{-} \neq \emptyset, E^{+} \neq \emptyset \\
(d_1^{+}(\eta_{\text{max}}^{+}), d_{\text{UB}}] &\text{ if } E^{-} = \emptyset, E^{+} \neq \emptyset \\
[0, d_1^{-}(\eta_{\text{max}}^{-})) & \text{ if } E^{-} \neq \emptyset, E^{+} = \emptyset \\
[0, d_{\text{UB}}] &\text{ if }  N_1 + N_2^{\text{max}} < T \\
\emptyset &\text{ otherwise.} \\
\end{cases}\label{appeq: S-max}
\end{align}
\end{subequations}
We will reformulate \eqref{appeq: i=2-equation} on each interval (set) separately, assuming it is nonempty.
\begin{enumerate}
\item ``$S_{\text{min}}$'': If $d_1 \in S_{\text{min}}$, then $f(d_1,\eta; \hat{\rho}_L,\tau_L) < f(d_1,\eta; \hat{\rho}_U,\tau_L)$ for all $\eta \in \{N_2^{\text{min}},\dotsc,N_2^{\text{max}}\}$ according to \Cref{lemma: f-omega}, so it is optimal to deliver $N_2^{\text{min}}$ fractions. Hence, on this interval \eqref{appeq: i=2-equation} is equivalent to
\begin{align}
q \leq f(d_1, N_2^{\text{min}}; \hat{\rho}_L,\tau_L).
\end{align}
\item ``$S_{\eta}^{-}$'': From \Cref{lemma: intersection1} we know that $f(d_1,\eta, \hat{\rho}_L,\tau_L) = f(d_1,\eta, \hat{\rho}_U,\tau_L)$ if $d_1 = d_1^{-}(\eta)$ or $d_1 = d_1^{+}(\eta)$. In this case, the objective value equals $K$. Furthermore, function $f(d_1,\eta, \hat{\rho}_L,\tau_L)$ is convex and $f(d_1,\eta, \hat{\rho}_U,\tau_L)$ is concave in $d_1$. Consider the interval $[d_1^{-}(\eta),d_1^{-}(\eta-1)]$. It holds that
\begin{subequations} \label{appeq: dEta-inequalities}
\begin{align}
f(d_1,\eta ; \hat{\rho}_L,\tau_L) \leq K \leq f(d_1,\eta-1 ; \hat{\rho}_L,\tau_L)\\
f(d_1,\eta - 1 ; \hat{\rho}_U,\tau_L) \leq K \leq f(d_1,\eta ; \hat{\rho}_U,\tau_L).
\end{align}
\end{subequations}
This implies that if $d_1 \in [d_1^{-}(\eta),d_1^{-}(\eta-1)]$, it is optimal to deliver either $\eta$ or $\eta-1$ fractions. If we deliver $\eta$ fractions, the restricting worst-case scenario is $(\hat{\rho}_L,\tau_L)$ and the value $f$ is above $K$ for the scenario $(\hat{\rho}_U,\tau_L)$. If we deliver $\eta' > \eta$ fractions, the value for the scenario $(\hat{\rho}_L,\tau_L)$ decreases, while the value for the scenario $(\hat{\rho}_U,\tau_L)$ increases even further. Hence, delivering $\eta' > \eta$ fractions cannot be optimal. Similarly, delivering less than $\eta - 1$ fractions cannot be optimal. Therefore, if $d_1 \in [d_1^{-}(\eta),d_1^{-}(\eta-1)]$ it is optimal to deliver either $\eta$ or $\eta-1$ fractions. This implies that on the interval $S_{\eta}^{-}$ constraint \eqref{appeq: i=2-equation} is equivalent to
\begin{align} \label{appeq: constraint-eta}
q \leq \max\{ f(d_1, \eta ; \hat{\rho}_L,\tau_L),~ f(d_1, \eta-1 ; \hat{\rho}_U,\tau_L) \}.
\end{align}
Note that this result does not depend on the values $\hat{\rho}_L$ and $\hat{\rho}_U$, we only use that $\hat{\rho}_L < \frac{\tau_L}{\sigma} < \hat{\rho}_U$.

\item ``$S_{\eta}^{+}$'': Similar to the case for $S_{\eta}^{+}$, one can show that for $d_1 \in S_{\eta}^{+}$ constraint \eqref{appeq: i=2-equation} is equivalent to \eqref{appeq: constraint-eta}.

\item ``$S_{\text{max}}$'': If $d_1 \in S_{\text{max}}$, then $f(d_1,\eta; \hat{\rho}_L,\tau_L) > f(d_1,\eta; \hat{\rho}_U,\tau_L)$ for all $\eta \in \{N_2^{\text{min}},\dotsc,N_2^{\text{max}}\}$ according to \Cref{lemma: f-omega}, so it is optimal to deliver $N_2^{\text{max}}$ fractions. Hence, on this interval \eqref{appeq: i=2-equation} is equivalent to
\begin{align}
q \leq f(d_1, N_2^{\text{max}}; \hat{\rho}_U,\tau_L).
\end{align}
\end{enumerate}
For sets $S_{\eta}^{-}$ and $S_{\eta}^{+}$ the reformulation is the same. Therefore, define
\begin{align} \label{appeq: S-eta}
S_{\eta} = S_{\eta}^{-} \cup S_{\eta}^{+}.
\end{align}
Putting everything together, for $d_1 \in [0, d_{\text{UB}}]$ the constraint \eqref{appeq: i=2-equation} is equivalent to
\begin{align} \label{appeq: piecewise}
\begin{aligned}
q \leq & \sum_{\eta \in \{N_2^{\text{min}},\dotsc,N_2^{\text{max}}\}}  \max\{ f(d_1, \eta ; \hat{\rho}_L,\tau_L),~ f(d_1, \eta-1 ; \hat{\rho}_U,\tau_L) \} I(d_1 | S_{\eta}) \\
           &+  f(d_1, N_2^{\text{min}} ; \hat{\rho}_L,\tau_L)I(d_1 | S_{\text{min}}) + f(d_1,N_2^{\text{max}} ; \hat{\rho}_U,\tau_L)I(d_1 | S_{\text{max}}),~~\forall (\hat{\rho},\hat{\tau}) \in Z_{\text{ID}}^{\text{int}} \cap \{(\hat{\rho},\hat{\tau}) : \hat{\tau} \leq \tau_L + r^{\tau} \}.
\end{aligned}
\end{align}
In order to find a tractable conservative robust counterpart of \eqref{appeq: piecewise}, denote
\begin{subequations} \label{appeq: omega-int}
\begin{align}
\rho^{\text{int}}_L &= \max\{\rho_L,\frac{\tau_L}{\sigma} - 2r^{\rho} \} \\
\rho^{\text{int}}_U &= \min\{ \rho_U, \frac{\tau_L}{\sigma} + 2r^{\rho}\},
\end{align}
\end{subequations}
and note that $\rho^{\text{int}}_L \leq \hat{\rho}_L < \frac{\tau_L}{\sigma} < \hat{\rho}_U \leq \rho^{\text{int}}_U$. Only if $d_1 \in S_{\eta}$, the robust counterpart is conservative. By \Cref{lemma: f-omega}, it holds that function $f$ is strictly decreasing, constant or strictly increasing in $\rho$ for fixed $d_1$, so
\begin{align*}
f(d_1, \eta ; \hat{\rho}_L,\tau_L) \geq \min\{f(d_1, \eta ; \rho^{\text{int}}_L,\tau_L),f(d_1, \eta ; \frac{\tau_L}{\sigma},\tau_L) \} = \min\{f(d_1, \eta ; \rho^{\text{int}}_L,\tau_L),K \}  = f(d_1, \eta ; \rho^{\text{int}}_L,\tau_L),
\end{align*}
where the second equality follows from \eqref{appeq: dEta-inequalities}. A similar result holds for $f(d_1, \eta -1 ; \hat{\rho}_U,\tau_L)$. Furthermore, as shown before, $f$ is increasing in $\rho$ on $S_{\text{min}}$ and decreasing in $\rho$ on $S_{\text{max}}$. Therefore, a conservative approximation of \eqref{appeqB: inexact-con2} is given by
\begin{align}
\begin{aligned}
q \leq & \sum_{\eta \in \{N_2^{\text{min}},\dotsc,N_2^{\text{max}}\}}  \max\{ f(d_1, \eta ; \rho^{\text{int}}_L,\tau_L),~ f(d_1, \eta-1 ; \rho^{\text{int}}_U,\tau_L) \} I(d_1 | S_{\eta}) \\
           &+  f(d_1, N_2^{\text{min}} ; \rho^{\text{int}}_L,\tau_L)I(d_1 | S_{\text{min}}) + f(d_1,N_2^{\text{max}} ; \rho^{\text{int}}_U,\tau_L)I(d_1 | S_{\text{max}}),
\end{aligned}
\end{align}
and the RHS is $p(d_1)$.
\end{proof}
\noindent Function $p(d_1)$ is a piece-wise function. On intervals defined by $S_{\text{min}}$ and $S_{\text{max}}$ it is convex and concave, respectively. On any interval $S_{\eta}$ function $p(d_1)$ is the maximum of a concave and convex function.

\end{document}